%% file: main.tex
\documentclass{lmcs}


\ACMCCS{Theory of computation~Logic and verification;
Theory of computation~Automata over infinite objects}

\input{setup/packages.tex}

\input{setup/macros.tex}

\begin{document}

\title[Doubly Fair Parity Games]{Doubly Fair Parity Games}

\author[D.~Hausmann]{Daniel Hausmann\lmcsorcid{0000-0002-0935-8602}}[a]
\author[N.~Piterman]{Nir Piterman\lmcsorcid{0000-0002-8242-5357}}[b]
\author[I.~Sa\u{g}lam]{Irmak Sa\u{g}lam\lmcsorcid{0000-0002-4757-1631}}[c]
\author[A.-K.~Schmuck]{Anne-Kathrin Schmuck\lmcsorcid{0000-0003-2801-639X}}[c]

\address{University of Liverpool, Liverpool, United Kingdom}
\email{hausmann@liverpool.ac.uk}

\address{University of Gothenburg and Chalmers University of Technology, Gothenburg, Sweden}
\email{piterman@chalmers.se}

\address{Max Planck Institute for Software Systems (MPI-SWS), Kaiserslautern, Germany}
\email{\{isaglam, akschmuck\}@mpi-sws.org}

\thanks{Daniel Hausmann was supported by the ERC Consolidator grant D-SynMA (No.~772459) and by the EPSRC through grant EP/Z003121/1.}
\thanks{Nir Piterman was supported by the ERC Consolidator grant D-SynMA (No.~772459), Swedish research council (VR) project (No.~2020-04963), and the Wallenberg AI, Autonomous Systems and Software Program (WASP) funded by the Knut and Alice Wallenberg Foundation.}
\thanks{Irmak Sa\u{g}lam was supported by the DFG project SCHM 3541/1-1.}
\thanks{Anne-Kathrin Schmuck was supported by the DFG projects SCHM 3541/1-1 and 389792660, as part of TRR 248--CPEC.}

\begin{abstract}
We consider two-player games over finite graphs in which both players are restricted by fairness constraints on their moves. Given a two-player game graph $G=(V,E)$ and a set of fair moves $E_f\subseteq E$, a player is said to play \emph{fair} in $G$ if they choose an edge $e\in E_f$ infinitely often whenever the source node of $e$ is visited infinitely often. Otherwise, they play \emph{unfair}. We equip such games with two $\omega$-regular winning conditions $\alpha$ and $\beta$ deciding the winner of mutually fair and mutually unfair plays, respectively. Whenever one player plays fair and the other plays unfair, the fairly playing player wins the game. The resulting games are called \emph{fair $\alpha/\beta$ games}.

We formalize fair $\alpha/\beta$ games and show that they are determined. For fair parity/parity games, i.e., fair $\alpha/\beta$ games where $\alpha$ and $\beta$ are each given by a parity condition over $G$, we provide a polynomial reduction to normal parity games via a gadget construction inspired by the reduction of stochastic parity games to parity games. We further give a direct \emph{symbolic fixpoint algorithm} to solve fair parity/parity games. On a conceptual level, we illustrate the translation between the gadget-based reduction and the direct symbolic algorithm, uncovering the underlying similarities of solution algorithms for fair and stochastic parity games, as well as for the recently considered class of fair games in which only one player is restricted by fair moves.
\end{abstract}

\maketitle

\input{intro.tex}
\input{preliminaries.tex}
\input{fairgames.tex}

\input{strategiesinfairgames.tex}

\input{reductiontonormalgames.tex}

\input{fixpointcharacterization.tex}

\input{conclusion.tex}

\bibliographystyle{alphaurl}
\bibliography{generic}

\end{document}

%% file: setup/packages.tex
\usepackage{hyperref}
\usepackage{breakurl}             
\usepackage{underscore}           
\usepackage{caption}
\usepackage{bm}
\usepackage{tikz}
 \usetikzlibrary{trees}
 \usetikzlibrary{shapes}
 \usetikzlibrary{fit}
 \usetikzlibrary{shadows}
 \usetikzlibrary{backgrounds}
 \usetikzlibrary{arrows,automata,calc}
 \usetikzlibrary{positioning}
\usetikzlibrary{calc}
\tikzset{
   n/.style= {circle,fill,inner sep=1.5pt,node distance=2cm}
  ,acc/.style={circle,draw,inner sep=3pt,node distance=2cm}
  ,phantom/.style={circle},
  ,arr/.style={->, >=stealth, semithick, shorten <= 3pt, shorten >= 3pt}
}
\usetikzlibrary{arrows.meta, positioning,patterns,snakes}

\usepackage{mathtools,amsmath,amssymb,amsfonts}

\usepackage[mathlines]{lineno}

\usepackage{thm-restate}
\usepackage{xcolor}
\usepackage{soul}
\usepackage[utf8]{inputenc}
\usepackage{amsmath,amssymb}
\usepackage{xspace}
\usepackage{stmaryrd} 
\usepackage{enumitem}
\usepackage{algorithmicx}
    \usepackage[ruled]{algorithm}
\usepackage[noend]{algpseudocode}
\usepackage{graphicx}
\usepackage{csquotes}

\usepackage[author=anonymous,margin,footnote,draft]{fixme}
\FXRegisterAuthor{dh}{adh}{DH}
\FXRegisterAuthor{np}{anp}{NP}

\usepackage[capitalise]{cleveref}
\crefname{definition}{Def.}{Def.}
\crefname{problem}{Prob.}{Prob.}
\crefname{algorithm}{Algo.}{Algo.}
\crefname{theorem}{Thm.}{Thm.}
\crefname{lemma}{Lem.}{Lem.}
\crefname{appendix}{App.}{App.}
\crefname{line}{line}{line}
\crefname{section}{Sec.}{Sec.}
\crefname{corollary}{Cor.}{Cor.}
\crefname{example}{Ex.}{Ex.}
\crefname{line}{Line}{Lines}
\crefname{item}{Item}{Items}
\crefname{subsection}{Subsec.}{Subsec.}
\crefname{inpara}{Inline Item}{Inline Item}
\crefname{figure}{Fig.}{Fig.}
\crefname{table}{Table}{Tables}

\crefname{thm}{Thm.}{Thms.}
\crefname{lem}{Lem.}{Lems.}
\crefname{prop}{Prop.}{Props.}
\crefname{cor}{Cor.}{Cors.}
\crefname{defi}{Def.}{Defs.}
\crefname{exa}{Ex.}{Exs.}
\crefname{rem}{Rem.}{Rems.}
\crefname{obs}{Obs.}{Obs.}
\crefname{nota}{Notation}{Notations}
\crefname{invariant}{Invariant}{Invariants}

\Crefname{thm}{Theorem}{Theorems}
\Crefname{lem}{Lemma}{Lemmas}
\Crefname{prop}{Proposition}{Propositions}
\Crefname{cor}{Corollary}{Corollaries}
\Crefname{defi}{Definition}{Definitions}
\Crefname{exa}{Example}{Examples}
\Crefname{rem}{Remark}{Remarks}
\Crefname{obs}{Observation}{Observations}
\Crefname{nota}{Notation}{Notations}
\Crefname{invariant}{Invariant}{Invariants}

%% file: setup/macros.tex
\usepackage[
colorinlistoftodos, 
textwidth=\marginparwidth, 
textsize=scriptsize, 
]{todonotes}

\newcommand{\Win}{\mathsf{Win}}
\newcommand{\Inf}{\mathsf{Inf}}
\newcommand{\Cyc}{\mathsf{Cyc}}

\newcommand{\orphans}{O}

\newcommand{\Cpre}{\mathsf{Cpre}}
\newcommand{\Apre}{\mathsf{Apre}}
\newcommand{\Vn}{V^{\mathsf{n}}}
\newcommand{\Vfair}{V^{\mathsf{fair}}}
\newcommand{\fair}{\mathsf{fair}}
\newcommand{\play}{\mathsf{play}}

\newcommand{\out}{\mathsf{out}}
\newcommand{\vout}{{v^\mathsf{out}}}
\newcommand{\uout}{{u^\mathsf{out}}}
\newcommand{\even}{\mathsf{even}}
\newcommand{\odd}{\mathsf{odd}}
\newcommand{\mycount}{\mathtt{count}}
\newcommand{\mymod}{\mathsf{mod}}

\newcommand{\llambda}{{\lambda_{\downarrow}}}

\renewcommand{\Box}{\square}
\renewcommand{\Diamond}{\lozenge}

\makeatletter
\def\moverlay{\mathpalette\mov@rlay}
\def\mov@rlay#1#2{\leavevmode\vtop{%
   \baselineskip\z@skip \lineskiplimit-\maxdimen
   \ialign{\hfil$\m@th#1##$\hfil\cr#2\crcr}}}
\newcommand{\charfusion}[3][\mathord]{
    #1{\ifx#1\mathop\vphantom{#2}\fi
        \mathpalette\mov@rlay{#2\cr#3}
      }
    \ifx#1\mathop\expandafter\displaylimits\fi}
\makeatother

\newcommand{\dotcup}{\charfusion[\mathbin]{\cup}{\cdot}}

\newcommand{\myparagraph}[1]{\medskip\par\noindent\textbf{\textbf{#1}}\hspace{6pt}}

\makeatletter
\renewcommand\subparagraph{%
  \@startsection{subparagraph}{5}{\parindent}%
  {3.25ex \@plus 1ex \@minus .2ex}%
  {-1em}%
  {\normalfont\normalsize\bfseries}}
\makeatother

\definecolor{colored}{rgb}{1,0.5,0}

\makeatletter
\renewcommand{\subsubsection}{\@startsection{subsubsection}{3}{\z@}%
  {-18pt plus -4pt minus -4pt}%
  {6pt plus 2pt minus 1pt}%
  {\normalfont\normalsize\bfseries}}
\makeatother

\makeatletter
\newcommand{\subsubsubsection}{\@startsection{paragraph}{4}{\z@}%
  {-15pt plus -3pt minus -3pt}%
  {6pt plus 2pt minus 1pt}%
  {\normalfont\normalsize\bfseries}}
\makeatother

\theoremstyle{definition}
\newtheorem{invariant}[thm]{Invariant}

\newenvironment{inparaenum}[1][]{%
  \begin{enumerate*}[#1]%
}{%
  \end{enumerate*}%
}

%% file: intro.tex
\section{Introduction}\label{sec:intro}
\emph{Omega-regular games} are a popular abstract modelling formalism
for many core computational problems in the context of
correct-by-construction synthesis of reactive software or hardware.
This abstract view was initiated by the seminal work of Church
\cite{church1963application} and its independent solutions by B\"uchi
and Landweber and Rabin \cite{Rab69,BL69}.
Since then these ideas have been refined and extended for solving the
\emph{reactive synthesis problem} \cite{PR88,SF06,MeyerSL18}. 

However, before using such synthesis techniques, the reactive software
design problem at hand needs to be abstractly modelled as a two-player
game. In order for the subsequently synthesized software to be
`correct-by-construction' this game graph needs to reflect all possible
interactions between involved components in an abstract manner.
Building such a game graph with the `right' level of abstraction is a
known severe challenge, in particular, if the synthesized software is
interacting with existing components that already possess certain
behavior. Here, part of the modelling challenge amounts to finding  the
`right' power of both players in the resulting abstract game to ensure
that winning strategies do not fail to exist due to an unnecessarily
conservative overapproximation of modeling uncertainty (or the dual
problem due to underapproximation).

In this context, \emph{fairness} has been adopted as a notion to abstractly model known characteristics of the involved components in a very concise manner. \emph{Fairness assumptions} have been used in model checking \cite{AminofBK04} and scheduler synthesis for the classical AMBA arbiter \cite{NirGR1} or shared resource management \cite{CAFMR13}. 
Notably, fairness assumptions have also gained attention in cyber-physical system design \cite{thistle1998control,NOL17,MajumdarMallikSchmuckSoudjani24} and robot motion planning \cite{DIRS18,AGR20}. In all these applications, fairness is used as an \emph{assumption} that the synthesized (or verified) component can rely on. In particular, if these assumptions are modelled by \emph{transition fairness} over a two-player game arena\footnote{Whenever we interpret players in a one-sided manner as environment and system, we choose the environment player as the $\forall$-player, as we need to take all possible environment moves into account. Similarly, the system is the $\exists$-player in this scenario.} $(V_\forall,V_\exists,E)$ -- i.e., by a set of fair \emph{environment} moves $E_f\subseteq E$ (i.e., with $V_\forall$ as their domain) that need to be taken infinitely often if the source node is seen infinetely often along a play -- the resulting synthesis games can be solved efficiently \cite{banerjee2022fast,sauglam2023solving}.  

While most existing work has only looked at fairness as an \emph{assumption} to weaken the opponent in the synthesis game, all mentioned applications also naturally allow for scenarios where multiple components with intrinsically fair behavior are interacting with each other in a non-trivial manner. 
For example, the ability of a concurrent process to eventually free a shared resource might depend on how fair re-allocation is implemented in other threads. 
On an abstract level, the formal reasoning about such scenarios requires to understand how intrinsic fairness influences the reactive decision making of dependent processes. 
Algorithmically, such synthesis questions can be addressed by allowing fairness restrictions on both players in a game. 
We refer to such games simply as \emph{fair games}.

\smallskip
\noindent\textbf{Motivating Example.}
In order to illustrate such fair games, we consider the problem of deploying a new (green) robot in a robotic warehouse environment where autonomous (red) robots are already operating. The warehouse has narrow passages between adjacent regions that only one robot can pass at a time. The green robot has an $\omega$-regular objective $\alpha$ that specifies desired sequences of visited regions in the workspace. As the (high-level) objective of the red robot is not known, we assume that it acts adversarial in general. Nevertheless, it is reasonable to assume that all robots have a tie-breaking mechanism implemented for obstacle avoidance, i.e., they must eventually move left or right if an obstacle blocks their way. The existence of such a tie-breaking mechanism in the red robots acts as a local fairness assumption on the environment for the green robot. The green robot's own obligation for tie-breaking acts as an additional local fairness specification.

Now consider the scenario where two robots are facing each other at a gate, as depicted in~\cref{fig:stf2}. 
 While both robots block the gate from one side, neither of them can move forward, but if the green robots moves left or the red robot moves right, the other robot can take the gate to reach region $A$.
  With the mentioned requirement for tie-breaking, none of the robots is allowed to block the gate forever and both eventually have to move to the side.

    \begin{figure}
    \centering
     \includegraphics[width=0.4\linewidth]{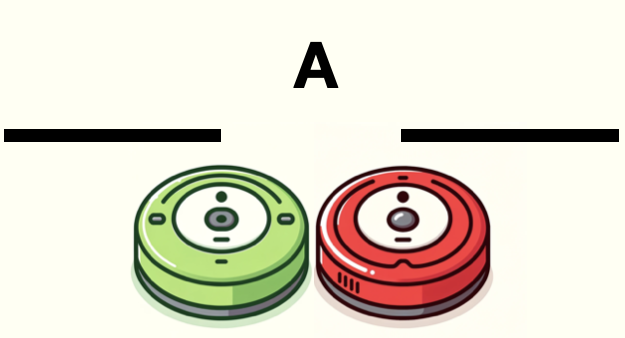}
    \caption{Deadlock caused by 
    fairness constraints of two robots facing a door.}
    \label{fig:stf2}
  \end{figure}

  Now let us assume that region $A$ is important for both robots, hence, both robots have an incentive to enter region $A$ first, to then move the game to an area preferable to them. However, the robot who breaks the tie first, (i.e., fulfills its fairness condition first) allows the \emph{other} robot to enter region $A$ first, which gives both robots the incentive to behave unfair. While it is very intuitive to make a player lose when she plays unfair and the other player plays fair, it is unclear who wins the game if both players play unfair.

  To resolve this issue, we can make the objectives of the robots completely adversarial by assigning one of the players (say, green) the winner in a play where both players play unfair. In the above example, this would give the red robot the incentive to break the tie first. 
  While this makes it harder for the red robot to spoil the objective of the green one, we might be interested in a more symmetric game which does not favor the green robot in all non-determined states of the graph. 
  We therefore consider a second $\omega$-regular objective $\beta$ that determines the winner of (mutually) unfair plays. This results in fair games $G=(A,\alpha,\beta)$ which are determined (as shown in~\cref{sec:fairgames}).

\smallskip
\noindent\textbf{Contribution.}
Motivated by the above mentioned examples where interactive decision making of two dependent processes is influenced by intrinsic fairness constraints imposed on their decisions, this paper studies \emph{fair games} $G=(A,\alpha,\beta)$ as their abstraction. In particular, we give solution algorithms for these games when both $\alpha$ and $\beta$ are parity conditions induced by two different priority functions over the node set. 
We call such games \emph{fair parity/parity games}.

Obviously, the previously discussed one-sided version of fair games, which we call \emph{$\forall$-fair games} (as only the $\forall$-player -- i.e., the environment -- is restricted by strong transition fairness), is a special case of fair games. Both enumerative \cite{sauglam2023solving} and symbolic solution algorithms \cite{banerjee2022fast} have recently been proposed for $\forall$-fair games, showing that strong transition fairness can be handled efficiently in both types of algorithms. This observation is closely related to a result for \emph{stochastic games}, i.e., two-player games with an \emph{additional} `half' player that takes all its moves uniformly at random. For
the purpose of qualitative analysis, such stochastic parity games have been shown in \cite{CJH03} to be reducible to (standard) parity games by the use of \enquote{gadgets} that turn stochastic nodes into a small combination of $\forall$- and $\exists$-player nodes. While it is known that stochastic games can be reduced to \emph{$\forall$-fair games} (and hence, fair games), it was not investigated how the different solution approaches compare. 
The main conceptual contribution of this paper is a unified understanding of all these solution approaches for the general class of fair games.

Concretely, our contribution is three-fold:\\
\begin{enumerate*}[(1)]
\item We formalize fair games as a generalization of $\forall$-fair games and qualitative  stochastic games in such a way that determinacy is preserved. Moreover, within this general framework, the players' behaviour in the presence of fairness violations by either party can be fine-tuned. Determinacy is a desirable property in synthesis, as it enables an agent to analyze the causes of its \emph{losing behaviour} by examining the opponent's counter-strategies. The ability to fine-tune the agents' behaviour in scenarios where one or neither party has an incentive to adhere to fairness assumptions constitutes a natural generalization and represents a novel contribution of this work.\\
 \item We show a reduction of fair parity/parity games to (standard) parity games, inspired by the gadget-based reduction of stochastic parity games to parity games in~\cite{CJH03}. This reduction enables the use of parity game solvers over the reduced game (in particular enumerative ones such as progress measure algorithms~\cite{DBLP:conf/lics/JurdzinskiL17}) and gives a gadget-based reduction of $\forall$-fair parity games to parity games as a corollary.\\
 \item We then show how our gadget construction can be used to define a \emph{symbolic fixpoint algorithm} to solve fair parity/parity games directly (without the need for a reduction). We show the direct symbolic algorithm for $\forall$-fair parity games in \cite{banerjee2022fast} coinciding with our algorithm for this particular subclass of fair games.
\end{enumerate*}

With this, we believe that this paper uncovers the underlying
similarities of solution algorithms for fair, $\forall$-fair and
stochastic parity games. Further, we show how these conceptual
similarities can be used to build both \emph{enumerative} and direct
\emph{symbolic} algorithms. This is of interest as both are known to
have complementary strengths, depending on how the synthesis instance is
provided, and this connection was, to the best of our knowledge, not
known before.

This paper extends a previous conference publication~\cite{DBLP:conf/fossacs/HausmannPSS24}. Besides providing full
proofs and additional examples, we also include a treatment of doubly fair games with B\"uchi objectives
(Section~\ref{subsection:reductionbuchi})
that we intend to serve as a gentle introduction to fair games and as a stepping stone towards full-fledged 
fair parity games.
On a technical level, we provide detailed proofs of the
correctness of all presented reductions from fair games to normal games,
and of the symbolic fixpoint algorithms.
This is particularly involved for the gadget reduction from fair parity/parity games to (normal) parity games (Theorem~\ref{thm:fairParityPlusParityToParity});
while this reduction is seemingly based on adding memory to previously known gadget constructions, resulting in a construction that looks similar to previous gadgets, the proof of this particular result is quite substantial. 
The proof involves a detailed case distinction of the interactive (fairness) behaviour of both players.
This reveals how the gadget construction correctly resolves all possible interactions of fairness constraints \emph{locally} without allowing players to 'cheat' by constructing a play that locally satisfies the mutual constraints but becomes (un)fair in the limit. Thereby, the proof provides deep insights into the structure of winning strategies in fair parity/parity games.

\smallskip
\noindent\textbf{Structure.} We start by recalling basic concepts and notation related to games in~\cref{section:preliminaries}.
In~\cref{sec:fairgames} we introduce fair $\alpha$/$\beta$ games (for $\omega$-regular objectives $\alpha$ and $\beta$) and show their determinacy.
\cref{section:reductionparity} provides polynomial reductions to games without fairness
constraints for the cases that
\begin{itemize} 
	\item $\alpha$ is a B\"uchi condition and $\beta=\bot$ (\cref{subsection:reductionbuchi});
	\item $\alpha$ is a parity condition and $\beta=\bot$ (\cref{subsection:reductionparitybot});
	\item both $\alpha$ and $\beta$ are parity conditions (\cref{subsection:reductionparityparity}).
\end{itemize}
Finally, we provide fixpoint algorithms for the latter two types of games in
\cref{section:fixpoints}.

%% file: preliminaries.tex
\section{Preliminaries}\label{section:preliminaries}

\noindent We introduce infinite-duration $\omega$-regular two-player games over finite graphs with additional \emph{strong transition fairness conditions} on both players.
For readability, we call the considered games (and their respective notions) simply \emph{fair}.

\myparagraph{Infinite Sequences.} We denote the set of
infinite sequences over a set $U$ by $U^\omega$.  We often view sequences
$\tau=u_1u_2\ldots\in U^\omega$ as
functions $\tau:\mathbb{N}\to U$, writing
$\tau(i)=u_i$.
Furthermore, we let
$\mathsf{Inf}(\tau):=\{u\in U\mid \forall i.\, \exists j>i.\,
\tau(j)=u\}$
denote the set of elements of $U$ that occur infinitely often in $\tau$. 
We use $\tilde{\tau}$ to denote the longest suffix (i.e. tail) of $\tau$ that is a sequence over $\mathsf{Inf}(\tau)$.
Given a function $f:U \rightarrow W$, we denote by $f(\tau)\in
W^\omega$ the pointwise application of $f$ to $\tau$.
Given natural numbers $m < n$ , we write $[m, n]:=\{m,\ldots,n\}$ and simply write $[n]:=\{1,\ldots,n\}$.

\myparagraph{Fair Game Arenas.}
We consider two-player games, played by player $\exists$ and its antagonist, player $\forall$.
A \emph{fair game arena} $A=(V_\exists,V_\forall,E,E_f)$ consists
of a set of \emph{nodes} $V=V_\exists\cup V_\forall$ that is partitioned
into the sets of \emph{existential nodes} $V_\exists$ and \emph{universal nodes} $V_\forall$,
together with a set $E\subseteq V\times V$ of \emph{moves} (or \emph{edges}) that is partitioned
into the set $E_f\subseteq E$ of \emph{fair moves} (or \emph{fair edges})  and the set
  $E\setminus E_f$ of \emph{normal moves} (or \emph{normal edges}).
If $E_f=\emptyset$, then we sometimes omit this component for brevity.
Given a node $v\in V$ and a binary relation $R\subseteq V\times V$, we write
$R(v)$ to denote the set $\{w\in V\mid (v,w)\in R\}$.
We assume that $E$ is right-total, that is, $E(v)\neq\emptyset$ for all
$v\in V$. We call a node $v$ \emph{fair}, if it is the source node of
a fair edge, i.e., $E_f(v)\neq \emptyset$ and collect all fair nodes
in the set $\Vfair=\{v \in V~|~ E_f(v)\neq \emptyset\}$ and define
$\Vn=V\setminus \Vfair$ to be the set of nodes that are not fair (`normal nodes').
We denote $\Vfair_\exists=\Vfair\cap V_\exists$ and $\Vfair_\forall =
\Vfair\cap V_\forall$.

\myparagraph{Plays.}
A \emph{play} $\tau=v_0v_1\ldots$ on $A$ is an infinite
sequence of nodes
s.t. $v_{i+1}\in E(v_i)$ for all $i\geq 0$. A finite prefix of a play is called a \textit{finite play}.
Given a play $\tau=v_0v_1\ldots$, we define the associated sequence of moves
$\tau_m=(v_0,v_1)(v_1,v_2)\ldots$.
Additionally, if $p$ is a player in $\{\exists, \forall \}$, we denote the other player by $1-p$.
We let $\mathsf{plays}(A)$ denote the set of all plays on $A$.

\myparagraph{$i$-fair Plays.}
Let $i\in\{\exists,\forall\}$. A play $\tau$ is called
\emph{$i$-fair} if, intuitively, player $i$ behaves fairly along $\tau$. Formally, $\tau$ is $i$-fair
if for every node $v\in V_i\cap\mathsf{Inf}(\tau)$, it holds
that ${E_f}|_v \subseteq\mathsf{Inf}(\tau_m)$ where
${E_f}|_v=\{(v,v')\in E_f\mid v'\in V\}$ denotes the set of
 fair edges that start at $v\in V$.
 Given a play $\tau$, we write $\fair_i(\tau)$ to indicate that $\tau$ is $i$-fair. We call a play \emph{mutually fair} if it is both $\exists$- and $\forall$-fair and \emph{mutually unfair} if it is neither $\exists$- nor $\forall$-fair.

\myparagraph{Strategies.} A strategy for player $i \in \{\exists, \forall\}$ (or an \emph{i-strategy}) is a function $p : V^*\cdot V_i \to V$ where for all $u \cdot v \in V^* \cdot V_i$  it holds that $p(u \cdot v) \in E(v)$.
A strategy $p$ is called \emph{positional} if 
$p(u\cdot v) = p(w \cdot v)$ for all $u, w \in V^*$ and $v \in V_i$.

A strategy $p$ for player $i$ is said to \emph{admit} a play $\tau=v_0v_1\ldots$ if for all $k \in \mathbb{N}$, $v_k \in V_i$ implies $p(v_0 \ldots v_k) = v_{k+1}$.
Alternatively, $\tau$ is said to be \emph{compliant} with $p$.
We write $\Sigma$ for the set of $\exists$-strategies and $\Pi$ for the set of $\forall$-strategies over a given arena.
Given a node $v \in V$ and strategies $s\in\Sigma$ and $t\in\Pi$, we obtain a unique induced play $\play_v(s,t)$ in the game arena. If we do not care about the initial node of this play, we simply write $\play(s,t)$.

A strategy for player $i \in \{\exists, \forall\}$ is an \emph{i-fair strategy} if every play it admits is $i$-fair. We write $\Sigma^\fair$ (resp. $\Pi^\fair$) for the set of $\exists$-fair (resp. $\forall$-fair) strategies. 

\myparagraph{Omega-regular Winning Conditions.}
We consider winning conditions given by an
$\omega$-regular~\cite{Tho97,Mazala01} language
$\varphi\subseteq V^\omega$ over the node set $V$. In particular, we
write $\varphi=\bot$ and $\varphi=\top$ to denote the trivial winning
conditions $\emptyset$ and $V^\omega$, respectively.
In particular, we
focus our attention to \emph{parity} winning conditions.
For a priority function $\lambda: V \to [k]$ that maps nodes of a
game arena to the natural numbers bounded by $k$ for some $k \in
\mathbb{N}$, the Parity$(\lambda)$ condition is given by $\varphi =
\{\tau \in V^\omega \mid \max(\mathsf{Inf}(\lambda(\tau))) \text{ is even}\}$.
We also
consider so-called B\"uchi conditions, that is, parity conditions
with set of priorities $\{1,2\}$.

\myparagraph{Omega-regular Games.}
An $\omega$-regular game is traditionally defined as a tuple $G = (A,
\alpha)$ where $A$ is a game arena \emph{without fair edges}, i.e.
$E_f = \emptyset$ and $\alpha \subseteq V^\omega$ an
$\omega$-regular winning condition.
An $\exists$-strategy $s \in \Sigma$ is said to be \emph{winning} (for player $\exists$) from a node $v \in V$ if for all $t \in \Pi$,  $\play_v(s, t) \in \alpha$. Dually,
a $\forall$-strategy $t \in \Pi$ is said to be winning (for player $\forall$)
from a node $v \in V$ if for all $s \in \Sigma$,  $\play_v(s, t) \not
\in \alpha$.

\myparagraph{Determinacy and Winning Regions.}
In $\omega$-regular games, every node $v\in V$ is won by one \emph{and
only one} of the players
\cite{Mar75,Mazala01}.
This property of a game is called \emph{determinacy}, and $\omega$-regular games
are \emph{determined}.
We denote the nodes from which player $\exists$ (resp. player $\forall$) has a winning strategy in $G$ by $\Win_\exists(G)$ (resp. $\Win_\forall(G)$).
When $G$ is clear from the context, we drop the parantheses and
write $\Win_\exists$ and $\Win_\forall$ instead.
Determinacy then amounts to $\Win_\exists\cup
\Win_\forall=V$ and $\Win_\exists\cap \Win_\forall=\emptyset$.

\myparagraph{Node Conventions for Figures.} Throughout this paper, in all figures,
the rectangular nodes represent $\forall$-player nodes and the nodes with round corners represent $\exists$-player nodes.

%% file: fairgames.tex

\section{Fair Games}\label{sec:fairgames}


As already outlined in the motivating example in~\cref{sec:intro}, the
interpretation of winning conditions over fair games influences the
characteristics of resulting winning strategies. To formalize this
intuition, we first recall a particular subclass of fair games,
namely those where only one player is restricted by an additional
fairness condition, in~\cref{sec:forallfair}. Then we use these
games to motivate winning semantics for the general class of fair
games. 

\subsection{Determinacy of $\forall$-Fair Games}\label{sec:forallfair}

\leavevmode\par\medskip

 A $\forall$-fair game is a tuple $G = (A, \alpha)$ where $A$ is a
 game arena with $\Vfair \subseteq V_\forall$ (\emph{called a
 $\forall$-fair game arena}), and $\alpha$ is an $\omega$-regular
 winning condition.

 In $\forall$-fair games, fairness constraints typically model known
 behavior of existing components that player $\exists$ (i.e., the
 to be synthesized system) can rely on.
 This is formalized by defining that player $\exists$ wins a
 $\forall$-fair game with winning condition $\alpha$ from node $v$
 if
 \begin{subequations}\label{eq:one-sided-wr}
 	\begin{align}\label{eq:one-sided-wr-E}
 		\exists s \in \Sigma.\, \forall t \in \Pi^\fair. \, \play_v(s, t) \in
 		\alpha. 
 	\end{align}
That is, player $\exists$ (or shortly, $\exists$) wins if they have a strategy that can win
against all $\forall$-fair $\forall$-strategies.

Our intuition tells us that this can be converted to reasoning about
general strategies for $\forall$-player (or shortly, $\forall$) by allowing $\exists$ to
win whenever $\forall$ plays unfairly.
In order to see this, we can look at the complement of~\cref{eq:one-sided-wr-E}, i.e., the description of when player
$\forall$ wins;
namely, 
\begin{equation}\label{eq:neg-one-sided-wr-E}
  \forall s \in \Sigma. \, \exists t\in \Pi^\fair. \, \play_v(s,t)
\notin \alpha. \tag{$\neg$1a}
\end{equation}
We can replace the quantification over fair strategies with a
quantification over all strategies but require that, in addition to refuting $\alpha$,
the resulting play be fair:
\begin{equation}\label{eq:fairness-inside}\forall s \in \Sigma.\,\exists t\in \Pi. \, \fair_\forall(\play_v(s,t))
\wedge \play_v(s,t)\notin \alpha. \tag{$\neg$1a*}\end{equation}
Next, we show in~\cref{lemma:elongation-in}  that
if strategy $t\in \Pi$ satisfies
$\fair_\forall(\play_v(s,t))$ then we can find a fair strategy
$t'\in \Pi^\fair$ for which $\play_v(s,t) = \play_v(s, t')$.
This proves ~(\ref{eq:fairness-inside}) $\Rightarrow$~(\ref{eq:neg-one-sided-wr-E}). As the other direction is trivial,~\cref{lemma:elongation-in} shows the equivalence of~\cref{eq:neg-one-sided-wr-E} and~\cref{eq:fairness-inside}.

\begin{lemma}\label{lemma:elongation-in}
Let $s \in \Sigma$, $t \in \Pi$ and $v \in V$. If $\fair_\forall(\play_v(s,t))$ is true, then there is a strategy $t' \in \Pi^\fair$ that satisfies $\play_v(s,t) = \play_v(s,t')$.
Similarly, if $\fair_\exists(\play_v(s,t))$ is true, then there is a strategy $s' \in \Sigma^\fair$ that satisfies $\play_v(s,t) = \play_v(s',t)$.
\end{lemma}
\begin{linenomath}
\begin{proof}[Proof of~\cref{lemma:elongation-in}]
    By definition, $\fair_\forall(\play_v(s,t))$ implies for all nodes $w\in V_\forall\, \cap\,\mathsf{Inf}(\play_v(s,t))$, that ${E_f}|_w \subseteq\mathsf{Inf}(\play_v(s,t)_m)$. To construct the desired strategy $t'$, we simply define a
    strategy that copies strategy $t$ on every occurring history (i.e., prefix) of $\play_v(s,t)$ and defines a dummy fair strategy on every history that does not occur in $\play_v(s,t)$ (i.e., every finite play that is not a prefix of $\play_v(s,t)$). The outcome of this dummy strategy will never be observed in $\play_v(s,t)$, 
    and thus $\play_v(s,t)$ will be equal to $\play_v(s,t')$.
    We define $t'$ formally as follows. For every $w \in V_\forall$, fix an order on the fair outgoing edges of $w$ and let $e_w$ denote the number of fair outgoing edges of $w$. Furthermore, let $\mycount(h, w)$ for a history $h \in V^*$ and $w \in V$ denote the number of times $w$ occurs in $h$.
    Define the strategy $t': V^* \cdot V_\forall \to V$ as follows: 
    For $h \in V^*$ with $t(h, w) = w'$ and $h \cdot w$ is a prefix of $\play_v(s,t)$, put $t'(h,w)=w'$.
    For $h\in V^*$ and $h \cdot w$ is not a prefix of $\play_v(s,t)$,
    put $t'(h,w)$ to the $\mycount(h, w) (\mymod \, e_w)^{th}$ fair outgoing edge of $w$, if $e_w \neq 0$. If $e_w = 0$, then define $t'(h,w)$ to be an arbitrary successor of $w$.
    In any play $\rho$ compliant with $t'$, any play that diverges from $\play_v(s,t)$ takes all the outgoing fair edges of all $w \in V_\forall$ that it visits infinitely often, which makes $t'$ a $\forall$-fair strategy. The proof for $s'$ is identical. 
\end{proof}
\end{linenomath}

Due to determinacy of $\omega$-regular games,~\cref{eq:fairness-inside} is equivalent to \begin{equation}\label{eq:fairness-inside-swapped}\exists t\in \Pi.\, \forall s \in
\Sigma.\,\fair_\forall(\play_v(s,t)) \wedge \play_v(s,t)\notin\alpha.\tag{$\neg$1a**}\end{equation}
In particular, this implies that 
for a node $v \in \Win_\forall$,
the fixed $\forall$-player strategy $t$ satisfying~\cref{eq:fairness-inside-swapped} is fair.
We therefore conclude that the following equation is equivalent to~\cref{eq:fairness-inside-swapped}, and therefore is the complement of~\cref{eq:one-sided-wr-E}:
  \begin{align}\label{eq:one-sided-wr-A}
\exists t \in \Pi^\fair.\, \forall s \in \Sigma. \, \play_v(s, t) \not \in \alpha.
\end{align}
 \end{subequations}
This statement is equivalent to the determinacy of $\forall$-fair
games: from each node either player $\exists$ has a winning strategy or
player $\forall$ has a winning $\forall$-fair strategy.

  
\subsection{From $\forall$-Fair Games to Defining Determined Fair Games}\label{sec:forallfairtofair}

\leavevmode\par\medskip

Given a fair game arena $A$ and an $\omega$-regular objective $\alpha$, a natural attempt to define winning regions in fair games with fairness constraints for both players would be to generalize~\cref{eq:one-sided-wr} to 
\begin{subequations}\label{equ:Call}
 \begin{align}
 v \in \Win_\exists\, \text{ if } &\, \exists s \in \Sigma^\fair.\, \forall t \in \Pi^\fair. \, \play_v(s, t) \in \alpha, \text{ and } 
 \label{eq:Cone}\\
 v\in \Win_\forall\, \text{ if } &\, \exists t \in \Pi^\fair.\, \forall s \in \Sigma^\fair. \, \play_v(s, t) \not \in \alpha. 
 \label{eq:Ctwo}
\end{align}
\end{subequations}

 In words, equations~\eqref{eq:Cone} and~\eqref{eq:Ctwo} express that for $i \in \{\exists, \forall\}$, $v$ is won by player $i$ if there exists an $i$-fair strategy for player $i$ that wins against all $(1-i)$-fair strategies of player $(1-i)$. However, in this case, $\Win_\exists \, \cup \, \Win_\forall \neq
V$. Indeed, equations~\eqref{eq:Cone} and~\eqref{eq:Ctwo} are
not complements of each other, that is,
\begin{align*}
  \exists s\in\Sigma^\fair.\,\forall t\in\Pi^\fair.\,\play(s,t)\in\alpha\qquad\not\Leftrightarrow\qquad
  \forall t\in\Pi^\fair.\,\exists s\in\Sigma^\fair.\,\play(s,t)\in\alpha.
\end{align*}
This observation shows that fair games in which winning regions are defined
by~\cref{equ:Call} undetermined.
The undetermined nodes $\orphans\subseteq V$ -- nodes from which none
of the players has a fair winning strategy -- form a separate set
of nodes, i.e., $V = \Win_\exists \dotcup \Win_\forall \dotcup
\orphans$.
To see this, consider the following example.

\begin{example}\label{example:division-of-objectives}\upshape
  Consider the fair game arena depicted in~\cref{fig:division} 
 where fair edges are shown by dashed lines,
  $\alpha=$ Parity$(\lambda)$ and each node is labeled by its priority assigned by $\lambda$.
  We
  observe that the player $\exists$ cannot enforce reaching the even node
 with a $\exists$-fair strategy from the two middle nodes.
  Every $\exists$-fair $\exists$-strategy $s$ has a counter
  $\forall$-fair $\forall$-strategy: choose the fair edge outgoing from
  the square node \emph{after} $s$ chooses the fair edge outgoing
  from the node with round corners.
  On the other hand, player $\forall$ cannot prevent the play from reaching the even node
   with a $\forall$-fair strategy from these
  nodes for exactly the same reason.
  Hence, the middle two nodes are neither in
  $\Win_\exists$ nor in $\Win_\forall$.
  That is, these two nodes are undetermined; therefore they form $\orphans$.

\end{example}
  \begin{figure}
  \centering
  \tikzset{every state/.style={minimum size=50pt}}
    \begin{tikzpicture}[
      auto,
      node distance=0.8cm,
      semithick
      ]
       \node[draw, rounded corners] (0) {$1$};
       \node (yo) [right of=0] {};
/       \node[rectangle,draw] (1) [right of=yo] {$1$};
       \node (yo1) [left of=0] {};
       \node[draw, rounded corners] (2) [left of=yo1] {$1$};
       \node (yo2) [right of=1] {};
       \node[draw, rounded corners] (3) [right of=yo2] {$2$};
       \path[->] (2) edge [loop left] node [below] {} (2);
       \path[->] (3) edge [loop right] node [below] {} (3);
       \path[->] (0) edge [bend right=30] node [pos=0.3,below] {} (1);
       \path[->] (1) edge [bend right=30] node [pos=0.3,below] {} (0);
       \path[->,dashed] (0) edge node [pos=0.3,left] {} (2);
       \path[->,dashed] (1) edge node [pos=0.3,left] {} (3);
       \draw [
    thick,
    decoration={
        brace,
        mirror,
        raise=0.5cm
    },
    decorate
] (0.west) -- (1.east);
       \draw [
    thick,
    decoration={
        brace,
        mirror,
        raise=0.5cm
    },
    decorate
] (2.west) -- (2.east);
       \draw [
    thick,
    decoration={
        brace,
        mirror,
        raise=0.5cm
    },
    decorate
] (3.west) -- (3.east);
    \node (t1) [below of=2] {$\mathsf{Win}_\forall$};
    \node (t2) [below of=3] {$\mathsf{Win}_\exists$};
    \node (t3) [below of=yo] {$O$};

    \end{tikzpicture}
    \caption{A simple fair game arena discussed
    in~\cref{example:division-of-objectives}.}
    \label{fig:division}
    
  \end{figure}

In order to better understand the distinction between
Equations~\ref{eq:Cone} and~\ref{eq:Ctwo}, we rely again on translation
to $\omega$-regular games.
In~\cref{lemma:equivalence-2a-3a} we show that the following is a reformulation of~\cref{eq:Cone}:
\begin{subequations}\label{eq:Ball}
\begin{equation}
  \exists s\in\Sigma.\forall t\in\Pi.\fair_\exists(\play_v(s,t))\wedge
 (\fair_\forall(\play_v(s,t)) \Rightarrow \play_v(s,t)\in\alpha).\label{eq:Bone}
\end{equation}

\begin{lemma}\label{lemma:equivalence-2a-3a}
A node $v \in V$ satisfies~\cref{eq:Cone} if and only if it satisfies~\cref{eq:Bone}. 
\end{lemma}

\begin{proof}[Proof of~\cref{lemma:equivalence-2a-3a}]
{(\ref{eq:Bone})$\Rightarrow$(\ref{eq:Cone})}: Let $v$ satisfy~\cref{eq:Bone}. Fix an $\exists$-strategy $\hat{s}\in \Sigma$ that satisfies 
\begin{equation}\label{eq:referenced}\forall t\in\Pi.\fair_\exists(\play_v(\hat{s},t))\wedge
    (\fair_\forall\,(\play_v(\hat{s},t)) \Rightarrow \play_v(\hat{s},t)\in\alpha).\tag{L2 - 1}
\end{equation} 

As for all $t\in\Pi$, $\fair_\exists(\play_v(\hat{s},t))$, we have
$\hat{s}\in \Sigma^\fair$.
Furthermore, as for all $t\in \Pi$, $\fair_\forall\,(\play_v(\hat{s},t)) \Rightarrow \play_v(\hat{s},t)\in\alpha$, we have for all $t \in \Pi^\fair$
that $\play_v(\hat{s},t)\in\alpha$. It follows that $v $ satisfies~\cref{eq:Cone}.

{(\ref{eq:Cone})$\Rightarrow$(\ref{eq:Bone})}: Let $v$ satisfy~\cref{eq:Cone}. Fix an $\exists$-strategy $\hat{s}\in \Sigma^\fair$ that satisfies
\begin{equation*}\forall t\in\Pi^\fair.\, \play_v(\hat{s}, t) \in \alpha.
\end{equation*} 

\begin{linenomath}
We will show that $\hat{s}$ satisfies~\cref{eq:referenced}. Pick a strategy $t \in \Pi$. Since $\hat{s} \in \Sigma^\fair$, $\fair_\exists(\play_v(\hat{s},t))$ trivially holds.
If $t \in \Pi^\fair$ then both $\fair_\forall(\play_v(\hat{s},t))$ and $\play_v(\hat{s},t) \in \alpha$ are satisfied. 
On the other hand if $t\not \in \Pi^\fair$, then there exists some $\exists$-strategy $s' \in \Sigma$ for which $\fair_\forall(\play_v(s',t))$ does not hold. If $\hat{s}$ is such a strategy, i.e. if $\play_v(\hat{s},t)$ is not $\forall$-fair, then the implication in~\cref{eq:referenced} trivially holds. 
If $\play_v(\hat{s},t)$ is $\forall$-fair, by~\cref{lemma:elongation-in} we know the existence of a $\forall$-strategy $t' \in \Pi^\fair$ with $\play_v(\hat{s},t) = \play_v(\hat{s},t')$. Due to~\cref{eq:Cone}, we have $\play_v(\hat{s},t') \in \alpha$, thus $v$ satisfies~\cref{eq:Bone}. 
\end{linenomath}
\end{proof}

A proof identical to that of~\cref{lemma:equivalence-2a-3a} witnesses that the following is a reformulation of~\cref{eq:Ctwo}:
\begin{equation}
	\exists t\in\Pi.\forall s\in\Sigma.\fair_\forall(\play_v(s,t))\wedge
	(\fair_\exists(\play_v(s,t)) \Rightarrow \play_v(s,t)\not\in\alpha).
	\label{eq:Btwo-alternative}
\end{equation}

\noindent Straightforward complementation of \cref{eq:Btwo-alternative} yields
\begin{equation*}
 \forall t \in \Pi. \exists s\in\Sigma. \fair_\forall(\play_v(s,t))\Rightarrow
    (\fair_\exists(\play_v(s,t)) \wedge \play_v(s,t)\in\alpha).\tag{$\neg$3b}
\end{equation*}
As the inner set is a Borel set, due to the determinacy theorem for games with Borel set objectives~\cite{Mar75}, the existential and universal quantifiers in the beginning can be swapped to obtain the following equivalent formulation:
\begin{equation}
 \exists s\in\Sigma.\forall t\in\Pi.\fair_\forall(\play_v(s,t))\Rightarrow
    (\fair_\exists(\play_v(s,t)) \wedge \play_v(s,t)\in\alpha).\tag{$\neg$3b}\label{eq:Btwo}
\end{equation}
\end{subequations}
To summarise, (\ref{eq:Cone})$=$(\ref{eq:Bone}) and (\ref{eq:Ctwo})$=$(\ref{eq:Btwo-alternative}) and $\neg$(\ref{eq:Ctwo})$=$(\ref{eq:Btwo}).
It is not hard to see that the difference between~\cref{eq:Bone}
and~\cref{eq:Btwo} -- so (\ref{eq:Cone}) and $\neg$(\ref{eq:Ctwo}) -- is in the way fairness is handled.
Namely, in~\cref{eq:Bone} player $\exists$ loses
whenever she plays unfairly regardless of how player $\forall$ plays.
Dually, in~\cref{eq:Btwo} player $\exists$ wins immediately
when player $\forall$ plays unfairly regardless of how player $\exists$
plays.
It follows that determinacy can be regained by deciding the winner of
the four different combinations of fairness with an $\omega$-regular
winning condition each, as summarized in the following table. 

\vspace{0.5cm}

  \begin{center}
  \begin{tabular}{|r|c|c|}
  \hline
    & $\fair_\exists(\tau)$ & $\neg\fair_\exists(\tau)$\\
    \hline
     $\fair_\forall(\tau)$ & $\tau\in\alpha$ & $\tau\in\gamma$ \\
     \hline
    $\neg \fair_\forall(\tau)$ & $\tau\in\delta$ & $\tau\in\beta$\\
  \hline
  \end{tabular}\label{table:division-of-objectives}
  \end{center}

\vspace{0.5cm}

  With this generalization, we obtain the winning condition~\eqref{eq:Bone} if $\beta=\gamma=\bot$ and $\delta=\top$ -- that is, $v \in \Win_\exists$ if and only if it satisfies~\eqref{eq:Bone}; otherwise, $v\in \Win_\forall$.
  Similarly, we obtain~\eqref{eq:Btwo} if $\gamma=\bot$ and $\beta=\delta=\top$ -- that is, $v \in \Win_\exists$ if and only if it satisfies~\eqref{eq:Btwo}; otherwise, $v\in \Win_\forall$.

We note that the discussion of determinacy has crucial importance to
the analysis of games and the decision of how to model particular
scenarios.
For example, if fairness of $\forall$-player arises from physical
constraints (as, e.g., in \cite{banerjee2022fast}) then it might make
sense to consider the winning condition~(\ref{eq:Ctwo})$=$(\ref{eq:Btwo}), which corresponds to
$\beta=\top$, as this allows player $\exists$ to win immediately when $\forall$-fairness is violated.
Dually, if fairness of player $\exists$ must be adhered to, then it
makes sense to consider the winning condition~(\ref{eq:Cone})$=$(\ref{eq:Bone}), which corresponds to
$\beta=\bot$.
Our formulation allows to further fine tune what happens when both act
unfairly by adjusting $\beta$.

Given the intuition that fairness constraints are actually additional
obligations for both players, the choice of $\gamma=\bot$ and
$\delta=\top$ assumed in Equations~\eqref{equ:Call}-\eqref{eq:Ball} is very
natural. However, allowing mutually unfair plays to be decided by a
different $\omega$-regular winning condition $\beta$, allows games with
more symmetric
winning semantics e.g., by setting $\beta=\alpha$.
We therefore restrict our attention in this paper to fair games with two
winning conditions $\alpha$ and $\beta$ while if $i$-player plays
fairly but $(1-i)$-player plays unfairly, $i$-player wins, i.e.,
$\gamma:=\bot$ and $\delta:=\top$. This is formalized
next.

\begin{definition}[Fair Games]\label{def:fairgames}
  A \emph{fair game} $G=(A,\alpha,\beta)$ consists of a fair game arena
  $A$ together with two ($\omega$-regular) winning conditions
  $\alpha,\beta\subseteq \mathsf{plays}(A)$ where $\alpha$ and $\beta$
  determine the winner of mutually fair and mutually unfair plays,
  respectively. In fair games, a play that is $i$-fair and
  $(1-i)$-unfair is won by player $i$.
  Formally, in the fair game $G=(A,\alpha,\beta)$, $v \in \Win_\exists $ if and only if,
  \begin{align}
    \exists s\in\Sigma.\forall
    t\in\Pi.\,&\fair_\exists(\play_v(s,t))\wedge
    (\fair_\forall(\play_v(s,t)) \Rightarrow
    \play_v(s,t)\in\alpha)\notag\\
   \vee (&\neg \fair_\exists(\play_v(s,t))\wedge \neg
   \fair_\forall(\play_v(s,t)) \wedge
   \play_v(s,t)\in\beta)\label{eq:Bexists}
  \end{align}
\end{definition}
  The winning region of player $\forall$ is defined dually by defining $v \in \Win_\forall$ if and only if,
  \begin{align*}
    \exists t\in\Pi.\forall s\in\Sigma.\,&\fair_\forall(\play_v(s,t))\wedge (\fair_\exists(\play_v(s,t)) \Rightarrow \play_v(s,t)\not\in\alpha)\notag\\
   \vee (&\neg \fair_\forall(\play_v(s,t))\wedge \neg \fair_\exists(\play_v(s,t)) \wedge \play_v(s,t)\not \in\beta)
  \end{align*}

 Just like in the acquisition of~\cref{eq:Btwo}, a straightforward complementation of \cref{eq:Bexists} together with a swap of existential and universal quantifiers (due to Borel determinacy) shows that the above formula is equivalent to the complement of~\cref{eq:Bexists}. Consequently, fair games are determined.

  \begin{notation}
  	We call a fair game $G=(A,\alpha,\beta)$ a fair $\alpha / \beta$ game.
  Further, if $\alpha$ or $\beta$ are given by mentioned winning conditions 
   (e.g. $\alpha=$ Parity$(\lambda)$, $\beta=\bot$),
   with slight abuse of notation, we refer to the game with the name of the objectives (e.g. fair parity/$\bot$ game).
  \end{notation}

\begin{remark}
Stochastic games (see, e.g.,~\cite{CJH03}) allow for an additional set $V_s$ of stochastic game nodes
that belong to neither player $\exists$ nor player $\forall$,
and for which the stochasticity is resolved uniformly at random.
It is known that for purposes of qualitive analysis (i.e., the computation of almost-sure winning strategies), stochastic games can be seen as the special case of $\forall$-fair games in which
$E(v)\subseteq E_f$ holds for all stochastic nodes $v\in V_s$,
and $E_f\cap E(v)=\emptyset$ for all non-stochastic nodes $v\in V_\exists\cup V_\forall$,
that is, all stochastic edges are fair edges, but no
non-stochastic edges are fair edges. This encoding treats
stochastic branching as adversarial for the system (player $\exists$).
\end{remark}

%% file: strategiesinfairgames.tex
\subsection{Mutually Fair Strategies in Fair Parity Games}

\leavevmode\par\medskip

In~\cref{sec:forallfairtofair} and in particular in~\cref{example:division-of-objectives} we have discussed the mutually unfair plays and strategies that take such plays into account in fair $\alpha/\beta$ games.
In this section, we start restricting our attention to fair parity/$\beta$ games (as this will be our focus for most of the rest of the paper) and discuss the particularities of mutually fair strategies in such games. 
We will do this with the help of the games $G_1-G_4$ depicted in~\cref{fig:fairstrategies}.
None of these games allows for a mutually unfair play. This is because on all given arenas the unfair behaviour of one player
makes the play trivially fair for the other player. Therefore, the winning regions are independent of $\beta$ in these examples.

\begin{figure}
	\centering
	\tikzset{every state/.style={minimum size=10pt}}
	\begin{tikzpicture}[
		auto,
		node distance=0.8cm,
		semithick
		]
		\node[draw, rectangle] (n3f) {$3$};
		\node[draw, rounded corners] (n4f) [right= 7mm of n3f] {$4$};
		\node (empty) [above=0.5 cm of n4f] {};
		\path[->] (n3f) edge [loop above] node [below] {} (n3f);
		\path[->, dashed] (n3f) edge [bend right=30] node [pos=0.3,below] {} (n4f);
		\path[->, red, very thick] (n4f) edge [bend right=30] node [pos=0.3,below] {} (n3f);
		
		\node (label1-0) [above=5mm of n3f] {};
		\node (label1) [left=0.1cm of label1-0] {$G_1$:};
		\node[draw, rectangle] (n3) [right= 25mm of empty] {$3$};
		\node[draw, rectangle] (n1) [left= 7mm of n3] {$1$};
		\node (label2) [left=3mm of n1] {$G_2$:};
		\node[draw, rounded corners] (n4) [right= 7mm of n3] {$4$};
		\path[->] (n3) edge [loop above] node [below] {} (n3);
		\path[->, blue, very thick] (n1) edge [loop above] node [below] {} (n1);
		\path[->, dashed] (n3) edge [bend right=30] node [pos=0.3,below] {} (n4);
		\node (empty2) [below=0.5 cm of n4] {};
		\path[->] (n4) edge [bend right=30] node [pos=0.3,below] {} (n3);
		\path[->, blue, very thick] (n3) edge [bend right=30] node [pos=0.3,below] {} (n1);
		\path[->] (n1) edge [bend right=30] node [pos=0.3,below] {} (n3);
		\node[draw, rectangle] (n3n) [below= 1cm of n1] {$3$};
		\node (label3) [left=3mm of n3n] {$G_3$:};
		\node[draw, rounded corners] (n4n) [right= 7mm of n3n] {$4$};
		\node[draw, rounded corners] (n5n) [right= 7mm of n4n] {$5$};
		\path[->] (n3n) edge [loop above] node [below] {} (n3n);
		\path[->, dashed, blue, very thick] (n3n) edge [bend right=30] node [pos=0.3,below] {} (n4n);
		\path[->] (n4n) edge [bend right=30] node [pos=0.3,below] {} (n3n);
		\path[->, dashed] (n4n) edge [bend right=30] node [pos=0.3,below] {} (n5n);
		\path[->] (n5n) edge [bend right=30] node [pos=0.3,below] {} (n4n);
		\node[draw, rounded corners] (n4l) [right=25mm of empty2] {$4$};
		\node[draw, rectangle] (n3l) [left=7mm of n4l]{$3$};
		\node[draw, rounded corners] (n5l) [right= 7mm of n4l] {$5$};
		\node[draw, rounded corners] (n6l) [above= 5mm of n4l] {$6$};
		\node[draw, rectangle] (n1l) [below= 5mm of n4l] {$1$};
		\path[->] (n3l) edge [loop above] node [below] {} (n3l);
		\path[->, dashed] (n3l) edge [bend right=30] node [pos=0.3,below] {} (n4l);
		\path[->] (n4l) edge [bend right=30] node [pos=0.3,below] {} (n3l);
		\path[->, dashed, red, very thick] (n4l) edge [bend right=30] node [pos=0.3,below] {} (n5l);
		\path[->,red, very thick] (n5l) edge [bend right=30] node [pos=0.3,below] {} (n4l);
		\path[->, red, very thick] (n4l) edge [bend right=30] node [pos=0.3,below] {} (n6l);
		\path[->,red, very thick] (n6l) edge [bend right=30] node [pos=0.3,below] {} (n4l);
		\path[->, dashed, red, very thick] (n4l) edge [bend right=30] node [pos=0.3,below] {} (n1l);
		\path[->] (n1l) edge [bend right=30] node [pos=0.3,below] {} (n4l);
		
		\node (label4) [left=6mm of n6l] {$G_4$:};
	\end{tikzpicture}
	\caption{\small 
		Four fair parity/$\beta$ games: dashed lines represent fair edges.   Games $G_1$ and $G_4$ are won by $\exists$-player and $G_2$ and $G_3$ are won by $\forall$-player. In each case, a respective winning strategy is shown by colored edges. A set of colored edges represents a 	strategy that takes only the colored edges in the game, and whenever a source node is visited infinitely often, all its colored outgoing edges are taken alternatingly.
	}
	\label{fig:fairstrategies}
\end{figure}

In game $G_1$, both nodes are won by player $\exists$. Player $\forall$ loses
node $3$ since taking the self loop on $3$ makes the play visit $3$
infinitely often, however, it forces $\forall$ to play fairly, implying
that they must take the edge to $4$ infinitely often as well.
Therefore, every $\forall$-fair play is won by player $\exists$ since the
priority $4$ is seen infinitely often in such plays. Also note that if
$\forall$-player decides not to play fairly, they immediately lose
since all plays are trivially $\exists$-fair.

To get to game $G_2$, we append node $1$ to the left of $G_1$. Here,
all the nodes are won by player $\forall$. This is because they
win node $3$ by eventually taking the outgoing edge to $1$ and then
staying in $1$ forever with the self-loop. By doing so, player $\forall$
evades
his obligation to take the fair edge by forcing each play to see node
$3$ a finite number of times.

To get to game $G_3$, we append node $5$ to the right of game $G_1$.
Again, all the nodes are won by player $\forall$ even though this time they
cannot evade taking their fair edges.
In this game, player $\forall$ wins due to the obligation of player $\exists$ to play fairly.
In a play starting from $3$, player $\forall$ must eventually take the
outgoing edge to $4$. From there on, any $\forall$-fair play will visit node $4$
infinitely often, forcing $\exists$ to take their outgoing edge to $5$
infinitely often. As a consequence, in every mutually fair play, $5$ is
seen infinitely often. Therefore, all nodes in the game are won by player $\forall$.

Finally, to get to game $G_4$, we append two nodes to game $G_3$. This
time, all the nodes are won by player $\exists$. They still need
to take their fair outgoing edges to $5$ (and this time, also to the
new node $1$) infinitely often. But now they can also take the
outgoing edge to $6$ infinitely often and thereby win the game.

%% file: reductiontonormalgames.tex
\section{Reduction to Games without Fairness}\label{section:reductionparity}

In this section, we show how fair B\"uchi and parity games can be reduced to games without
fairness constraints. Our reductions work by systematically replacing each fair node in the fair game with a 3-step gadget, transforming the fair game into an equivalent standard parity game.

We establish reductions with linear blow-up for fair $\alpha/\beta$ games when $\alpha$ is parity and $\beta \in \{\top, \bot\}$, and a reduction with quadratic blow-up for fair parity/parity games where both objectives are non-trivial. Our gadget constructions are inspired by the gadget-based approach of Chatterjee et al.~\cite{CJH03} for stochastic parity games.

Section~\ref{subsection:reductionbuchi} provides an intuitive presentation via fair Büchi/$\bot$ games, building understanding of the gadget construction with minimal technical complexity. Section~\ref{sec:alternative-gadgets} presents alternative gadget variants and their implications. Section~\ref{subsection:reductionparitybot} gives the formal reduction for fair parity/$\bot$ games. Section~\ref{subsection:reductionparityparity} extends to fair parity/parity games, where memory is required to track objective interactions, yielding a reduction with quadratic increase of game arena size.

\input{reductionofmutuallyfairbuchi.tex}

\input{alternativegadgets.tex}

\input{reductionofmutuallyfairparity.tex}

\input{reductionofmutuallyfairparityparity.tex}

%% file: reductionofmutuallyfairbuchi.tex
\subsection{Reduction of Fair B\"uchi/$\bot$ Games}\label{subsection:reductionbuchi}

\leavevmode\par\medskip

Before presenting the formal framework for fair parity/$\bot$ games in Section~\ref{subsection:reductionparitybot}, we develop intuition through a simpler case: fair Büchi/$\bot$ games. Since Büchi games are a special case of parity games, this serves as a gentle introduction to the general reduction technique while capturing the essential ideas.

\noindent \textbf{Büchi conditions as restricted parity conditions.} A Büchi condition is a Parity$(\lambda)$ condition where $\lambda$ ranges over $\{1, 2\}$. A Büchi condition can alternatively be represented by specifying the set $T$ of nodes with priority~$2$, rather than defining the priority function $\lambda$ explicitly. A Büchi$(T)$ condition is then given by $\varphi = \{\tau \mid \mathsf{Inf}(\tau)\cap T\neq \emptyset\}$, where $T = \lambda^{-1}(2)$.

\noindent \textbf{The reduction.} Consider a fair Büchi/$\bot$ game $G=(A, \text{Büchi}(T), \bot)$ where $A = (V_\exists,V_\forall,E,E_f)$ is a fair game arena. Our reduction replaces each fair node $v \in \Vfair$ with an appropriate 3-step gadget, as shown in~\cref{fig:buchigadgets-existential}:
\begin{itemize}
    \item Each node $v \in \Vfair_\exists$ is replaced with an \emph{existential gadget} (left gadget in~\cref{fig:buchigadgets-existential}).
    \item Each node $v \in \Vfair_\forall$ is replaced with a \emph{universal gadget} (right gadget in~\cref{fig:buchigadgets-existential}).
\end{itemize}
All incoming edges to $v$ in $G$ are redirected to the top node of the gadget.
The gadget's outgoing edges connect to the successors of $v$ in $G$: specifically, to $E(v)$ (all outgoing edges) and $E_f(v)$ (fair outgoing edges only).

The reduced game $G'$ has priority function $\lambda: V \cup V^{\mathsf{gad}} \to [3]$, where $V^{\mathsf{gad}}$ denotes the newly introduced gadget nodes. Priorities are assigned as follows:
\begin{itemize}
    \item \emph{Original nodes:} For $v \in T$, set $\lambda(v) = 2$; for $v \in V \setminus T$, set $\lambda(v) = 1$.
    \item \emph{Gadget bottom nodes:} The priorities of the bottom-level nodes in each gadget are explicitly shown next to the nodes in~\cref{fig:buchigadgets-existential}.
    \item \emph{Gadget intermediate nodes:} The priorities of intermediate nodes $v_1$ and $v_2$ do not affect the outcome. Without loss of generality, we can assign them the minimum priority $1$.
\end{itemize}
This reduction yields an equivalent parity game $G'$ with three priorities, preserving winning regions: player~$i \in \{\exists, \forall\}$ wins a node $v \in V$ in the original fair game $G$ if and only if player~$i$ wins the corresponding node in the reduced parity game~$G'$.

\noindent \textbf{Organization of this section.} We present the intuition behind this reduction through multiple complementary perspectives. We emphasize that the explanations provided here are intended to build intuition rather than constitute a formal proof—we provide the rigorous correctness argument  in~\cref{subsection:reductionparitybot}. Our presentation is structured as follows: 
\begin{itemize} 
  \item We begin by describing the gadget construction in Section~\ref{subsec:construction}, accompanied by an illustrative example. 
  \item Section~\ref{subsec:branching} explains the design rationale behind the branching structure of the gadgets, clarifying how they encode fairness constraints through priorities.
  \item Section~\ref{subsec:local-correctness} provides a local correctness argument, demonstrating why the gadgets preserve winning regions at the level of individual nodes. \end{itemize}

\noindent In Section~\ref{sec:alternative-gadgets}, we present alternative gadget designs that achieve the same reduction with different structural properties. The formal reduction for fair parity/$\bot$ games is developed in~\cref{subsection:reductionparitybot}; the correctness of the Büchi/$\bot$ case follows as a special case. While the formal proofs establish global correctness through rigirous argumentation, our goal here is to provide local, intuitive understanding of the gadget mechanisms.


\subsubsection{Gadget construction via an example}\label{subsec:construction}

We illustrate the reduction by applying it to a concrete Büchi/$\bot$ game. 
\begin{example}
Consider the Büchi/$\bot$ game derived from the fair game arena in~\cref{fig:division}, previously discussed in~\cref{example:division-of-objectives}. The $\alpha$-objective is the Büchi condition inherited from that example, and we set $\beta = \bot$ to declare mutually unfair plays as winning for player~$\forall$. 

In the original fair game, the leftmost and rightmost nodes are won by players $\forall$ and $\exists$, respectively, based solely on the $\alpha$-objective. The two middle nodes are won by player~$\forall$ due to the choice of $\beta = \bot$.

Applying the gadgets from~\cref{fig:buchigadgets-existential} yields the equivalent parity game shown in~\cref{fig:Buechi-gadgets-running-ex}.
In this figure, the nodes with the gray background depict game nodes $V$ inherited from $G$ and the newly added gadget nodes $V^{\mathsf{gad}}$ are depicted with white background, omitting the notation for the gadget nodes.

   One can readily verify that the nodes with the gray background have the same winners in both the Büchi/$\bot$ game and the reduced parity game: all nodes with gray background except for the rightmost one are won by the $\forall$ player. 
   We depict the associated positional winning strategy for $\forall$ player from his winning region by thicker blue edges.
\end{example}

\begin{figure}
  \begin{minipage}[t]{0.49\linewidth}\centering
\tikzset{every state/.style={minimum size=15pt}, every node/.style={minimum size=15pt}}\scalebox{.8}{ 
    \begin{tikzpicture}[
    auto,
    node distance=0.8cm,
    semithick
    ]
     \node[draw, rounded corners] (0) {$v$};
     \node (yo) [below of=0] {};
     \node (yoz) [below of=yo] {};
     \node[rectangle,draw] (1) [left of=yoz] {$v_1$};
     \node[rectangle,draw] (2) [right of=yoz] {$v_2$};
     \node (yoy) [below of=yoz] {};
     \node[rectangle,draw, label=right:{$2$}] (4) [below of=yoy] {$v^\forall_1$};
     \node (yo3) [left of=4] {};
     \node[draw, rounded corners, colored, label=right:{$1$}] (3) [left of=yo3] {$v^\exists_1$}; 
     \node (yo4) [right of=4] {};
     \node[draw, rounded corners, label=right:{$3$}] (5) [right of=yo4] {$v^\exists_2$}; 
     \node (yo5) [below of=3] {};
     \node (yo6) [below of=4] {};
     \node (yo7) [below of=5] {};
     \node (out1) [below of=yo5] {\color{colored}$E(v)$};
     \node (out2) [below of=yo6] {$E_f(v)$};
     \node (out3) [below of=yo7] {$E(v)$};
     \path[->] (0) edge node {} (1);
     \path[->] (0) edge node {} (2);
     \path[->, colored] (1) edge node {} (3);
     \path[->] (1) edge node {} (4);
     \path[->] (2) edge node {} (5);
     \path[->, colored] (3) edge node {} (out1);
     \path[->] (4) edge node {} (out2);
     \path[->] (5) edge node {} (out3);
   
  \end{tikzpicture}
\qquad
    \begin{tikzpicture}[
    auto,
    node distance=0.8cm,
    semithick
    ]
    \node[draw, rounded corners] (0) {$v$}; 
     \node (yoz) [below of=yo] {};
     \node[rectangle,draw] (1) [left of=yoz] {$v_1$};
     \node (yoy) [below of=yoz] {};
     \node[rectangle,draw, label=right:{$2$}] (4) [below of=yoy] {$v^\forall_1$}; 
     \node (yo3) [left of=4] {};
     \node[draw, rounded corners, colored, label=right:{$1$}] (3) [left of=yo3] {$v^\exists_1$}; 
     \node (yo5) [below of=3] {};
     \node (yo6) [below of=4] {};
     \node (out1) [below of=yo5] {\color{colored}$E_f(v)$};
     \node (out2) [below of=yo6] {$E(v)$};
     \path[->] (0) edge node {} (1);
     \path[->, colored] (1) edge node {} (3);
     \path[->] (1) edge node {} (4);
     \path[->, colored] (3) edge node {} (out1);
     \path[->] (4) edge node {} (out2);
  \end{tikzpicture}}
\caption{Existential gadgets that replace $v \in \Vfair_\exists$ (left) and $v \in \Vfair_\forall$ (right) in the fair B\"uchi/$\bot$ game.}\label{fig:buchigadgets-existential}
\end{minipage}
\vrule 
\begin{minipage}[t]{0.49\linewidth} \centering
\tikzset{every state/.style={minimum size=15pt}, every node/.style={minimum size=15pt}}
   \scalebox{.8}{ \begin{tikzpicture}[
    auto,
    node distance=0.8cm,
    semithick
    ]
     \node[rectangle,draw] (0) {$v$}; 
     \node (yo) [below of=0] {};
     \node (yoz) [below of=yo] {};
     \node[draw, rounded corners, colored] (1) [left of=yoz] {$v_1$};
     \node[draw, rounded corners] (2) [right of=yoz] {$v_2$};
     \node (yoy) [below of=yoz] {};
     \node[rectangle,draw, label=right:{$2$}] (4) [below of=yoy] {$v^\forall_2$}; 
     \node (yo3) [left of=4] {};
     \node[draw, rounded corners, colored, label=right:{$1$}] (3) [left of=yo3] {$v^\exists_1$}; 
     \node (yo4) [right of=4] {};
     \node[draw, rounded corners, label=right:{$3$}] (5) [right of=yo4] {$v^\exists_2$}; 
     \node (yo5) [below of=3] {};
     \node (yo6) [below of=4] {};
     \node (yo7) [below of=5] {};
     \node (out1) [below of=yo5] {\color{colored}$E(v)$};
     \node (out2) [below of=yo6] {$E_f(v)$};
     \node (out3) [below of=yo7] {$E(v)$};
     \path[->, colored] (0) edge node {} (1);
     \path[->] (0) edge node {} (2);
     \path[->, colored] (1) edge node {} (3);
     \path[->] (2) edge node {} (4);
     \path[->] (2) edge node {} (5);
     \path[->, colored] (3) edge node {} (out1);
     \path[->] (4) edge node {} (out2);
     \path[->] (5) edge node {} (out3);
  \end{tikzpicture}
\qquad
    \begin{tikzpicture}[
    auto,
    node distance=0.8cm,
    semithick
    ]
    \node[rectangle,draw] (0) {$v$};
     \node (yoz) [below of=yo] {};
     \node[draw, rounded corners, colored] (1) [left of=yoz] {$v_1$};
     \node[draw, rounded corners] (2) [right of=yoz] {$v_2$};
     \node (yoy) [below of=yoz] {};
     \node[rectangle,draw, label=right:{$2$}] (4) [below of=yoy] {$v^\forall_2$}; 
     \node (yo3) [left of=4] {};
     \node[draw, rounded corners, colored, label=right:{$1$}] (3) [left of=yo3] {$v^\exists_1$}; 
     \node (yo5) [below of=3] {};
     \node (yo6) [below of=4] {};
     \node (out1) [below of=yo5] {\color{colored}$E_f(v)$};
     \node (out2) [below of=yo6] {$E(v)$};
     \path[->, colored] (0) edge node {} (1);
     \path[->] (0) edge node {} (2);
     \path[->, colored] (1) edge node {} (3);
     \path[->] (2) edge node {} (4);
     \path[->, colored] (3) edge node {} (out1);
     \path[->] (4) edge node {} (out2);
  \end{tikzpicture}}
  \caption{Universal gadgets that replace $v \in \Vfair_\exists$ (left) and $v \in \Vfair_\forall$ (right) in the fair B\"uchi/$\bot$ game. } \label{fig:buchigadgets-universal}
\end{minipage}\caption*{Nodes are labeled with their priorities. Colored edges can be ignored if $v \in T$.}
\end{figure}

\subsubsection{Intuition behind the branching structure}\label{subsec:branching}

Having seen the gadget construction in action, we now explain the design rationale: \emph{why} do these gadgets correctly encode fairness constraints? The key insight is that gadgets monitor which player controls branching at fair nodes, using the parity condition to detect and penalize unfair behavior. Intuitively, if a fair node is visited infinitely often and one player behaves fairly, then the owner of that fair node must eventually relinquish control of the branching, or pay a high price.

\noindent \textbf{\textbullet \hspace{1mm} $\Vfair_\exists$ gadget.}
Consider the gadget on the left in~\cref{fig:buchigadgets-existential}, which replaces a node $v \in \Vfair_\exists$. Player~$\exists$ faces a strategic choice at the top node:
\begin{itemize}
    \item \emph{Retain control:} Move to $v_2$, then to $v_2^\exists$, triggering priority~$3$. This allows $\exists$ to choose any successor of $v$ (fair or unfair).
    \item \emph{Yield control:} Move to $v_1$, passing the decision to player~$\forall$.
\end{itemize}

When player~$\exists$ yields control by moving to $v_1$, player~$\forall$ then chooses:
\begin{itemize}
    \item \emph{Pick a fair move:} Select a fair successor via $v_1^\forall$, triggering priority~$2$.
    \item \emph{Return control:} Move to $v_1^\exists$, allowing $\exists$ to pick any successor, triggering priority~$1$.
\end{itemize}

\noindent \textit{Why this works.} Player~$\exists$ can retain control indefinitely by repeatedly triggering priority~$3$, but this is costly: since $3$ is the maximum priority, doing so infinitely often causes $\exists$ to lose the parity game. Thus, $\exists$ can only retain control finitely many times, which suffices to mimic any $\exists$-fair strategy from $G$ that visits $v$ infinitely often. 

\noindent \textbf{\textbullet \hspace{1mm} $\Vfair_\forall$ gadget.}
Consider the gadget on the right side in~\cref{fig:buchigadgets-existential}, which replaces a node $v \in \Vfair_\forall$. The situation is dual. Player~$\forall$ can:
\begin{itemize}
    \item \emph{Retain control:} Move to $v_1^\forall$, and choose any successor, triggering priority~$2$.
    \item \emph{Yield control:} Move to $v_1^\exists$, passing the decision to player~$\exists$, allowing it to pick any fair successor, triggering priority~$1$.
\end{itemize}

\noindent \textit{Why this works.} As in the $\Vfair_\exists$ gadget, player~$\forall$ can retain control at a cost—but here the cost is priority~$2$, which is less catastrophic than the priority~$3$ cost in the $\Vfair_\exists$ gadget. This asymmetry is crucial for handling mutually unfair plays: if both players ``need to" behave unfairly, they will both choose the ``retain control" option in their respective gadgets, taking the rightmost branches. In such plays, the maximum priority seen will be~$3$, ensuring that mutually unfair plays are won by player~$\forall$.

  \begin{figure}
  \centering
  \tikzset{every state/.style={minimum size=50pt}}
    \begin{tikzpicture}[
      auto,
      node distance=0.8cm,
      semithick
      ]
       \node[draw, rounded corners, fill=gray!30] (gray0) {$1$};
       \node (yo) [right=1.5 cm of gray0] {};
/       \node[draw,rounded corners, fill=gray!30] (gray1) [right=1.5cm of yo] {$1$};
       \node (yo1) [left=2.5cm of gray0] {};
       \node (yo2) [right=2.5cm of gray1] {};
       \path[->] (gray0) edge [bend right=15] node [pos=0.3,below] {} (gray1);
       \path[->] (gray1) edge [bend right=15] node [pos=0.3,below] {} (gray0);


    \node (gyo) [below of=gray0] {};
     \node (yoz) [below of=gyo] {};
     \node[rectangle,draw] (1) [left=1em of yoz] {\phantom{1}};
     \node[rectangle,draw] (2) [right=1em of yoz] {\phantom{1}};
     \node (yoy) [below of=yoz] {};
     \node[rectangle,draw] (4) [below of=yoy] {$2$};
     \node (g-yo3) [below of=1] {};
     \node (yo3) [below of=g-yo3] {};
     \node[draw, rounded corners] (3) [left=1em of yo3] {$1$}; 
     \node (g-yo4) [below of=2] {};
     \node (yo4) [below of=g-yo4] {};
     \node[draw, rounded corners] (5) [right=1em of yo4] {$3$}; 

     \path[->] (gray0) edge node {} (1);
     \path[->] (gray0) edge node {} (2);
     \path[->] (1) edge node {} (3);
     \path[->, blue, ultra thick] (1) edge node {} (4);
     \path[->, blue, ultra thick] (2) edge node {} (5);


     \node (below3) [below=0.7cm of 3] {};

     \node[draw, rounded corners, fill=gray!30] (gray2) [left=1cm of below3] {$1$}; 
     \path[->] (gray2) edge [loop left] node [below] {} (gray2);

     \tikzset{extends/.style={->, shorten >=1pt}}

     \coordinate (mid3) at ($(3.south |- gray2.east)$);
     \coordinate (mid4) at ($(4.south |- gray2.east)$);
     \coordinate (mid5) at ($(5.south |- gray2.east)$);

     \draw (3.south) -- (mid3);
     \draw[blue, ultra thick] (4.south) -- (mid4);
     \draw (5.south) -- (mid5);

     \draw[extends] (mid5) -- (gray2.east);
     \draw[extends, blue, ultra thick] (mid4) -- (gray2.east);


    \node (below5) [below of=5] {};
    \coordinate[right=1em of below5] (under5);
    \draw (3.south) -- (under5);
    \draw (5.south) -- (under5);
    \path[->] (under5) edge node {} (gray1);



     \node (scnd-yo) [below of=gray1] {};
     \node[rectangle,draw] (scnd-1) [below of=scnd-yo] {\phantom{1}};
     \node (scnd-yoy) [below of=scnd-1] {};
     \node (scnd-yoyo) [below of=scnd-yoy] {};
     \node[rectangle,draw] (scnd-4) [right=1em of scnd-yoyo] {$2$}; 
     \node[draw, rounded corners] (scnd-3) [left=1em of scnd-yoyo] {$1$}; 
     \path[->] (gray1) edge node {} (scnd-1);
     \path[->] (scnd-1) edge node {} (scnd-3);
     \path[->, blue, ultra thick] (scnd-1) edge node {} (scnd-4);

     \node (below-scnd-4) [below=0.7cm of scnd-4] {};
      \node[draw, rounded corners, fill=gray!30] (gray3) [right=1cm of below-scnd-4] {$2$};
      \path[->] (gray3) edge [loop right] node [below] {} (gray3);

     \coordinate (mid-scnd-3) at ($(scnd-3.south |- gray3.west)$);
     \coordinate (mid-scnd-4) at ($(scnd-4.south |- gray3.west)$);

     \draw (scnd-3.south) -- (mid-scnd-3);
     \draw(scnd-4.south) -- (mid-scnd-4);
     \draw[extends] (mid-scnd-3) -- (gray3.west);

     \path[->, blue, ultra thick] (scnd-4) edge node {} (gray0);


       \draw [
    thick,
    decoration={
        brace,
        mirror,
        raise=0.5cm
    },
    decorate
] (gray2.west) -- (gray2.east);
       \draw [
    thick,
    decoration={
        brace,
        mirror,
        raise=0.5cm
    },
    decorate
] (gray3.west) -- (gray3.east);
    \node (t1) [below of=gray2] {$\mathsf{Win}_\forall$};
    \node (t2) [below of=gray3] {$\mathsf{Win}_\exists$};

    \end{tikzpicture}
    \caption{The reduced parity game equivalent to the Büchi/$\bot$ game where the game graph and the $\alpha$ priorities are inherited from
    the fair game arena depicted in~\cref{fig:division}. The reduction uses the existential gadgets from~\cref{fig:buchigadgets-existential}.}
    \label{fig:Buechi-gadgets-running-ex}
  
  \end{figure}

\subsubsection{Local preservation of winning regions}\label{subsec:local-correctness}
We now argue informally why the gadget reduction preserves winning regions at the local level. The formal proof appears in~\cref{subsection:reductionparitybot}; here we focus on building intuition by examining how individual nodes behave under the transformation.

\paragraph*{Recall: Winning condition for fair Büchi/$\bot$ games.}
 We recall that the $\exists$-winning regions for fair Büchi/$\bot$ games are given via~\cref{eq:Bone}, that is, $v \in \Win_\exists(G)$ if and only if
\begin{equation}
    \exists s\in\Sigma.\forall t\in\Pi.\fair_\exists(\play_v(s,t))\wedge
   (\fair_\forall(\play_v(s,t)) \Rightarrow \play_v(s,t)\in \text{B\"{u}chi(T)}) \notag
\end{equation}

According to this, player $\exists$ wins a node $v$ in $G$ iff she has a $\exists$-fair strategy from $v$ that defeats all $\forall$-fair counter-strategies. 
Conversely, player~$\forall$ wins $v$ if either (case I)~there exists a $\forall$-fair strategy that defeats all $\exists$-fair counter-strategies, or (case II)~player~$\exists$ has no $\exists$-fair strategy from $v$.

For nodes $v \notin \Vfair$, the gadgets make no changes, so the local winning conditions are trivially preserved. We therefore focus on fair nodes $v \in \Vfair$.

\noindent \textbf{Case: Player~$\exists$ wins $v \in \Vfair_\exists$.}
Suppose $v \in \Win_\exists \cap \Vfair_\exists$, and let $s \in \Sigma^\fair$ be an $\exists$-fair winning strategy from $v$. We will focus our attention to this node $v$ in the reduced game $G'$.

Consider a play $\rho = \play_v(s,t)$ for some $\forall$-fair strategy $t \in \Pi^\fair$. We distinguish two cases based on how often $v$ is visited in $\rho$:
\begin{description}[font=\normalfont\itshape]
    \item[Case (i): $v$ visited finitely often.] In the reduced game, player~$\exists$ takes the rightmost branch (moving through $v_2 \to v_2^\exists$) and mimics the successor choices from $\rho$. Since $v$ is visited finitely often in $\rho$, priority~$3$ appears finitely often in $G'$ and does not affect the winner. Therefore, player~$\exists$ wins $v$ in $G'$.
    
   \item[Case (ii): $v$ and some priority-$2$ node visited infinitely often.] Player~$\exists$ \\takes the left branch to $v_1$ in the gadget. Since $\rho$ is $\exists$-fair and visits $v$ infinitely often, all fair successors $E_f(v)$ must be visited infinitely often in $\rho$ (otherwise $s$ would not be $\exists$-fair). This means the play cycles back to $v$ regardless of which fair successor is chosen.

     Player~$\exists$ must now win from both $v_1^\exists$ and $v_1^\forall$:
    \begin{itemize}
        \item If player~$\forall$ moves to $v_1^\forall$ (taking a fair successor), the priority sequence is $\lambda(v) \to 1 \to 2 \to \lambda(w)$ for $w \in E_f(v)$. Since priority~$2$ is seen infinitely often, player~$\exists$ wins.
        \item If player~$\forall$ moves to $v_1^\exists$ (returning control), player~$\exists$ mimics the successor choices from $\rho$. Since $\rho$ visits a priority-$2$ node infinitely often, player~$\exists$ wins.
    \end{itemize}
\end{description}

\noindent \textbf{Case: Player~$\exists$ wins $v \in \Vfair_\forall$.}
Suppose $v \in \Win_\exists \cap \Vfair_\forall$, and let $s \in \Sigma^\fair$ be an $\exists$-fair winning strategy. Again, we distinguish cases:

\begin{description}[font=\normalfont\itshape]
    \item[Case (i): $v$ visited finitely often.] There is no constraint on fair edges being taken. Player~$\exists$ wins from all successors of $v$ in $G$, hence also in the gadget:
    \begin{itemize}
        \item If player~$\forall$ takes the right branch to $v_1^\forall$ (choosing any successor), the play either cycles through $v$ (triggering priority~$2$ infinitely often, so $\exists$ wins) or reaches another cycle where $\exists$ wins.
        \item If player~$\forall$ takes the left branch to $v_1^\exists$, player~$\exists$ mimics $\rho$ and ensures $v$ is visited finitely often, winning as in the original game.
    \end{itemize}
    
    \item[Case (ii): $v$ and some priority-$2$ node visited infinitely often.] All fair successors $E_f(v)$ lie on a cycle with $v$:
    \begin{itemize}
        \item If player~$\forall$ moves to $v_1^\forall$, priority~$2$ is visited infinitely often, so player~$\exists$ wins.
        \item If player~$\forall$ moves to $v_1^\exists$, player~$\exists$ mimics $\rho$ and wins.
    \end{itemize}
\end{description}

\noindent \textbf{Case: Player~$\forall$ wins via a $\forall$-fair strategy (case I).}
Now suppose $v \in \Win_\forall$, and player~$\forall$ has a $\forall$-fair winning strategy $t \in \Pi^\fair$. Consider a play $\rho = \play_v(s,t)$ for some $\exists$-fair strategy $s \in \Sigma^\fair$. The argument is dual to the $\exists$-winning case.

For $v \in \Vfair_\exists$, we distinguish:
\begin{description}[font=\normalfont\itshape]
    \item[Case (i): $v$ visited finitely often] 

    \begin{itemize}
        \item If player~$\exists$ moves to $v_1$, player~$\forall$ responds by moving to $v_1^\forall$, mimicking $\rho$ by taking fair successors. Since $v$ is visited finitely often, player~$\forall$ wins.
        \item If player~$\exists$ moves to $v_2$, the play either visits $v$ finitely often (so priority~$3$ appears finitely often and $\forall$ wins) or visits priority~$3$ infinitely often (and $\forall$ wins by parity).
    \end{itemize}
    
    \item[Case (ii): $v$ visited infinitely often] No priority-$2$ nodes are visited infinitely often in $\rho$.
    \begin{itemize}
        \item If player~$\exists$ moves to $v_1$, player~$\forall$ responds by moving to $v_1^\exists$. No priority-$2$ node is visited infinitely often in $\rho$, so there are no priority-$2$ nodes that can be reached by $\exists$ player choices. Therefore, player~$\forall$ wins in $G'$ as in $G$.
        \item If player~$\exists$ moves to $v_2$, priority~$3$ is visited infinitely often, so player~$\forall$ wins by parity.
    \end{itemize}
\end{description}

For $v \in \Vfair_\forall$, player~$\exists$ moves to $v_1$:
\begin{description}[font=\normalfont\itshape]
    \item[Case (i): $v$ visited finitely often] From $v_1$, player~$\forall$ moves to $v_1^\forall$, mimics $\rho$, and wins.
    \item[Case (ii): $v$ visited infinitely often] No priority-$2$ nodes are visited infinitely often in $\rho$. From $v_1$, player~$\forall$ moves to $v_1^\exists$; no priority-$2$ node is visited infinitely often, so player~$\forall$ wins.
\end{description}

\noindent \textbf{Case: Player~$\forall$ wins due to mutually unfair plays (case II).}
Finally, suppose player~$\exists$ has no $\exists$-fair strategy from $v$. In mutually unfair plays, player~$\exists$ must signal priority~$3$ infinitely often (by repeatedly retaining control in $\Vfair_\exists$ gadgets), while player~$\forall$ signals priority~$2$ infinitely often (by retaining control in $\Vfair_\forall$ gadgets). Since the maximum priority is~$3$ (odd), the parity condition ensures player~$\forall$ wins such plays.


\begin{remark}\label{rem:minimal-gadgets}
In the arguments above, observe that whenever the root node $v$ of a gadget has priority $\lambda(v) = 2$ (i.e., $v \in T$), player~$\forall$ has no strategic incentive to take the colored edges in the gadgets of~\cref{fig:buchigadgets-existential}. 

The reason is that player~$\forall$ would only consider these edges if $v$ were visited infinitely often—but in such cases, priority~$2$ is already triggered by $v$ itself, making the right (black) branch strictly more advantageous than the colored branch. Thus, the colored edges become redundant whenever $v \in T$.

Eliminating these redundant edges yields \emph{equivalent minimal gadgets} for fair Büchi/$\bot$ games, as discussed in detail in~\cref{rem:opt}.
\end{remark}


%% file: alternativegadgets.tex
\subsection{Alternative Gadgets}\label{sec:alternative-gadgets}

The parity/$\bot$ reduction presented in the next section employs 3-step gadgets to replace fair nodes, similar to the ones presented in the previous section. However, for each fair node, we provide two equivalent gadget variants. In this section, we explain how to interpret these variants and why both are presented.
We refer to them as \emph{existential} and \emph{universal} gadgets, according to the owner of the root node. These variants differ in who controls the initial branching within the gadget, yet both correctly encode the same fairness constraints.

\paragraph*{The two variants.} For parity/$\bot$ games, the left and right sides of~\cref{fig:existentialgadgetsparitybot} illustrate existential gadgets (where player~$\exists$ controls the root node) and universal gadgets (where player~$\forall$ controls the root node), respectively. Similarly, for Büchi/$\bot$ games, \cref{fig:buchigadgets-existential} and~\cref{fig:buchigadgets-universal} illustrate existential and universal gadgets, respectively. Although we provided the intuitive explanation in~\cref{subsection:reductionbuchi} using existential gadgets, we remark that the universal gadgets in~\cref{fig:buchigadgets-universal} can also be used for the reduction. Its correctness is again a special case of the parity/$\bot$ proof.

Although we present both variants, for the purposes of the reduction it suffices to fix one gadget type and apply it uniformly to all appropriate fair nodes. For instance, one may replace all fair nodes $v \in \Vfair_\exists$ with existential gadgets and all $v \in \Vfair_\forall$ with universal gadgets—or vice versa. More generally, the choice of variant can be made independently for each node, even at random, without affecting correctness.

\paragraph*{Key insight: Separating ownership from control.} The existence of both variants reveals an important structural property: \emph{the owner of a fair node need not be the player who controls the initial branching in the corresponding gadget.} For example, a fair node owned by player~$\exists$ can be replaced with either an existential gadget (where $\exists$ controls branching) or a universal gadget (where $\forall$ controls branching).

This flexibility arises because the gadgets encode fairness through a more subtle mechanism than simple ownership. Both players participate in the branching decisions within the gadget, and the priority structure ensures that fairness obligations are enforced regardless of who makes the first choice. In this sense, the gadgets implement a form of \emph{mutual control}: neither player can unilaterally avoid their fairness obligations, but both have strategic options for how to respond to the other player's behavior.


\paragraph*{Connection to stochastic parity games.} This flexibility has implications for stochastic games. Stochastic parity and Rabin games correspond to $\forall$-fair games in our framework, where random nodes correspond to $\Vfair_\forall$ nodes. In~\cite{CJH03} and~\cite{CLH05}, the authors use universal gadgets to replace random ($\Vfair_\forall$) nodes—that is, the player constrained by fairness controls the first-level branching of the gadget. Our existential gadgets show that this need not be the case: the choice of gadget variant is not dictated by the player constrained by fairness. This provides a deeper understanding of how stochasticity and fairness connect to the interplay between players. In particular, our existential gadgets for parity/$\bot$ games yield an alternative gadget-based reduction for stochastic parity games.


\paragraph*{Parity/parity case.} We conjecture that, analogous to the parity/$\bot$ case, two variants of each gadget can also be constructed for the parity/parity reduction (Section~\ref{subsection:reductionparityparity}). However, to avoid further complicating the presentation, we provide only the existential variant in that setting. Exploring the universal variant for parity/parity games remains an interesting direction for future work.

%% file: reductionofmutuallyfairparity.tex
\subsection{Reduction of Fair Parity/$\bot$ Games}\label{subsection:reductionparitybot}

\leavevmode\par\medskip

In this section, we extend our attention from fair games with B\"uchi conditions to 
fair games with parity conditions. 
We start by considering fair $\alpha / \beta$ games with
$\alpha$ being a parity condition and $\beta$ being $\bot$
and present a linear reduction from such fair parity/$\bot$ games to (normal) parity
games. To this end, let $G=(A, \text{Parity}(\lambda), \bot)$ be
a fair parity/$\bot$ game, where $A = (V_\exists,V_\forall,E,E_f)$ is a fair game arena such that $V = V_\exists \dotcup V_\forall$ and where $\lambda: V \to [2k]$ is a priority function.

The reduction to parity games replaces fair nodes $v \in
\Vfair$ in $G$ with the
gadgets given in~\cref{fig:existentialgadgetsparitybot}. 
Nodes $v \in \Vfair_\exists$ in $G$ are replaced with one of the
gadgets on the top (i.e., the incoming edges to $v$ are redirected to $v$
in the root, and the outgoing edges on the third level lead to
$E(v)$ and $E_f(v)$, which are the outgoing edges and outgoing fair
edges of $v$ in $G$, respectively), and nodes $v \in
\Vfair_\forall$ in $G$ are replaced with one of the gadgets at the bottom.
We refer to the gadgets on the left hand side as \emph{existential} gadgets and the ones on the
right hand side as \emph{universal} gadgets, indicating the player picking the first move. 
Nodes in $\Vn$ are not altered.

\begin{figure} 
	\centering
	\tikzset{every state/.style={minimum size=15pt}, every node/.style={minimum size=15pt}}
	\begin{tikzpicture}[
		auto,
		node distance=0.8cm,
		semithick
		]
		\path[use as bounding box] (-6.2,-3.7) rectangle (6.2,0.6);
  		\node[draw,rounded corners] (vleft) at (-3.3cm,0) {$v$};
		\node[draw,rectangle] (vright) at (3.3cm,0) {$v$};

		\node (yo) [below= 4mm of vleft] {};
		\node (yoleft) [left= 3.5mm of yo] {$\dots$};
		\node (yoright) [right= 3.5mm of yo] {$\dots$};
		\node[rectangle,draw] (r1) [left= 18mm of yo] {$v_1$};
		\node (yo10) [below= 4mm of r1] {};
		\node[draw,rounded corners] (c1) [left=0.2mm of yo10] {$v^\exists_1$};
		\node at ([yshift=-9mm]c1.north west) {\scalebox{0.7}{$1$}};
		\node[rectangle,draw] (r2) [right=0.2mm of yo10] {$v^\forall_1$};
		\node at ([yshift=-9mm]r2.north west) {\scalebox{0.7}{$2$}};
		
		\node[rectangle,draw] (rnm) [below= 4mm of vleft] {$v_i$};
		\node (yon0) [below= 4mm of rnm] {};
		\node[draw,rounded corners] (cnm) [left=0.2mm of yon0] {$v^\exists_i$};
		\node at ([yshift=-9mm]cnm.north west) {\scalebox{0.7}{$2i-1\,\,$}};
		\node[rectangle,draw] (rn) [right= 0.2mm of yon0]  {$v^\forall_i$};
		\node at ([yshift=-9mm]rn.north west) {\scalebox{0.7}{$2i$}};
		\node[rectangle,draw] (rmaxevenm) [right= 18mm of yo] {$v_{k+1}$};
		\node[draw,rounded corners] (cmaxevenm) [below= 3mm of rmaxevenm] {$v^\exists_{k+1}$};
		\node at ([yshift=-9mm]cmaxevenm.north west) {\scalebox{0.7}{$2k+1$}};
		
		\path[->] (vleft) edge node {} (r1);
		\path[->] (r1) edge node {} (c1);
		\path[->] (r1) edge node {} (r2);
		\path[->] (vleft) edge node {} (rnm);
		\path[->] (rnm) edge node {} (cnm);
		\path[->] (rnm) edge node {} (rn);
		\path[->] (vleft) edge node {} (rmaxevenm);
		\path[->] (rmaxevenm) edge node {} (cmaxevenm);
		
		\node (out1) [below= 5mm of c1] {$E(v)$};
		\node (out2) [below= 5mm of r2] {$E_f(v)$};
		\node (outnm) [below= 5mm of cnm] {$E(v)$};
		\node (outn) [below= 5mm of rn] {$E_f(v)$};
		\node (outmaxevenm) [below= 5mm of cmaxevenm] {$E(v)$};
		\path[->] (c1) edge node {} (out1);
		\path[->] (r2) edge node {} (out2);
		\path[->] (cnm) edge node {} (outnm);
		\path[->] (rn) edge node {} (outn);
		\path[->] (cmaxevenm) edge node {} (outmaxevenm);
		
		\node (Ryo) [below= 4mm of vright] {};
		\node (Ryoleft) [left= 3.5mm of Ryo] {$\dots$};
		\node (Ryoright) [right= 3.5mm of Ryo] {$\dots$};
		\node[rounded corners,draw] (Rr1) [left= 18mm of Ryo] {$v_1$};
		\node[draw,rounded corners] (Rc1) [below= 4mm of Rr1] {$v^\exists_1$};
		\node at ([yshift=-9mm]Rc1.north west) {\scalebox{0.7}{$1$}};
		
		\node[rounded corners, draw] (Rrnm) [below= 4mm of vright] {$v_i$};
		\node (Ryon0) [below= 4mm of Rrnm] {};
		\node[rectangle, draw] (Rcnm) [left=0.2mm of Ryon0] {$v^\forall_i$};
		\node at ([yshift=-9mm]Rcnm.north west) {\scalebox{0.7}{$2i-2\,\,$}};
		\node[draw,rounded corners] (Rrn) [right= 0.2mm of Ryon0]  {$v^\exists_i$};
		\node at ([yshift=-9mm]Rrn.north west) {\scalebox{0.7}{$2i-1\,\,$}};
		\node[rounded corners, draw] (Rrmaxevenm) [right= 18mm of Ryo] {$v_{k+1}$};
		\node (Ryomaxevenm0) [below= 4mm of Rrmaxevenm] {};
		\node[draw,rectangle] (Rcmaxevenm) [left=0.2mm of Ryomaxevenm0] {$v^\forall_{k+1}$};
		\node at ([yshift=-9mm]Rcmaxevenm.north west) {\scalebox{0.7}{$2k$}};
		\node [draw,rounded corners]  (Rrmaxeven) [right=0.2mm of Ryomaxevenm0] {$v^\exists_{k+1}$};
		\node at ([yshift=-9mm]Rrmaxeven.north west) {\scalebox{0.7}{$2k+1$}};
		\path[->] (vright) edge node {} (Rr1);
		\path[->] (Rr1) edge node {} (Rc1);
		\path[->] (vright) edge node {} (Rrnm);
		\path[->] (Rrnm) edge node {} (Rcnm);
		\path[->] (Rrnm) edge node {} (Rrn);
		\path[->] (vright) edge node {} (Rrmaxevenm);
		\path[->] (Rrmaxevenm) edge node {} (Rcmaxevenm);
		\path[->] (Rrmaxevenm) edge node {} (Rrmaxeven);
		
		\node (Rout1) [below= 5mm of Rc1] {$E(v)$};
		\node (Routnm) [below= 5mm of Rcnm] {$E_f(v)$};
		\node (Routn) [below= 5mm of Rrn] {$E(v)$};
		\node (Routmaxevenm) [below= 5mm of Rcmaxevenm] {$E_f(v)$};
		\node (Routmaxeven) [below= 5mm of Rrmaxeven] {$E(v)$};
		\path[->] (Rc1) edge node {} (Rout1);
		\path[->] (Rcnm) edge node {} (Routnm);
		\path[->] (Rrn) edge node {} (Routn);
		\path[->] (Rcmaxevenm) edge node {} (Routmaxevenm);
		\path[->] (Rrmaxeven) edge node {} (Routmaxeven);
		
		
	\end{tikzpicture}
	\begin{tikzpicture}[
		auto,
		node distance=0.8cm,
		semithick
		]
		\path[use as bounding box] (-6.2,-3.7) rectangle (6.2,0.6);
  		\node[draw,rounded corners] (vleft) at (-3.3cm,0) {$v$};
		\node[draw,rectangle] (vright) at (3.3cm,0) {$v$};

		\node (yo) [below= 4mm of vleft] {};
		\node (yoleft) [left= 3.5mm of yo] {$\dots$};
		\node (yoright) [right= 3.5mm of yo] {$\dots$};
		\node[rectangle,draw] (r1) [left= 18mm of yo] {$v_1$};
		\node (yo10) [below= 4mm of r1] {};
		\node[draw,rounded corners] (c1) [left=0.2mm of yo10] {$v^\exists_1$};
		\node at ([yshift=-9mm]c1.north west) {\scalebox{0.7}{$1$}};
		\node[rectangle,draw] (r2) [right=0.2mm of yo10] {$v^\forall_1$};
		\node at ([yshift=-9mm]r2.north west) {\scalebox{0.7}{$2$}};
		
		\node[rectangle,draw] (rnm) [below= 4mm of vleft] {$v_i$};
		\node (yon0) [below= 4mm of rnm] {};
		\node[draw,rounded corners] (cnm) [left=0.2mm of yon0] {$v^\exists_i$};
		\node at ([yshift=-9mm]cnm.north west) {\scalebox{0.7}{$2i-1\,\,$}};
		\node[rectangle,draw] (rn) [right= 0.2mm of yon0]  {$v^\forall_i$};
		\node at ([yshift=-9mm]rn.north west) {\scalebox{0.7}{$2i$}};
		\node[rectangle,draw] (rmaxevenm) [right= 18mm of yo] {$v_k$};
		\node (yomaxevenm0) [below= 4mm of rmaxevenm] {};
		\node[draw,rounded corners] (cmaxevenm) [left=0.2mm of yomaxevenm0] {$v^\exists_k$};
		\node at ([yshift=-9mm]cmaxevenm.north west) {\scalebox{0.7}{$2k-1\,\,$}};
		\node [rectangle,draw]  (rmaxeven) [right=0.2mm of yomaxevenm0] {$v^\forall_k$};
		\node at ([yshift=-9mm]rmaxeven.north west) {\scalebox{0.7}{$2k$}};
		
		\path[->] (vleft) edge node {} (r1);
		\path[->] (r1) edge node {} (c1);
		\path[->] (r1) edge node {} (r2);
		\path[->] (vleft) edge node {} (rnm);
		\path[->] (rnm) edge node {} (cnm);
		\path[->] (rnm) edge node {} (rn);
		\path[->] (vleft) edge node {} (rmaxevenm);
		\path[->] (rmaxevenm) edge node {} (cmaxevenm);
		\path[->] (rmaxevenm) edge node {} (rmaxeven);
		
		\node (out1) [below= 5mm of c1] {$E_f(v)$};
		\node (out2) [below= 5mm of r2] {$E(v)$};
		\node (outnm) [below= 5mm of cnm] {$E_f(v)$};
		\node (outn) [below= 5mm of rn] {$E(v)$};
		\node (outmaxevenm) [below= 5mm of cmaxevenm] {$E_f(v)$};
		\node (outmaxeven) [below= 5mm of rmaxeven] {$E(v)$};
		\path[->] (c1) edge node {} (out1);
		\path[->] (r2) edge node {} (out2);
		\path[->] (cnm) edge node {} (outnm);
		\path[->] (rn) edge node {} (outn);
		\path[->] (cmaxevenm) edge node {} (outmaxevenm);
		\path[->] (rmaxeven) edge node {} (outmaxeven);
		
		\node (Ryo) [below= 4mm of vright] {};
		\node (Ryoleft) [left= 3.5mm of Ryo] {$\dots$};
		\node (Ryoright) [right= 3.5mm of Ryo] {$\dots$};
		\node[rounded corners,draw] (Rr1) [left= 18mm of Ryo] {$v_1$};
		\node[draw,rounded corners] (Rc1) [below= 4mm of Rr1] {$v^\exists_1$};
		\node at ([yshift=-9mm]Rc1.north west) {\scalebox{0.7}{$1$}};
		
		\node[rounded corners, draw] (Rrnm) [below= 4mm of vright] {$v_i$};
		\node (Ryon0) [below= 4mm of Rrnm] {};
		\node[rectangle, draw] (Rcnm) [left=0.2mm of Ryon0] {$v^\forall_i$};
		\node at ([yshift=-9mm]Rcnm.north west) {\scalebox{0.7}{$2i-2\,\,$}};
		\node[draw,rounded corners] (Rrn) [right= 0.2mm of Ryon0]  {$v^\exists_i$};
		\node at ([yshift=-9mm]Rrn.north west) {\scalebox{0.7}{$2i-1\,\,$}};
		\node[rounded corners, draw] (Rrmaxevenm) [right= 18mm of Ryo] {$v_{k+1}$};
		\node[draw,rectangle] (Rcmaxevenm) [below= 3mm of Rrmaxevenm] {$v^\forall_{k+1}$};
		\node at ([yshift=-9mm]Rcmaxevenm.north west) {\scalebox{0.7}{$2k$}};
		%
		\path[->] (vright) edge node {} (Rr1);
		\path[->] (Rr1) edge node {} (Rc1);
		\path[->] (vright) edge node {} (Rrnm);
		\path[->] (Rrnm) edge node {} (Rcnm);
		\path[->] (Rrnm) edge node {} (Rrn);
		\path[->] (vright) edge node {} (Rrmaxevenm);
		\path[->] (Rrmaxevenm) edge node {} (Rcmaxevenm);
		
		\node (Rout1) [below= 5mm of Rc1] {$E_f(v)$};
		\node (Routnm) [below= 5mm of Rcnm] {$E(v)$};
		\node (Routn) [below= 5mm of Rrn] {$E_f(v)$};
		\node (Routmaxevenm) [below= 5mm of Rcmaxevenm] {$E(v)$};
		\path[->] (Rc1) edge node {} (Rout1);
		\path[->] (Rcnm) edge node {} (Routnm);
		\path[->] (Rrn) edge node {} (Routn);
		\path[->] (Rcmaxevenm) edge node {} (Routmaxevenm);
		
		
	\end{tikzpicture}\caption{\small Existential (left) and universal (right) gadgets for $v \in \Vfair_\exists$ (top) and $v \in \Vfair_\forall$ (bottom) in fair parity/$\bot$ games.
		For $i \in [k+1]$, priorities of nodes $v_i^\exists$ and $v_i^\forall$ are given below them, priorities of nodes $v_i$ are ignored, and the priority of $v$ is unaltered.}\label{fig:existentialgadgetsparitybot}
\end{figure} 


Our correctness proof for the reduction shows correctness for all combinations of the gadgets (i.e., one can replace each $v \in \Vfair_\exists$ ($v \in \Vfair_\forall$) with each of the gadgets on the top (bottom)), we provide intuition only for the existential gadgets in order to avoid repetition.

Assume that all nodes $v \in \Vfair$ are replaced with their existential gadgets. Within a subgame that starts at a fair node $v\in\Vfair$, the two players intuitively
interact as follows. The $\exists$-player gets to pick a number $i$,
indicating the priorities ($2i-1$ or $2i$) they intend to visit infinitely often in every play that visits $v$ 
infinitely often. In turn, $\forall$-player gets to either
pick an outgoing edge at $v$ (for this, she pays the price of visiting the even priority $2i$), or
allow $\exists$ to pick an outgoing edge
(in which case she is rewarded with a
visit to the odd priority $2i-1$). Depending on the owner
of $v$, the edge picked by $\forall$ (if $v\in \Vfair_\exists$),
or the edge picked by $\exists$ (if $v\in \Vfair_\forall$) is required to be a fair edge, that is, to be contained in $E_f$.
Thus player $\forall$ can insist on exploring fair edges
at $\Vfair_\exists$ nodes, but has to pay a price for it;
dually, player $\forall$ eventually has to allow player $\exists$ to explore the fair
edges at $\Vfair_\forall$ nodes to win.

In the reduced game defined formally in the proof of~\cref{thm:fairParityToParity} below, the owner of a fair node $v$ can win \emph{fairly} from $v$
by either avoiding 
$v$ from some point on forever, or eventually allowing the opponent player to explore all fair edges
leading out of that node. The owner wins by playing \emph{unfairly} if and only if the opponent also plays unfairly and the owner is the $\forall$-player.

\begin{example}
	Recall the fair parity/$\beta$ games $G_1$ to $G_3$ from~\cref{fig:fairstrategies},
	and recall that $\beta$ is irrelevant in these example games due to the absence of mutually unfair plays.
	We apply the gadget construction
	from~\cref{fig:existentialgadgetsparitybot} to these games (treating them as parity/$\bot$ games),
	and obtain games $G'_1$ to $G'_3$ without fairness constraints as depicted in~\cref{fig:gadgetsexample}.
	We depict game nodes from the original game with gray background, and newly added gadget nodes
	with white background, omitting the notation for the gadget nodes.
	\begin{figure}
		\centering
		\tikzset{every state/.style={minimum size=10pt}}
		\begin{tikzpicture}[
			auto,
			node distance=0.8cm,
			semithick,
			extends/.style={->}
			]
			\node[draw, rectangle, fill=gray!30] (n3f) {$3$}	
			coordinate[right=1.35cm of n3f.east] (m);
			\node[draw, rounded corners] (n4fa) [below= 7mm of n3f] {\phantom{a}};
			\node[draw, rounded corners] (n4fb) [left= 7mm of n4fa] {\phantom{a}};
			\node[draw, rounded corners] (n4fc) [right= 7mm of n4fa] {\phantom{a}};
  		    \node[draw, rounded corners] (n4f1) [below= 7mm of n4fb] {$1$};
			\node[draw, rectangle] (n4f2) [right= 3.5mm of n4f1] {$2$}
			coordinate[below=0.35cm of n4f2.south] (m1);
			\node[draw, rounded corners] (n4f3) [right= 3.5mm of n4f2] {$3$};
			\node[draw, rectangle] (n4f4) [below= 7mm of n4fc] {$4$}
			coordinate[below=0.35cm of n4f4.south] (m2)
			coordinate[right=0.35cm of n4f4.east] (n);
			\node[draw, rounded corners, fill=gray!30] (n4f) [right= 7mm of n4f4] {$4$}
			coordinate[below=0.35cm of n4f.south] (m3);
			\path[->] (n4f2) edge [bend left=15] node [pos=0.3,below] {} (n3f);
			\path[->] (n3f) edge node [pos=0.3,below] {} (n4fa);
			\path[->] (n3f) edge node [pos=0.3,below] {} (n4fb);
			\path[->] (n3f) edge node [pos=0.3,below] {} (n4fc);
			\path[->,red,very thick] (n4fb) edge node [pos=0.3,below] {} (n4f1);
			\path[->] (n4fa) edge node [pos=0.3,below] {} (n4f2);
			\path[->,red,very thick] (n4fa) edge node [pos=0.3,below] {} (n4f3);
			\path[->,red,very thick] (n4fc) edge node [pos=0.3,below] {} (n4f4);
			\draw (n4f4.east) |- (n);
			\draw[red,very thick] (n4f.west) |- (n);
			\draw[red,very thick] (n) |- (m);
			\draw[red,very thick,extends] (m) -- (n3f.east);
			\draw[red,very thick] (n4f1.south) |- (m1);
			\draw (n4f2.south) |- (m1);
			\draw[red,very thick] (n4f3.south) |- (m2);
			\draw (n4f4.south) |- (m2);
			\draw[red,very thick] (m1) |- (m2);
			\draw[red,very thick] (m2) |- (m3);
			\draw[red,very thick,extends] (m3) -- (n4f.south);
			\node (label1) [left=15mm of n3f] {$G'_1$:};

			\node[draw, rectangle, fill=gray!30] (n3) [below= 35mm of n3f] {$3$}
			coordinate[right=1.35cm of n3.east] (my);
			\node (label2) [left=16mm of n3] {$G'_2$:};
			\node[draw, rounded corners] (n4ya) [below= 7mm of n3] {\phantom{a}};
			\node[draw, rounded corners] (n4yb) [left= 7mm of n4ya] {\phantom{a}};
			\node[draw, rounded corners] (n4yc) [right= 7mm of n4ya] {\phantom{a}};
			\node[draw, rounded corners] (n4y1) [below= 7mm of n4yb] {$1$};
			\node[draw, rectangle] (n4y2) [right= 3.5mm of n4y1] {$2$}
			coordinate[below=0.9cm of n4y2.west] (my5)
			coordinate[below=0.35cm of n4y2.south] (my1);
			\node (iiix) [below= -0.3mm of n4y2.west,xshift=1mm] {};
			\node[draw, rounded corners] (n4y3) [right= 3.5mm of n4y2] {$3$};
			\node[draw, rectangle] (n4y4) [below= 7mm of n4yc] {$4$}
			coordinate[below=0.35cm of n4y4.south] (my2)
			coordinate[below=0.9cm of n4y4.west] (my6)
			coordinate[right=0.35cm of n4y4.east] (ny);
			\node (iiiy) [below= -0.3mm of n4y4.west,xshift=1mm] {};
			\node[draw, rounded corners, fill=gray!30] (n4) [right= 7mm of n4y4] {$4$}
			coordinate[below=0.35cm of n4.south] (my3);
			\node[draw, rectangle, fill=gray!30] (n1) [left= 7mm of n4y1] {$1$}
			coordinate[left=1.65cm of my5] (my4);
			\node (empty2) [below=0.5 cm of n4] {};
			\path[->] (n4y2) edge [bend left=15] node [pos=0.3,below] {} (n3);
			\path[->] (n3) edge node [pos=0.3,below] {} (n4ya);
			\path[->] (n3) edge node [pos=0.3,below] {} (n4yb);
			\path[->,blue, very thick] (n3) edge node [pos=0.3,below] {} (n4yc);
			\path[->] (n4yb) edge node [pos=0.3,below] {} (n4y1);
			\path[->] (n4ya) edge node [pos=0.3,below] {} (n4y2);
			\path[->] (n4ya) edge node [pos=0.3,below] {} (n4y3);
			\path[->] (n4yc) edge node [pos=0.3,below] {} (n4y4);
			\draw (n4y1.south) |- (my1);
			\draw (n4y2.south) |- (my1);
			\draw (n4y3.south) |- (my2);
			\draw (n4y4.south) |- (my2);
			\draw (my1) |- (my2);
			\draw (my2) |- (my3);
			\draw[extends] (my3) -- (n4.south);
			\draw[blue, very thick] (iiix) |- (my5);
			\draw[blue, very thick] (iiiy) |- (my6);
			\draw[blue, very thick] (my5) |- (my6);
			\draw[blue, very thick] (my6) |- (my4);
			\draw[blue, very thick, extends] (my4) -- (n1.south);
			\path[->,blue, very thick] (n1) edge [loop right] node [below] {} (n1);
			\draw (n4y4.east) |- (ny);
			\draw (n4.west) |- (ny);
			\draw (ny) |- (my);
			\draw[extends] (my) -- (n3.east);
			\path[->] (n1) edge [bend left=30] node [pos=0.3,below] {} (n3);

\node[draw, rounded corners, fill=gray!30] (n3z) [right= 6cm of n3f] {$4$}
coordinate[right=2.5cm of n3z.east] (yar)
coordinate[right=1.35cm of n3z.east] (mz);
\node[draw, rectangle] (n4za) [below= 7mm of n3z] {\phantom{a}};
\node[draw, rectangle] (n4zb) [left= 7mm of n4za] {\phantom{a}};
\node[draw, rectangle] (n4zc) [right= 7mm of n4za] {\phantom{a}};
\node[draw, rounded corners] (n4z5) [below= 7mm of n4zc] {$5$};
\node (iz) [below= -0.3mm of n4z5.east,xshift=-1mm] {};
\node[draw, rounded corners] (n4z1) [left= 23mm of n4z5] {$1$}
coordinate[below=0.9cm of n4z1.east] (yo1);
\node (ix) [below= -0.3mm of n4z1.east,xshift=-1mm] {};
\node[draw, rectangle] (n4z2) [left= 16mm of n4z5] {$2$}
coordinate[below=0.35cm of n4z2.south] (mz1);
\node[draw, rounded corners] (n4z3) [left= 10mm of n4z5] {$3$};
\node (yo3) [below= -0.3mm of n4z3.east,xshift=-1mm] {};
\node (iy) [below= -0.3mm of n4z3.east,xshift=-1mm] {};
\node[draw, rectangle] (n4z4) [left= 3mm of n4z5] {$4$}
coordinate[below=0.35cm of n4z4.south] (mz2)
coordinate[right=0.35cm of n4z4.east] (nz);
\node[draw, rounded corners, fill=gray!30] (n4z) [right= 7mm of n4z5] {$5$}
coordinate[right=0.3cm of n4z.east] (ya)
coordinate[below=0.69cm of iz] (yo5)
coordinate[below=0.35cm of n4z.south] (mz3);
\draw (n4z.east) |- (ya);
\path[->] (n3z) edge node [pos=0.3,below] {} (n4za);
\path[->] (n3z) edge node [pos=0.3,below] {} (n4zb);
\path[->] (n3z) edge node [pos=0.3,below] {} (n4zc);
\path[->] (n4zb) edge node [pos=0.3,below] {} (n4z1);
\path[->,blue, very thick] (n4zb) edge node [pos=0.3,below] {} (n4z2);
\path[->] (n4za) edge node [pos=0.3,below] {} (n4z3);
\path[->,blue, very thick] (n4za) edge node [pos=0.3,below] {} (n4z4);
\path[->,blue, very thick] (n4zc) edge node [pos=0.3,below] {} (n4z5);
\draw (ix) |- (yo1);
\draw (iy) |- (yo3);
\draw (iz) |- (yo5);
\draw (yo1) |- (yo5);
\draw (n4z1.south) |- (mz1);
\draw[blue, very thick] (n4z2.south) |- (mz1);
\draw (n4z3.south) |- (mz2);
\draw[blue, very thick] (n4z4.south) |- (mz2);
\draw (n4z5.south) |- (mz2);
\draw[blue, very thick] (mz1) |- (mz2);
\draw[blue, very thick] (mz2) |- (mz3);
\draw[blue, very thick, extends] (mz3) -- (n4z.south);
\node (label1) [left=19mm of n3z] {$G'_3$:};

\node[draw, rectangle,xshift=-2mm, fill=gray!30] (n3x) [below= 35mm of n3z] {$3$}
coordinate[right=1.35cm of n3x.east] (mx)
coordinate[left=0.3cm of n3x] (yep)
coordinate[below=0.7cm of n3x] (yiz);
\draw[extends] (yo3) -- (n3x.north);
\node[draw, rounded corners] (n4xa) [left= 6mm of yiz] {\phantom{a}};
\node[draw, rounded corners] (n4xb) [left= 7mm of n4xa] {\phantom{a}};
\node[draw, rounded corners] (n4xc) [right= 6mm of yiz] {\phantom{a}};
\node[draw, rounded corners] (n4xd) [right= 7mm of n4xc] {\phantom{a}};
\node[draw, rounded corners] (n4x1) [below= 7mm of n4xb] {$1$};
\node[draw, rectangle] (n4x2) [right= 3.5mm of n4x1] {$2$}
coordinate[below=0.9cm of n4x2.west] (mx5)
coordinate[below=0.35cm of n4x2.south] (mx1);
\node (iix) [below= -0.3mm of n4x2.west,xshift=1mm] {};
\node[draw, rounded corners] (n4x3) [right= 3.5mm of n4x2] {$3$};
\node[draw, rectangle] (n4x4) [right= 3.5mm of n4x3] {$4$}
coordinate[below=0.9cm of n4x4.west] (mx7)
coordinate[below=0.35cm of n4x4.south] (mx0);
\node (iiy) [below= -0.3mm of n4x4.west,xshift=1mm] {};
\node[draw, rounded corners] (n4x5) [right= 3.5mm of n4x4] {$5$};
\node[draw, rectangle] (n4x6) [below= 7mm of n4xd] {$6$}
coordinate[below=0.35cm of n4x6.south] (mx2)
coordinate[right=1cm of mx2] (mx3)
coordinate[below=0.9cm of n4x4.west] (mx6)
coordinate[below=0.9cm of n4x6.west] (mx8)
coordinate[right=0.35cm of n4x4.east] (nx)
coordinate[left=1.2cm of mx5] (mx4);
\node (iiz) [below= -0.3mm of n4x6.west,xshift=1mm] {};
\path[->] (n3x) edge node [pos=0.3,below] {} (n4xa);
\path[->,blue, very thick] (n3x) edge node [pos=0.3,below] {} (n4xb);
\path[->] (n3x) edge node [pos=0.3,below] {} (n4xc);
\path[->] (n3x) edge node [pos=0.3,below] {} (n4xd);
\path[->] (n4xb) edge node [pos=0.3,below] {} (n4x1);
\path[->] (n4xa) edge node [pos=0.3,below] {} (n4x2);
\path[->] (n4xa) edge node [pos=0.3,below] {} (n4x3);
\path[->] (n4xc) edge node [pos=0.3,below] {} (n4x4);
\path[->] (n4xc) edge node [pos=0.3,below] {} (n4x5);
\path[->] (n4xd) edge node [pos=0.3,below] {} (n4x6);
\draw (n4x1.south) |- (mx1);
\draw (n4x2.south) |- (mx1);
\draw (n4x3.south) |- (mx0);
\draw (n4x4.south) |- (mx0);
\draw (n4x5.south) |- (mx2);
\draw (n4x6.south) |- (mx2);
\draw (mx1) |- (mx2);
\draw (mx2) |- (mx3);
\draw (mx3) |- (yar);
\draw[extends] (yar) -- (n3z.east);
\draw[blue, very thick] (iix) |- (mx5);
\draw[blue, very thick] (iiy) |- (mx6);
\draw[blue, very thick] (iiz) |- (mx8);
\draw[blue, very thick] (mx5) |- (mx6);
\draw[blue, very thick] (mx6) |- (mx4);
\draw[blue, very thick] (mx8) |- (mx4);
\draw[blue, very thick] (mx4) |- (yep);
\draw[blue, very thick,extends] (yep) -- (n3x.west);

		\end{tikzpicture}
		\caption{\small 
			Reduced games $G'_1$ to $G'_3$ without fairness constraints,
			obtained by applying the gadget construction to
			the fair example games $G_1$ to $G_3$ from~\cref{fig:fairstrategies}.
		}
		\label{fig:gadgetsexample}
	\end{figure}

The game $G'_1$ contains one gadget that is used to replace the $\forall$-fair node with priority
$3$ in $G_1$; in this example, we use the universal gadget.
Player $\forall$ can attempt to escape from this fair node by taking the right
branch in the gadget, which allows him to move to either successor of the fair node, 
but comes at the cost of visiting priority $4$. Since the game always cycles
back to the fair node, this escaping is a losing strategy for $\forall$.
Alternatively, $\forall$-player can let $\exists$-player pick a fair edge
by taking the left branch in the gadget; this however leads to the game node with priority
$4$ and also is a losing strategy for $\forall$. The remaining choice
is to let $\exists$-player decide the fairness by taking the middle branch in the gadget.
Then $\exists$-player can react to this by insisting to explore a fair edge
and paying the price of visiting priority $3$, which however allows her to visit the
game node with priority $4$, and hence is a winning strategy for $\exists$-player.
We denote this strategy by thicker red edges.

Game $G'_2$ again contains one universal gadget that is used to replace the $\forall$-fair node in $G_2$.
In this case, $\forall$-player can successfully escape from this fair node by taking the right
branch in the gadget, which allows him to move to either successor of the fair node, 
but comes at the cost of visiting priority $4$. The escape strategy then picks the non-fair transition in $G_2$
that leads from node $3$ to $1$. From this node, $\forall$-player can avoid the fair game node forever:
as the node $1$ in $G_2$ is not a fair node, $\forall$-player can simply take the looping transition at this node. The corresponding strategy is marked with thicker blue edges.

Finally, game $G'_3$ uses two gadgets, one being existential (replacing the $\exists$-fair node $4$) and the other one universal (replacing the $\forall$-fair node $3$).
In this case, player $\exists$ could 
take one of the left two branches in the gadget
for the $\exists$-fair node $4$. In both of these branches however, $\forall$-player can react
by insisting to explore a fair edge, and then taking the edge to the right-most game node in $G_3$, thereby
visiting priority $5$. Thus $\exists$-player cannot win by taking the left two branches in this gadget.
Alternatively, $\exists$-player can attempt to escape from the $\exists$-fair node $4$ by taking the right branch in the associated gadget.
This comes at the cost of visiting priority $5$ but gives her the choice
to either loop back to the $\exists$-fair node
via game node $5$,
or to progress the game to the $\forall$-fair node $3$.
The former is a losing strategy for $\exists$-player since it results in visiting priority $5$ infinitely often.
For the latter choice, it remains to analyse the possible strategies for the $\forall$-fair node $3$.
However, in the gadget for this node, $\forall$-player can simply take the left-most branch and force $\exists$-player to explore the only $\forall$-fair edge that leads back to the $\exists$-fair node $4$.
Therefore, intuitively, $\forall$-player wins by either exploring the $\exists$-fair edge 
infinitely often, or by having $\exists$-player escape from her fair node infinitely often. Again, we depict the associated overall winning strategy for $\forall$-player by thicker blue edges.

We point out that the winning regions in game $G_i$ are equal to the winning regions in game $G'_i$ restricted to the non-gadget nodes (for $1\leq i\leq 3$).
\end{example}

Next, we show the correctness of this reduction.

\begin{restatable}{theorem}{restatablefairParitytoParity}\label{thm:fairParityToParity} Let $G=(A,\text{Parity}(\lambda),\bot)$ be a fair parity / $\bot$ game, where $A = (V_\exists,V_\forall,E,E_f)$ is a fair game arena and  $\lambda: V \to [2k]$ is a priority function.
Then there is a parity game $G'$ with at most $n(3k+3)$ nodes and at most
$2k+1$ priorities such that $G'$ and $G$ are equivalent.
\end{restatable}

We fix $\max_\mathsf{\even} := 2k$ and $\max_{\odd} := 2k+1$.

\smallskip

  \noindent \textit{Proof. }
  Before we start the proof, we will formally define the reduced parity game $G'$ that uses \emph{only the existential gadgets}.
  That is, we assume that $G'$ uses the top left gadget in~\cref{fig:existentialgadgetsparitybot} to replace every $v \in \Vfair_\exists$ and the top right gadget to replace every  $v \in \Vfair_\forall$ in $G$.

Then $G'$ is formally defined as follows:
    \begin{align*}
    G'=(V'_\exists,V'_\forall,E',\Omega:V'\to[2k+1])
    \end{align*}
    where $V' = V'_\exists \dotcup V'_\forall$ by putting
    \begin{align*}
    V'_\exists=&\{v^\exists_i \mid i \in [k+1] \text{ for } v\in \Vfair_\exists \,\, \text{and }\, i \in [k] \,\text{ for } \, v\in \Vfair_\forall \} \, \cup \, \Vfair \cup V_\exists\\
    V'_\forall=&\{v_i \mid i \in [k+1]\, \text{ for } \,v\in \Vfair_\exists  \,\,\text{and }\, i \in [k] \text{ for } v\in \Vfair_\forall \}\\
	& \cup \, \{v^\forall_i\mid i \in [2k] \text{ for } v\in \Vfair\} \, \cup \, (V_\forall \setminus \Vfair)
    \end{align*}
    and
    \begin{align*}
    E'(v)=\begin{cases}
        E(v) & v\notin \Vfair \\
    \{v_i \mid i \in [k+1]\} & v\in  \Vfair_\exists \\
    \{v_i \mid i \in [k]\} & v\in  \Vfair_\forall\\
    \end{cases}, \quad \quad
      E'(v_i^\forall)=\begin{cases}
         \{v_i^\exists, v_i^\forall\} & i \neq k+1 \\
         \{ v_i^\exists\} & i = k+1\\
      \end{cases}
    \end{align*}
	\begin{align*}
	E'(v_i^p) &=
	\begin{cases}
	E(v)   & \Leftrightarrow v \in \Vfair_p \\
	E_f(v) & \Leftrightarrow v \in \Vfair_{1-p}
	\end{cases},
	\quad \quad
	&
	\Omega(v) &=
	\begin{cases}
	\lambda(v) & v \in V \setminus V' \\
	2i-1 & v = u_i^\exists \\
	2i   & v = u_i^\forall
	\end{cases}
	\end{align*}
	where we do not care about the priorities of $v_i$, so we can set simply them to the lowest priority $1$.

	Even though we have provided the formal definition of $G'$ that uses only the existential gadgets, the proof applies to any $G'$ obtained through an arbitrary combination of gadgets. The formal definition of each such gadget can be derived straightforwardly from the descriptions in~\cref{fig:existentialgadgetsparitybot}.

	Let $G'=(V'_\exists,V'_\forall,E',\Omega:V'\to[2k+1])$ be the parity game obtained by replacing the fair nodes in $G$ with an arbitrary combination of their corresponding existential and universal gadgets in~\cref{fig:existentialgadgetsparitybot}.

 Let $V' = V \cup V^{\mathsf{gad}}$ where $V$ is the set of nodes of $G$ and where $V^{\mathsf{gad}}$ (the \emph{gadget nodes}) is the set of nodes coming from the gadgets.
Note that the maximum priority in $G'$ is $\max_\odd = 2k+1$ which only occurs at some gadget nodes.
The maximum even priority in $G'$ is $\max_\even = 2k$ which can come from both $V^{\mathsf{gad}}$ and $V$.
It is easy to see that $|V'| \leq n(3k+3)$ and that $G'$ uses priorities $[2k+1]$.
To prove equivalence of $G$ and $G'$, we recall that the winning regions for fair parity/$\bot$ games are given via~\cref{eq:Bone}, that is, $v \in \Win_\exists(G)$ if and only if
\begin{equation}
    \exists s\in\Sigma.\forall t\in\Pi.\fair_\exists(\play_v(s,t))\wedge
   (\fair_\forall(\play_v(s,t)) \Rightarrow \play_v(s,t)\in\alpha).\tag{\ref{eq:Bone}}\label{eq:Bone-here}
\end{equation}


\noindent\textbf{($\Rightarrow$)} We first show that $v \in \Win_\exists(G') \cap V$ implies $v \in \Win_\exists(G)$. To do so, we fix a (positional) winning $\exists$-strategy $s'$ in $G'$ and construct an $\exists$-strategy $s$ in $G$ 
    such that $s$ is $\exists$-winning in $G$, that is, such that $s$ satisfies~\cref{eq:Bone}. To this end, it suffices to show that for all plays $\rho$ in $G$ that start from $v$ and are compliant with $s$,~\cref{eq:B1-s} holds.
    \begin{equation}
   \fair_\exists\,(\rho)\wedge
      (\fair_\forall(\rho) \Rightarrow \rho\in\alpha)\tag{\ref{eq:Bone}-s}\label{eq:B1-s} \end{equation}
  
    For this we show the two parts of the conjunction separately. We show (i) $\fair_\exists(\rho)$, that is, $s \in \Sigma^\fair$, and
    (ii) $\fair_\forall(\rho) \Rightarrow
    \rho\in\alpha$, that is, every $\forall$-fair play that is compliant with
    $s$ is $\exists$-winning with respect to the parity condition.
  
    \smallskip
    \noindent \textbf{Construction of the $s'$-subgame $G'_{s'}$: } Let $s'$ be a positional $\exists$-strategy winning every play from $v$ in $G'$.
    We denote the subgame of $G'$ in which $\exists$ nodes have only the outgoing edges $u \to s'(u)$ by $G'_{s'}$, and call it the \emph{$s'$-subgame}.
    Recall that all plays that start from $v$ in $G'_{s'}$ are $\exists$-winning.

    \smallskip
    \noindent \textbf{Notation of $n_u$ and $succ(u)$:} For the existential gadgets for both $u\in\Vfair_\exists$ and $u\in\Vfair_\forall$, we refer by $n_u$ to the index of the unique successor of $u$ in $G'_{s'}$. 
    That is, $s'(u) = u_{n_u}$. For the same gadgets, we denote $s'(u_{n_u}^\exists)$ by $succ(u)$. 
    For the universal gadgets for both $u\in\Vfair_\exists$ and $u\in\Vfair_\forall$, we let $n_u$ denote the index of the rightmost child of $u$ that is sent to its right child by $s'$. That is, 
    $n_u$ is the largest index $i$ such that $s'(u_i) = u_i^\exists$. For the same gadgets, we denote $s'(u_{n_u}^\exists)$ by $succ(u)$. 

    \smallskip
    \noindent \textbf{Construction of $s$: } We define $s: V^*\cdot
    V_\exists \to V$ as follows. For $u \in \Vfair_\exists$ there are two cases.
    \begin{inparaenum} \item If $n_u = k+1$, then we put $s(u)= succ(u)$. \item Otherwise, define $s(u)$ to cycle through the set $\{succ(u), E_f(u)\}$ starting at $succ(u)$. 
    \end{inparaenum}
    For $u \in V_\exists \setminus \Vfair_\exists$, we put $s(u)=s'(u)$.

    \smallskip
    \noindent \textbf{Constraining $G'_{s'}$ with $n_u$: } We constrain $G'_{s'}$ to its subgame by limiting the choices of $\forall$-player for a node $u$ replaced by the universal gadget.
    For every universal gadget encountered in $G'_{s'}$, we limit the choices of $u \in \Vfair_\forall$ to only $u \to u_{n_u}$ and $u \to u_{n_u +1}$ (if it exists). So, we remove all the other branches of the gadget for $u$ in $G'_{s'}$. We refer to the
    remaining game as $LG'_{s'}$, standing for \emph{limited} $G'_{s'}$. 
    Note that as $LG'_{s'}$ is a subgame of $G'_{s'}$, it is still $\exists$-winning.

    \smallskip
    \noindent \textbf{$\exists$-expansion: } Let $\rho= u^1 u^2 \ldots$ be some play in $G$ that is compliant with $s$. We define a play $\rho'$, referred to as the \emph{$\exists$-expansion of $\rho$}, as follows: define $\rho'$ to be the play on $LG'_{s'}$ that follows $\rho$ while `prioritising existential nodes on $succ(u)$'. 
    What is meant by this is, for a node $u^i \in \Vfair$, if $u^{i+1} = succ(u^i)$, then $\rho'$ takes the unique branch in $LG'_{s'}$ that leads to $u^{i+1}$ while passing through an existential node ${(u^i)}_j^\exists$. That is, 
    regardless of which gadget $u^i$ is replaced by, $\rho'$ takes the branch 
    \begin{equation} u^i \to u^i_{n_{u^i}} \to {(u^i)}^\exists_{n_{u^i}} \to succ(u^i) = u^{i+1} \tag{branch 1} \label{eq:branch1}
    \end{equation}
    On the other hand if $u^{i+1} \neq succ(u^i)$, then $\rho'$ takes the only other branch in $LG'_{s'}$, that is (branch 2) is defined as follows.
    \begin{enumerate}
        \item If $u^i \in \Vfair$ is replaced by an existential gadget, then (branch 2)
        is $u^i \to u^i_{n_{u^i}} \to {(u^i)}^\forall_{n_{u^i}} \to u^{i+1}$; \label{enum:existsgadgetunivbranch}
        \item if $u^i \in \Vfair$ is replaced by a universal gadget, then (branch 2)
        is $u^i \to u^i_{n_{u^i}+1} \to {(u^i)}^\forall_{n_{u^i}+1} \to u^{i+1}$. \label{enum:univgadgetunivbranch}
    \end{enumerate}

    Note that these branches include all possible transitions in $\rho$:
        all the successors of a node $u^i\in\Vfair_\forall$ are covered by one of the branches since (branch 2) leads to the universal node ${(u^i)}^\forall_{n_{u^i}}$ or ${(u^i)}^\forall_{n_{u^i}+1}$ from where every successor of $u^i$ can be picked;
        all the successors of a node $u^i\in\Vfair_\exists$ are covered by one of the branches, since by construction of $s$, all the successors of $u^i$ in $\rho$ are contained in the set $\{succ(u^i)\} \cup E_f(u^i)$, where (branch 1) covers the $succ(u^i)$ successors, and (branch 2) covers the $E_f(u^i)$ successors since in this case every fair successor of $u^i$ can be picked from the universal node ${(u^i)}^\forall_{n_{u^i}}$ or ${(u^i)}^\forall_{n_{u^i}+1}$.
        
    For $u^i \not \in \Vfair$, $\rho'$ just takes the edge $u^i \to u^{i+1}$. 
     
    \smallskip
    So $\rho'$ is a well-defined play in $LG'_{s'}$ that
    starts from $v$ and is $\exists$-winning.
    We observe that if we remove the gadget nodes from $\rho'$, we get $\rho$. That is, 
    we have $\rho' \mid_{V}=\rho$ where $\rho' \mid_{V}$ denotes the restriction of $\rho'$ to $V$.

    \smallskip
	Now we are ready to show that $\fair_\exists(\rho)$ and $\fair_\forall(\rho) \Rightarrow \rho \in \alpha$ for any play $\rho$ that is compliant with $s$.
	
    \noindent \textbf{(i) $\fair_\exists(\rho)$: } We observe that for every $\rho$ in $G$ that is compliant with $s$, by construction of $s$, the only nodes $u \in \Vfair_\exists$ that $\rho$ may not be fair on, are those for which $n_u = k+1$. So we only need to show that such nodes are visited finitely often in $\rho$.
	Let $\rho'$ be the $\exists$-expansion of $\rho$. 
    Since $\rho'|_V = \rho$, it suffices to show that such nodes $u$ do 
    not occur infinitely often in $\rho'$. 
    For if such a node $u$ is visited infinitely often by $\rho'$, then regardless of the gadget that $u$ is replaced with, the (branch 1) 
    $u \to u_{k+1} \to u_{k+1}^\exists$ is taken infinitely often, signalling the maximum priority $2k+1$. Then $\rho'$ is won by $\forall$-player, a contradiction.
    We conclude that $\rho$ is $\exists$-fair.

    \smallskip
    \noindent \textbf{(ii) $\fair_\forall(\rho) \Rightarrow \rho \in \alpha$: }
    Let $\rho$ be $\forall$-fair and consider the $\exists$-expansion
    $\rho'$ of $\rho$.
    Let $m$ be the largest priority in $\Inf(\rho')$. Then $m$ is even
    since player $\exists$ wins $\rho'$.
    Due to $\rho'\mid_V = \rho$, it suffices to show the existence of a node $u \in \Inf(\rho'\mid_V)$ that has priority $m$.
    This implies that the maximum priority in $\Inf(\rho)$ is $m$, and thus $\rho$ is $\exists$-winning.
  
    Towards a contradiction, assume that the priority $m$ appears only in $ V^{\mathsf{gad}} \cap \Inf(\rho')$. 
    Now let $F$ be the subgame of $LG'_{s'}$ that consists of nodes and edges taken infinitely often in $\rho'$. Then priority $m$ appears in $V^\mathsf{gad} \cap F$. 
    These gadget nodes exist in $F$ due to nodes 
    \begin{itemize} 
        \item $u \in \Vfair$ replaced by existential gadgets, and with $n_u = m/2$ (which corresponds to (branch 2)-1), or
        \item $u \in \Vfair$ replaced by universal gadgets, and with $n_u = m/2-1$ (which corresponds to (branch 2)-2)  \label{itemize:m-priority-nodes}
    \end{itemize}
    Note that for all such nodes $u$, (branch 1) of $u$ is also in $F$. This is because $u \to succ(u)$ is taken infinitely often in $\rho$. For $u \in \Vfair_\exists$, this is due to the construction of $s$, for $u \in \Vfair_\forall$, this is due to $\rho$ being $\forall$-fair; recall that in this case $succ(u) \in E_f(u)$.
    
    \smallskip
	\begin{linenomath}
    Next, we remove from $F$ all gadget nodes with priority $m$ (and every game node that is reachable only from those nodes). That is, we prune out (branch 2) of all the nodes that bring in $m$ priority gadget nodes to $F$. 
    Due to the remaining (branch 1)s, this pruning does not cause dead-ends. We refer to this pruned subgame of $F$ as $H$. Observe once more that all plays in $H$ are $\exists$-winning. However, the maximum priority in $H$ is $m-1$. This is due to the remaining (branch 1)s of the pruned nodes having this priority. 
    This implies that all infinite plays starting in $H$ get trapped in a subgame $H'$ of $H$ that does not have nodes with priority $m-1$.
    Since none of the nodes in $\Vfair \cap H'$ cause a gadget node with priority $m$, none of their branches get pruned. That is, all nodes in $H'$ have the
    same outgoing edges in $H'$ and in $F$. Then every play that starts in $H'$ in $F$ does not leave $H'$, making $H'$ exactly the set of nodes seen infinitely often in $\rho'$, that is, we have $H' = F$. This contradicts our initial assumption 
    that the maximum priority seen infinitely often in $\rho'$ is $m$; thus $\rho$ is $\exists$-winning. 
	\end{linenomath}




\input{paritybotproof.tex}

\begin{remark}\label{rem:opt}
The gadgets in~\cref{fig:existentialgadgetsparitybot} have equivalent \emph{minimal} versions that are obtained by removing redundant branches; in~\cref{fig:optimalexistentialparitybotgadgets} we
depict the existential versions of the minimal gadgets for fair
parity/$\bot$ games. In more detail, consider a node $v\in \Vfair$ 
with priority $\lambda(v)$. Then in~\cref{fig:existentialgadgetsparitybot}, every visit to the gadget 
for $v$ visits priority at least $\lambda(v)$. Thus the gadget branches
that visit priorities less than $\lambda(v)$ are redundant. We show this by case distinction on $\lambda(v)$ for the existential versions of the 
gadgets; the argumentation for the universal versions is analogous. Let
$\mathbf{c}=\lceil \lambda(v)/2\rceil$.
\begin{itemize}
\item If $\lambda(v)=1$, then the minimal version of the gadget is the same as the standard version of the gadget, so there is nothing to show.
\item If $\lambda(v)>1$ and $\lambda(v)$ is odd, then visits
to $v_{\mathbf{c}-j}$ (for $j>0$) have similar outcomes as visits to $v_\mathbf{c}$, but may be disadvantageous to $\exists$-player with
regards to satisfying the parity objective.
Let $\pi$ and $\pi'$ be paths through the existential gadget for $v$ that are identical except for that
$\pi$ visits $v_\mathbf{c}$ and $\pi'$ visits $v_\mathbf{c-j}$.
The maximal priority in $\pi'$ is $\lambda(v)$, a bad event for
$\exists$-player since $\lambda(v)$ is odd. The maximal priority
in $\pi$ on the other hand is either $\lambda(v)$ or $\lambda(v)+1$. So
$\exists$-player can only decrease their chance of winning by moving to 
$v_\mathbf{c-j}$, showing that the branch $v_\mathbf{c-j}$ is redundant.
\item If $\lambda(v)$ is even, then visits
to $v_\mathbf{c}$ and $v_{\mathbf{c}-j}$ (for $j>0$) even have identical outcomes regarding
the parity objective.
Let $\pi$ and $\pi'$ be paths through the existential gadget for $v$ that are identical except for that
$\pi$ visits $v_\mathbf{c}$ and $\pi'$ visits $v_\mathbf{c-j}$.
The maximal priority in both $\pi'$ and $\pi$ is $\lambda(v)$, so
the two paths through the gadget are equivalent.
\end{itemize}
Thus we can safely remove, from gadgets for game nodes $v$, 
all branches with nodes that have priority less than $\lambda(v)$.
While this does not affect the asymptotic size of the reduced games,
it leads to smaller games and is expected to have noticeable impact
in practice.

We further note that the proof of~\cref{thm:fairParityToParity} works without any changes for the minimal gadgets, as we always restrict ourselves to the \emph{limited} game graphs $LG'_{s'}$ and $LG'_{t'}$ in the proof, and the minimal gadgets do not change the limited game graphs. 
\end{remark}


\begin{figure}
	\centering
	\tikzset{every state/.style={minimum size=15pt}, every node/.style={minimum size=15pt}}
	\begin{tikzpicture}[
		auto,
		node distance=0.8cm,
		semithick
		]
		\node[draw,rounded corners, left= 3.2cm of 0, label=left:{\scalebox{0.7}{$\lambda(v)$}}] (vleft) {$v$};
		\node (yo) [below= 4mm of vleft] {};
		\node (yoleft) [left= 3.5mm of yo] {$\dots$};
		\node (yoright) [right= 3.5mm of yo] {$\dots$};
		\node[rectangle,draw] (r1) [left= 18mm of yo] {$v_\mathbf{c}$};
		\node (yo10) [below= 4mm of r1] {};
		\node[draw,rounded corners, colored] (c1) [left=0.2mm of yo10] {$v^\exists_\mathbf{c}$};
		\node at ([yshift=-9mm, xshift=-1mm]c1.north west) {\scalebox{0.7}{$2\mathbf{c}-1$}};
		\node[rectangle,draw] (r2) [right=0.2mm of yo10] {$v^\forall_\mathbf{c}$};
		\node at ([yshift=-9mm]r2.north west) {\scalebox{0.7}{$2\mathbf{c}$}};
		
		\node[rectangle,draw] (rnm) [below= 4mm of vleft] {$v_i$};
		\node (yon0) [below= 4mm of rnm] {};
		\node[draw,rounded corners] (cnm) [left=0.2mm of yon0] {$v^\exists_i$};
		\node at ([yshift=-9mm]cnm.north west) {\scalebox{0.7}{$2i-1\,\,$}};
		\node[rectangle,draw] (rn) [right= 0.2mm of yon0]  {$v^\forall_i$};
		\node at ([yshift=-9mm]rn.north west) {\scalebox{0.7}{$2i$}};
		\node[rectangle,draw] (rmaxevenm) [right= 18mm of yo] {$v_{k+1}$};
		\node[draw,rounded corners] (cmaxevenm) [below= 3mm of rmaxevenm] {$v^\exists_{{k+1}}$};
		\node at ([yshift=-9mm]cmaxevenm.north west) {\scalebox{0.7}{$2k+1$}};
		
		\path[->] (vleft) edge node {} (r1);
		\path[->, colored] (r1) edge node {} (c1);
		\path[->] (r1) edge node {} (r2);
		\path[->] (vleft) edge node {} (rnm);
		\path[->] (rnm) edge node {} (cnm);
		\path[->] (rnm) edge node {} (rn);
		\path[->] (vleft) edge node {} (rmaxevenm);
		\path[->] (rmaxevenm) edge node {} (cmaxevenm);
		
		\node (out1) [below= 5mm of c1, colored] {$E(v)$};
		\node (out2) [below= 5mm of r2] {$E_f(v)$};
		\node (outnm) [below= 5mm of cnm] {$E(v)$};
		\node (outn) [below= 5mm of rn] {$E_f(v)$};
		\node (outmaxevenm) [below= 5mm of cmaxevenm] {$E(v)$};
		\path[->, colored] (c1) edge node {} (out1);
		\path[->] (r2) edge node {} (out2);
		\path[->] (cnm) edge node {} (outnm);
		\path[->] (rn) edge node {} (outn);
		\path[->] (cmaxevenm) edge node {} (outmaxevenm);
		
		\node[draw,rounded corners, right= 3.2cm of 0, label=left:{\scalebox{0.7}{$\lambda(v)$}}] (vright) {$v$};
		\node (yo) [below= 4mm of vright] {};
		\node (yoleft) [left= 3.5mm of yo] {$\dots$};
		\node (yoright) [right= 3.5mm of yo] {$\dots$};
		\node[rectangle,draw] (r1) [left= 18mm of yo] {$v_\mathbf{c}$};
		\node (yo10) [below= 4mm of r1] {};
		\node[draw,rounded corners, colored] (c1) [left=0.2mm of yo10] {$v^\exists_\mathbf{c}$};
		\node at ([yshift=-9mm, xshift=-1mm]c1.north west) {\scalebox{0.7}{$2\mathbf{c}-1$}};
		\node[rectangle,draw] (r2) [right=0.2mm of yo10] {$v^\forall_\mathbf{c}$};
		\node at ([yshift=-9mm]r2.north west) {\scalebox{0.7}{$2\mathbf{c}$}};
		
		\node[rectangle,draw] (rnm) [below= 4mm of vright] {$v_i$};
		\node (yon0) [below= 4mm of rnm] {};
		\node[draw,rounded corners] (cnm) [left=0.2mm of yon0] {$v^\exists_i$};
		\node at ([yshift=-9mm]cnm.north west) {\scalebox{0.7}{$2i-1\,\,$}};
		\node[rectangle,draw] (rn) [right= 0.2mm of yon0]  {$v^\forall_i$};
		\node at ([yshift=-9mm]rn.north west) {\scalebox{0.7}{$2i$}};
		\node[rectangle,draw] (rmaxevenm) [right= 18mm of yo] {$v_k$};
		\node (yomaxevenm0) [below= 4mm of rmaxevenm] {};
		\node[draw,rounded corners] (cmaxevenm) [left=0.2mm of yomaxevenm0] {$v^\exists_k$};
		\node at ([yshift=-9mm]cmaxevenm.north west) {\scalebox{0.7}{$2k-1\,\,$}};
		\node [rectangle,draw]  (rmaxeven) [right=0.2mm of yomaxevenm0] {$v^\forall_k$};
		\node at ([yshift=-9mm]rmaxeven.north west) {\scalebox{0.7}{$2k$}};
		
		\path[->] (vright) edge node {} (r1);
		\path[->, colored] (r1) edge node {} (c1);
		\path[->] (r1) edge node {} (r2);
		\path[->] (vright) edge node {} (rnm);
		\path[->] (rnm) edge node {} (cnm);
		\path[->] (rnm) edge node {} (rn);
		\path[->] (vright) edge node {} (rmaxevenm);
		\path[->] (rmaxevenm) edge node {} (cmaxevenm);
		\path[->] (rmaxevenm) edge node {} (rmaxeven);
		
		\node (out1) [below= 5mm of c1, colored] {$E_f(v)$};
		\node (out2) [below= 5mm of r2] {$E(v)$};
		\node (outnm) [below= 5mm of cnm] {$E_f(v)$};
		\node (outn) [below= 5mm of rn] {$E(v)$};
		\node (outmaxevenm) [below= 5mm of cmaxevenm] {$E_f(v)$};
		\node (outmaxeven) [below= 5mm of rmaxeven] {$E(v)$};
		\path[->, colored] (c1) edge node {} (out1);
		\path[->] (r2) edge node {} (out2);
		\path[->] (cnm) edge node {} (outnm);
		\path[->] (rn) edge node {} (outn);
		\path[->] (cmaxevenm) edge node {} (outmaxevenm);
		\path[->] (rmaxeven) edge node {} (outmaxeven);

	\end{tikzpicture}\caption{Existential minimal gadgets for $v \in \Vfair_\exists$ (left) and $v \in \Vfair_\forall$ (right) following the notation from~\cref{fig:existentialgadgetsparitybot}.
		Here, $\mathbf{c}$ stands for $\lceil \lambda(v) \backslash 2 \rceil$. The full gadgets are used if $\lambda(v)$ is odd and the colored edges are discarded if it is even. Therefore, in a gadget for $v$, all entry and exit nodes have priority at least $\lambda(v)$.}\label{fig:optimalexistentialparitybotgadgets}
	\end{figure} 

   \begin{remark} [Reduction of fair parity/$\top$ games]\label{remark:reducitonofparitybot}
    Consider the gadgets from~\cref{fig:existentialgadgetsparitybot}. In order to play unfairly from a node $v \in \Vfair_\exists$, $\exists$-player has to take its rightmost branch and signal priority $\max_\odd$, whereas to play unfairly from a node $v \in \Vfair_\forall$, $\forall$-player has to take the rightmost branch and signal $\max_\even$. Since $\max_\odd > \max_\even$, this dynamic ensures that mutually unfair plays are $\forall$-winning.

The gadgets for a fair parity/$\top$ game with priority function $\lambda: V \to [2k]$ can be constructed as follows with the addition of priority $2k+2$. Take the gadgets from~\cref{fig:existentialgadgetsparitybot}.
In the existential gadget for $\Vfair_\exists$ 
 add another branch $\to v^\forall_{k+1} \to E_f(v)$ to $v_{k+1}$ and in the universal gadget for $\Vfair_\exists$ add a rightmost branch $\to v_{k+2} \to v_{k+2}^\forall \to E_f(v)$.
 In the existential gadget for $\Vfair_\forall$ add a rightmost branch $\to  v_{k+1} \to v^\exists_{k+1} \to E_f(v)$
 and in the universal gadget for $\Vfair_\forall$ add another branch $\to v^\exists_{k+1} \to E_f(v)$ to $v_{k+1}$.
 
All the newly added gadget nodes have priority $2k+2$ and therefore $\max_\even=2k+2 > \max_\odd=2k+1$, which ensures that mutually unfair plays are $\exists$-winning.
The correctness of this construction follows as a corollary of the reduction of fair parity/parity games given in~\cref{subsection:reductionparityparity} below.
   \end{remark}

%% file: paritybotproof.tex
    \noindent \textbf{($\Leftarrow$)} Next we show that $v \in \Win_\forall(G') \cap V$ implies $v \in \Win_\forall(G)$. To do so, we fix a (positional) winning $\forall$-strategy $t'$ in $G'$ and construct an $\forall$-strategy $t$ in $G$
    such that $t$ is $\forall$-winning in $G$, that is, such that $t$ satisfies~$\neg$\cref{eq:Bone}. 
 To this end,it suffices to show that for all plays $\rho$ in $G$ that start from $v$ and are compliant with $t$,~\cref{eq:B1-t} holds.
      \begin{equation}\fair_\exists(\rho)  \Rightarrow (\fair_\forall(\rho) \wedge \rho \not\in \alpha)
  \tag{$\neg$\ref{eq:Bone}-t}\label{eq:B1-t} \end{equation}
  We show the two parts of the conjunction separately. We
  show that (i) $ \fair_\exists(\rho)
  \Rightarrow \fair_\forall(\rho)$, that is, that every $\exists$-fair play
  that is compliant with $t$ is $\forall$-fair, and that
  (ii) $\fair_\exists(\rho) \Rightarrow
  \rho\not\in\alpha$, that is, that every $\exists$-fair play that is compliant with
  $t$ is $\forall$-winning w.r.t. the parity condition.
  
    \smallskip
    \noindent \textbf{Construction of the $t'$-subgame $G'_{t'}$: } Let $t'$ be a positional $\forall$-strategy winning every play from $v$ in $G'$.
    We denote the subgame of $G'$ where $\forall$ nodes $u$ have only the outgoing edges $u \to t'(u)$ by $G'_{t'}$, and call it \emph{the $t'$-subgame}.

   \smallskip
    \noindent \textbf{Notation of $n_u$ and $succ(u)$:} 
    For the universal gadgets for both $u\in\Vfair_\exists$ and $u\in\Vfair_\forall$, we refer by $n_u$ to the index of the unique successor of $u$ in $G'_{t'}$.  That is, $t'(u) = u_{n_u}$. For the same gadgets, we denote $t'(u_{n_u}^\forall)$ by $succ(u)$. 
    For the existential gadgets for both $u\in\Vfair_\exists$ and $u\in\Vfair_\forall$, \emph{this time we let $n_u-1$} be the index of the rightmost child of $u$ that is sent to its right child by $t'$. That is,
    $n_u-1$ is the largest index $i$ such that $t'(u_i) = u_i^\forall$. For the same gadgets, we denote $t'(u_{n_u-1}^\forall) $ by $succ(u)$. 


    \smallskip
    \noindent \textbf{Construction of $t$: } We define $t: V^*\cdot V_\forall \to V$ as follows.
    For $u \in \Vfair_\forall$, we define $t$ by case distinction.
    \begin{inparaenum} \item If $n_u = k+1$, then we put $t(u)= succ(u)$ \item Otherwise, $t(u)$ is defined to cycle through the set $\{succ(u), E_f(u)\}$, starting at $succ(u)$.
    \end{inparaenum}
    For $u \in V_\forall \setminus \Vfair_\forall$, we put $t(u) = t'(u)$.

     \smallskip
    \noindent \textbf{Constraining $G'_{t'}$ with $n_u$: } We constrain $G'_{t'}$ to its subgame by limiting the choices of $\exists$-player for a node $u$ replaced by an existential gadget.
    For every existential gadget encountered in $G'_{t'}$, we limit the choices of $u \in \Vfair_\exists$ to only $u \to u_{n_u-1}$ and $u \to u_{n_u}$ (if it exists). So, we remove all the other branches of the gadget for $u$ in $G'_{t'}$. We refer to the
    remaining game as $LG'_{t'}$, standing for \emph{limited} $G'_{t'}$. 
    Note that as $LG'_{t'}$ is a subgame of $G'_{t'}$, it is still $\forall$-winning.
   
    \smallskip

 \noindent \textbf{$\forall$-expansion: } Let $\rho= u^1 u^2 \ldots$ be some play in $G$ that is compliant with $t$. We define a play $\rho'$, referred to as the \emph{$\forall$-expansion of $\rho$}, as follows: define $\rho'$ to be the play on $LG'_{t'}$ that follows $\rho$ while `prioritising universal nodes on $succ(u)$'.     
    
    What is meant by this is, for a node $u^i \in \Vfair$, if $u^{i+1} = succ(u^i)$, then $\rho'$ takes the unique branch in $LG'_{t'}$ that leads to $u^{i+1}$ while passing through a universal node node ${(u^i)}_j^\exists$. That is, if $u^i$ is replaced with a universal gadget, then $\rho'$ takes (branch 1) defined as follows.
     \begin{enumerate}
        \item If $u^i \in \Vfair$ is replaced by an existential gadget, then (branch 1)
        is $u^i \to u^i_{n_{u^i}-1} \to {(u^i)}^\forall_{n_{u^i}-1} \to succ(u^i) = u^{i+1}$;
        \item if $u^i \in \Vfair$ is replaced by a universal gadget, then (branch 1)
        is $u^i \to u^i_{n_{u^i}} \to {(u^i)}^\forall_{n_{u^i}} \to succ(u^i) = u^{i+1}$.
    \end{enumerate}
    On the other hand if $u^{i+1} \neq succ(u^i)$, then $\rho'$ takes the only other branch in $LG'_{t'}$, which is the following branch regardless of which gadget $u^i$ is replaced by.
    \begin{equation} u^i \to u^i_{n_{u^i}} \to {(u^i)}^\exists_{n_{u^i}} \to u^{i+1} \tag{branch 2} \label{eq:second-branch1}
    \end{equation}

    For $u^i \not \in \Vfair$, $\rho'$ just takes the edge $u^i \to u^{i+1}$. 
     
    \smallskip
    So $\rho'$ is a well-defined play in $LG'_{t'}$ that
    starts from $v$ and is $\forall$-winning.
    We observe that if we remove the gadget nodes from $\rho'$, we get $\rho$. That is, $\rho' \mid_{V}=\rho$.

    \smallskip
    	Now we are ready to show $\fair_\exists(\rho)
  \Rightarrow \fair_\forall(\rho)$ and
    	$\fair_\exists(\rho) \Rightarrow
    	\rho\not\in\alpha$ hold for any play $\rho$ that starts at $v$ and is compliant with $t$. To this end, assume that $\fair_\exists(\rho)$ is true. We will show (i) $\fair_\forall(\rho)$ and (ii) $\rho\not\in\alpha$.
    	
    \noindent \textbf{(i) $\fair_\forall(\rho)$:}
    We observe that for every $\rho$ in $G$ that is compliant with $t$, by construction of $t$, the only nodes $u \in \Vfair_\forall$ that $\rho$ may not be fair on, are those for which $n_u = k+1$. 

     So it suffices
    to show that such nodes occur finitely often in $\rho$.

	  Let $\rho'$ be the $\forall$-expansion of $\rho$. 
    Since $\rho'|_V = \rho$, it suffices to show that such nodes $u$ do 
    not occur infinitely often in $\rho'$. 
    Denote by $F$ the subgame of $LG'_{t'}$ that consists of only the nodes and edges visited infinitely often by $\rho'$.
    Assume that there is a node $u \in \Vfair_\forall \cap \Inf(\rho)$ with $n_u = k+1$. 
    Then $u$ and the (branch 1) exists in $F$ where it is
    $ u \to u_k \to u_k^\forall $ if $u$ is replaced with  the existential gadget, and $u \to u_{k+1} \to u_{k+1}^\forall$ if $u$ is replaced with the universal gadget. Note that both of these branches signal priority $2k$ (see~\cref{fig:existentialgadgetsparitybot}).
     Since every infinite play in $F$ is $\forall$-winning, there must exist a node in $F$ with priority $\max_\odd = 2k+1$.

    Since only gadget nodes for $w \in \Vfair_\exists$ have this priority, we conclude the existence of a node $w \in \Vfair_\exists$ in $F$ with $n_w = k+1$.

    Due to $\rho$ being $\exists$-fair, both the branches (branch 1) and (branch 2) of every such node $w$ exist in $F$. 
    For all such nodes, remove from $F$ (branch 2), i.e. $u \to u_{k+1} \to u_{k+1}^\exists$, and all nodes and edges reachable solely from this branch. Call the resulting subgame $H$.
    Note that this pruning does not result in any dead-ends since $F$ contains also (branch 1) for each such node. 
    Every infinite path in $H$ is compliant with $t'$ and thus $\forall$-winning.
    However, the maximum priority in $H$ is $\max_\even = 2k$. This
    implies that no infinite play in $H$ visits a node $w \in
    \Vfair_\exists$ with $n_w = k+1$ infinitely often, since this would imply visiting (branch 1) of this node, which would trigger the priority $2k$.
    This implies the existence of a subgame $H'$ of $H$ that has no $w$ with $n_w = k+1$. Therefore, all nodes in $H'$ have the same outgoing edges in $H'$ and in $F$.
    It follows that once $\rho'$ reaches $H'$, it never leaves $H'$, that is, we have $F= H'$.
    This is in contradiction with our initial assumption of a node with priority $n_u = k+1$ existing in $F$.
  
    \smallskip
    \noindent \textbf{(ii) $\rho\not\in\alpha$:}
    We again consider the $\forall$-expansion
    $\rho'$ of $\rho$. Let
    $m+1$ be the largest (odd) priority in $\Inf(\rho')$.
    Due to the equivalence $\rho'\mid_V = \rho$, it suffices to show the existence of a node $u \in \Inf(\rho'\mid_V)$ that has priority $m+1$ to prove $\rho\not\in\alpha$.
 
    Assume to the contrary that there is no such node $u$ and one again, denote by $F$ the subgame consisting of the nodes and edges visited infinitely often by $\rho'$.
    Then there is no $u\in V \cap F$ with $\lambda(u) = m+1$.
    Since the priority $m+1$ is guaranteed to exist in $F$, it must appear only in gadget nodes.

    Then, there is at least one node $w$ in $F$ with $n_w = (m \backslash 2) + 1$ (one can check that this is the only case a gadget node signals priority $m+1$ on the existential or universal gadgets in $LG'_{t'}$). 
    For all $w$ with $n_w = (m \backslash 2) + 1$, remove from $F$ the (branch 2), i.e. $u \to u_{(m \backslash 2)+1} \to u_{(m \backslash 2)+1}^\forall$\textemdash the branches that signal priority $m+1$. Denote the subgame of $F$ reachable from $w$, by $H$. Once again, $H$ has no dead-ends.
    Since all gadget nodes that signal priority $m+1$ are removed in this process, $H$ contains no nodes with priority $m+1$. However, it contains
    gadget nodes with priority $m$ that come from the (branch 1)s of $w$ with $n_w = (m \backslash 2) + 1$\textemdash which is the maximum priority in $H$.
    Since all infinite plays in $H$ are compliant with $t'$ and thus $\forall$-winning, the priority $m$ is visited finitely often in infinite plays of $H$.
    This implies the existence of a subgame $H'$ of $H$ that does not contain nodes with priority $m$. As before, this subgame $H'$ is equal to $F$ since all the nodes in $H'$ have the same outgoing edges in $H'$ and $F$.
    However, this contradicts
    our initial assumption that the maximum priority that is visited infinitely often by $\rho'$ is $m+1$.

    \noindent We therefore conclude that $\rho$ is $\forall$-winning, that is,
    that $\rho\not\in\alpha$, as required. \qed

%% file: reductionofmutuallyfairparityparity.tex
\subsection{Reduction of Fair Parity/Parity Games}\label{subsection:reductionparityparity}

\leavevmode\par\medskip

In this section, we present a quadratic reduction from fair parity/parity games to parity games.
So let
$G=(A, \text{Parity}(\lambda),\text{Parity}(\Gamma))$ where $A = (V_\exists,V_\forall,E,E_f)$ is a fair game arena with $V = V_\exists \dotcup V_\forall$ and priority functions $\lambda:V\to[2k]$, $\Gamma: V\to [d]$.

The reduction is based on ideas from the previous section, in particular
adapting the basic structure of the introduced gadgets. However, in order to
correctly treat mutually unfair plays according to
the additional parity objective
$\Gamma$, we annotate game nodes $v\in V$
with two memory values $p\in[d]$ and $b\in\{\exists,\forall\}$.
The former is used to store the maximal priority according
to $\Gamma$ that the play has \emph{recently} seen; this value
is signalled (and reset after signalling)
from time to time by visiting an according priority in the 
reduced game. The value $b$ is used to decide (at certain nodes)
whether the memory value is signalled, or not.


\begin{figure}
\begin{footnotesize}
\begin{center}
  \tikzset{every state/.style={minimum size=15pt}, every node/.style={minimum size=15pt}}

    \begin{tikzpicture}[
      auto,
      node distance=0.8cm,
      semithick
      ]
       \node[draw,rounded corners, label=left:{\scalebox{0.7}{$\lambda(v)$}}] (0) {$v,p,b$};
       \node (yo) [below= 4mm of 0] {};
       \node (yoleft) [left= 1.3cm of yo] {$\dots$};
       \node (yoright) [right= 1.3cm of yo] {$\dots$};
       \node[rectangle,draw] (r1) [left= 35mm of yo] {$u_1$};
       \node (yo10) [below= 4mm of r1] {};
       \node[draw,rounded corners, label=left:{\scalebox{0.7}{$1$}}] (c1) [left=2mm of yo10] {$u^\exists_1$};
       \node[rectangle,draw, label=left:{\scalebox{0.7}{$2$}}] (r2) [right=2mm of yo10] {$u^\forall_1$};

       \node[rectangle,draw] (rnm) [below= 4mm of 0] {$u_{i}$};
       \node (yon0) [below=4mm of rnm] {};
       \node[draw,rounded corners, label=left:{\scalebox{0.7}{$2i-1$}}] (cnm) [left=2mm of yon0] {$u^\exists_{i}$};
       \node[rectangle,draw, label=left:{\scalebox{0.7}{$2i$}}] (rn) [right=2mm of yon0] {$u^\forall_{i}$};
       \node[rectangle,draw] (rmaxodd) [right= 35mm of yo] {$u_{k+1}$};
       \node[draw,rounded corners,
       label=right:{\color{orange}{\scalebox{0.7}{$\begin{cases}2k+1 & b=\exists\\2k+2+p & b=\forall\end{cases}$}}}] (cmaxodd) [below=4mm of rmaxodd] {$u^\exists_{k+1}$};
       \path[->] (0) edge node {} (r1);
       \path[->] (r1) edge node {} (c1);
       \path[->] (r1) edge node {} (r2);
       \path[->] (0) edge node {} (rnm);
       \path[->] (rnm) edge node {} (cnm);
       \path[->] (rnm) edge node {} (rn);
       \path[->] (0) edge node {} (rmaxodd);
       \path[->] (rmaxodd) edge node {} (cmaxodd);
       \node (out1) [below=5mm of c1] {$E(v,p,b)$};
       \node (out2) [below=5mm of r2] {$E_f(v,p,b)$};
       \node (outnm) [below=5mm of cnm] {$E(v,p,b)$};
       \node (outn) [below=5mm of rn] {$E_f(v,p,b)$};
       \node (outmaxodd) [below=5mm of cmaxodd] {$E(v,\color{orange}{p/1}
       \color{black},\exists)$};
       \path[->] (c1) edge node {} (out1);
       \path[->] (r2) edge node {} (out2);
       \path[->] (cnm) edge node {} (outnm);
       \path[->] (rn) edge node {} (outn);
       \path[->] (cmaxodd) edge node {} (outmaxodd);
    \end{tikzpicture}
\end{center}
\end{footnotesize}
    \caption{Gadget for $v \in \Vfair_\exists$ in fair parity/parity games; $u$ abbreviates $(v,p,b)$.}\label{fig:existentialgadgetparityparity}
\end{figure}


It indicates the player that has last taken the rightmost branch
in the gadget for one of its fair nodes.
If this bit keeps flipping between $\exists$ and $\forall$
forever, then both players intuitively insist on keeping control in one of their
respective fair nodes, enabling a mutually unfair play; in the reduced game, the memory content $p$ is signalled (and then reset to $1$) whenever the value flips from $\forall$ to $\exists$.

\begin{figure}
\begin{footnotesize}
\begin{center}

  \tikzset{every state/.style={minimum size=15pt}, every node/.style={minimum size=15pt}}
      \begin{tikzpicture}[
      auto,
      node distance=0.8cm,
      semithick
      ]
       \node[draw, rectangle,label=left:{\scalebox{0.7}{$\lambda(v)$}}] (0) {{$v,p,b$}};
       \node (yo) [below= 4mm of 0] {\large{$\dots$}};
       \node (yoleft) [left= 2.5cm of yo] {};
       \node[draw,rounded corners] (c1) [left= 0.1cm of yoleft] {$u_1$};
       \node[draw,rounded corners] (cn) [right= 1cm of yo] {{$u_{k+1}$}};
       \node (yo10) [below=4mm of c1] {};
       \node[draw,rounded corners,label=left:{\scalebox{0.7}{$1$}}] (c1-2) [left of=yo10] {$u^\exists_1$};
       \node[draw,rounded corners] (c2) [right= 0.8cm of yoleft] {$u_2$};
       \node (yo20) [below=4mm of c2] {};
       \node[rectangle,draw, label=left:{\scalebox{0.7}{$2$}}] (r2) [left of=yo20] {$u^\forall_2$};
       \node[draw,rounded corners,label=left:{\scalebox{0.7}{$3$}}] (c3) [right of=yo20] {$u^\exists_2$};
       \node (yon0) [below=4mm of cn] {};
       \node[rectangle, draw,label=left:{\scalebox{0.7}{$2k$}}] (rn) [left of=yon0] {$u^\forall_{k+1}$};
       \node[draw,rounded corners,label=right:{\scalebox{0.7}{$2k+1$}}] (cnp) [right of=yon0] {{$u^\exists_{k+1}$}};
       \node (rightdots) [right= 0.6cm of cn] {};
       \node[draw,rounded corners] (cmaxeven) [right= 7.5mm of rightdots] {{$u_{k+2}$}};
       \node (yocmaxeven) [below=4mm of cmaxeven] {};
       \node[draw,rectangle,label=right:{\scalebox{0.7}{$2k+2$}}] (rmaxeven) [below=4mm of cmaxeven] {{$u^\forall_{k+2}$}};
       \path[->] (0) edge node {} (c1);
       \path[->] (0) edge node {} (c2);
       \path[->] (0) edge node {} (cn);
       \path[->] (0) edge node {} (cmaxeven);
       \path[->] (c1) edge node {} (c1-2);
       \path[->] (c2) edge node {} (r2);
       \path[->] (c2) edge node {} (c3);
       \path[->] (cn) edge node {} (cnp);
       \path[->] (cn) edge node {} (rn);
       \path[->] (cmaxeven) edge node {} (rmaxeven);
      \node (out1) [below=5mm of c1-2] {$E_f(v,p,b)$};
      \node (out2) [below=5mm of r2] {$E(v,p,b)$};
      \node (out3) [below=5mm of c3] {$E_f(v,p,b)$};
      \node (outn) [below=5mm of rn] {$E(v,p,b)$};
      \node (outnp) [below=5mm of cnp] {$E_f(v,p,b)$};
      \node (outmaxeven) [below=5mm of rmaxeven] {$E(v,p,\forall)$};
      \path[->] (c1-2) edge node {} (out1);
      \path[->] (r2) edge node {} (out2);
      \path[->] (c3) edge node {} (out3);
      \path[->] (rn) edge node {} (outn);
      \path[->] (cnp) edge node {} (outnp);
      \path[->] (rmaxeven) edge node {} (outmaxeven);
    \end{tikzpicture}
    \end{center}
    \end{footnotesize}

    \caption{Gadget for $v \in \Vfair_\forall$ in fair parity/parity games; $u$ abbreviates $(v,p,b)$.}\label{fig:universalgadgetparityparity}
   \end{figure}

Formally, the reduction is as follows. The game is based
on the set $V\times[d]\times[2]$ of base nodes, where we use $[2]$ to denote
$\{\exists,\forall\}$; intuitively, a
node $(v,p,b)$ from this set corresponds to $v\in V$, annotated with
memory values $p$ and $b$ as described above.
In order to succinctly refer to the combination of taking
a move in $G$ and updating the memory components, we overload notation and put
\begin{align*}
E(v,p,b)&=\{(w,p',b)\in V\times [d]\times [2]\mid w\in E(v)\text{ and } p'=\max(p,\Gamma(v))\}\\
E_f(v,p,b)&=\{(w,p',b)\in V\times [d]\times [2]\mid w\in E_f(v)\text{ and } p'=\max(p,\Gamma(v))\}
\end{align*}
for $(v,p,b)\in V\times[d]\times[2]$.
Thus, a triple $(w,p',b)$ is contained in
$E(v,p,b)$ if there is a move $(v,w)\in E$ and $p'$ is the
maximum of the previous memory value $p$ and the current priority
$\Gamma(v)$ of $v$; in $E_f(v,p,b)$, we require $(v,w)\in E_f$ instead.
In both functions, the argument $b$ is used
to explicitly set this component of the memory to either $\exists$
or $\forall$.
The reduced game consists of subgames
that start at annotated nodes $u=(v,p,b)\in V\times[d]\times [2]$.
In case that $v\in \Vn$, the game just proceeds according to
$E(v,p,b)$, with ownership of $(v,p,b)$
determined by whether $v\in V_\exists$ or $v\in V_\forall$;
this corresponds to taking a move at a normal node in $G$, but updating the memory component $p$, and keeping the component $b$ unchanged.
For fair nodes $v\in \Vfair$,
the subgame
consists of three levels, and after these three steps leads back to a node
from $V\times[d]\times [2]$.
\cref{fig:existentialgadgetparityparity} and~\ref{fig:universalgadgetparityparity} show
the subgames that start at $(v,b,p)\in V\times[d]\times [2]$ such that
$v\in \Vfair_\exists$ and $v\in \Vfair_\forall$,
respectively, adapting the existential gadget for $v \in \Vfair_\exists$ and the universal one for $v\in \Vfair_\forall$.

The rightmost branches in these gadgets overwrite the last component $b$ with $\exists$ and $\forall$, respectively. 
The colored values in the rightmost branch in the gadget in 
\cref{fig:existentialgadgetparityparity} depend on the value of $b$. If $b=\forall$
(corresponding to $\forall$-player being the one that last has taken the rightmost branch), then the priority $2k+2+p$ is signalled and the memory value $p$
is reset to $1$;
if $b=\exists$
(corresponding to $\exists$-player having taken the rightmost branch
last), then the priority $2k+1$ is signalled and the memory value $p$ is not updated further.

Before we state the theorem, we demonstrate the reduction on a simple game, adapted from~\cref{fig:division} by coloring the nodes twice, with $\alpha$ and $\beta$ priorities. 

 \begin{figure}
  \centering
  \tikzset{every state/.style={minimum size=50pt}}
    \begin{tikzpicture}[
      auto,
      node distance=0.8cm,
      semithick
      ]
       \node[draw, rounded corners] (0) {$1/$\color{orange}{$1$}};
       \node (yo) [right of=0] {};
       \node[rectangle,draw] (1) [right of=yo] {$1/$\color{orange}{$1$}};
       \node (yo1) [left of=0] {};
       \node[draw, rounded corners] (2) [left of=yo1] {$1/$\color{orange}{$1$}};
       \node (yo2) [right of=1] {};
       \node[draw, rounded corners] (3) [right of=yo2] {$2/$\color{orange}{$1$}};
      \node[below=2pt of 2] {$v_1$};
      \node[below=2pt of 0] {$v_2$};
      \node[below=2pt of 1] {$v_3$};
      \node[below=2pt of 3] {$v_4$};

       \path[->] (2) edge [loop left] node [below] {} (2);
       \path[->] (3) edge [loop right] node [below] {} (3);
       \path[->] (0) edge [bend right=30] node [pos=0.3,below] {} (1);
       \path[->,blue, very thick] (1) edge [bend right=30] node [pos=0.3,below] {} (0);
       \path[->,dashed] (0) edge node [pos=0.3,left] {} (2);
       \path[->,dashed] (1) edge node [pos=0.3,left] {} (3);
    
    \end{tikzpicture}
    \caption{A parity/parity game arena adapted from~\cref{fig:division}. Each node is labeled twice, with $\alpha$- and $\beta$-priorities, where the $\beta$ priorities are shown in orange. The $\alpha$ priorities coincide with those in~\cref{fig:division}. Priorities are written inside the nodes, while the node names are placed below them. Universal player wins this game from all the nodes except for~$v_4$, following the winning strategy indicated by the thicker blue edge; whereas existential player wins from $v_4$ via the trivial strategy.}
    \label{fig:parityparityexample}
    
  \end{figure}
\begin{example}
  We demonstrate the reduction on the game depicted in~\cref{fig:parityparityexample}. 
Each game node is assigned two colors corresponding to the $\alpha$- and $\beta$-priorities. The $\beta$-priorities, chosen to remain simple and avoid a further blow-up in the reduced game, are uniformly set to~$1$. This setup ensures that $\forall$-player wins all mutually unfair plays.
In this game, $\forall$-player wins from the set $\{v_1, v_2, v_3\}$ by taking the (thicker blue) edge from~$v_3$ to~$v_2$.
If $\exists$-player acts fairly and takes her fair edge from~$v_2$ to~$v_1$,
then the play eventually remains in~$v_1$, only visiting priority~$1$,
thus yielding a win for~$\forall$-player.
If, on the other hand, $\exists$-player chooses to act unfairly by always taking the edge from~$v_2$ back to~$v_3$,
then the resulting play is mutually unfair, and the winner is determined by the maximal $\beta$-priority encountered,
i.e.~$1$—which gives the victory to~$\forall$-player. 

The reduced parity game is shown in~\cref{fig:parityparityexample2}, where the nodes~$v_2$ and~$v_3$ are first duplicated for the two possible values of $b\in\{\exists,\forall\}$ and then replaced by the corresponding gadgets from~\cref{fig:existentialgadgetparityparity} for $p=1$ (as $1$ is the only $\beta$-priority). If the maximal $\beta$-priority were larger, a (linear) blow-up in the state space of the reduced game occurs, as a separate copy of each fair node and each value of~$p$ is then required.
We further depict nodes inherited from the original game with a gray background, and the gadget nodes with a white background, omitting the notation for the gadget nodes.
For nodes coming from the parity/parity game, the node name is displayed inside the node and its priority is shown next to it; for gadget nodes, priorities are written inside the node.

Our reduction guarantees that whenever a node~$v$ is winning for player~$j \in \{\exists, \forall\}$ in the original parity/parity game, the node~$(v, 1, \exists)$ is likewise winning for player $j$ in the reduced parity game.
Accordingly, we observe that $(v_i, 1, \exists)$ is winning for~$\forall$-player for $1 \leq i \leq 3$, and winning for~$\exists$-player for~$i = 4$.
The winning strategy for~$\forall$-player is indicated by the thicker blue edges.
With this strategy, $\forall$-player wins all nodes in the reduced game except for~$(v_4, 1, \exists)$ and~$(v_4, 1, \forall)$, which are won by~$\exists$-player via the trivial strategy.
\begin{figure} 
		\centering
		\tikzset{every state/.style={minimum size=10pt}}
		\begin{tikzpicture}[
			auto,
			node distance=0.8cm,
			semithick,
			extends/.style={->},
      ]
			\node[draw, rounded corners, fill=gray!30] (q) {$v_2, 1, \exists$};
      \node (label-q) [right of=q, xshift=1.5mm] {$\mathbf{1}$};
      \node (q-mid) [below= 1cm of q] {};
			\node[draw, rectangle] (qa) [left= 3mm of q-mid] {\phantom{a}}
      coordinate[below=1cm of qa] (qa-mid);
      \node (qa-mid-node) at (qa-mid) {};
			\node[draw, rectangle] (qb) [right= 3mm of q-mid] {\phantom{a}}
      coordinate[below=1cm of qb] (qb-mid);
      \node[draw, rounded corners] (qa1) [left= 1mm of qa-mid-node] {$1$};
			\node[draw, rectangle] (qa2) [right= 1mm of qa-mid-node] {$2$};
			\node[draw, rounded corners] (qb-mid-node) at (qb-mid) {$3$};
      \node[draw, rounded corners, fill=gray!30] (q-trap) [below= 1.2cm of qa-mid-node] {$v_1, 1, \exists$};
      \node (label-q-trap) [left of= q-trap, xshift=-1.5mm] {$\mathbf{1}$};

      \node[draw, rectangle, fill=gray!30] (p) [below= 0.3mm of qb-mid-node,xshift=3cm] {$v_3, 1, \exists$};
      \node (label-p) [left of=p, yshift=9pt, xshift=-1.5mm] {$\mathbf{1}$};
      \node[draw, rounded corners] (pb) [below= 1cm of p] {\phantom{a}}
      coordinate[below = 1cm of pb] (pb-mid);
      \node[draw, rounded corners] (pa) [left= 1cm of pb] {\phantom{a}}
      coordinate[below = 1cm of pa] (pa-mid);
      \node[draw, rounded corners] (pc) [right= 1cm of pb] {\phantom{a}}
      coordinate[below = 1cm of pc] (pc-mid);
      \node[draw, rounded corners] (pa-mid-node) at (pa-mid) {$1$};
      \node (pb-mid-node) at (pb-mid) {};
      \node[draw, rectangle] (pb1) [left= 1mm of pb-mid-node] {$2$};
      \node[draw, rounded corners] (pb2) [right= 1mm of pb-mid-node] {$3$};
      \node[draw, rectangle] (pc-mid-node) at (pc-mid) {$4$};
      \node[draw, rounded corners, fill=gray!30] (p-trap1) [below= 1.2cm of pa-mid-node] {$v_4, 1, \exists$};
      \node (label-p-trap1) [left of= p-trap1, xshift=-1mm] {$\mathbf{2}$};
      \node[draw, rounded corners, fill=gray!30] (p-trap2) [below= 1.2cm of pb2] {$v_4, 1, \forall$};
      \node (label-p-trap2) [left of= p-trap2, xshift=-1mm] {$\mathbf{2}$};

      \node[draw, rounded corners, fill=gray!30] (r) [right= 1cm of p-trap2] {$v_2, 1, \forall$};
      \node (label-r) [left of= r, xshift=-1mm] {$\mathbf{1}$};
      \node (r-mid) [below= 1cm of r] {};
			\node[draw, rectangle] (ra) [left= 3mm of r-mid] {\phantom{a}}
      coordinate[below=1cm of ra] (ra-mid);
      \node (ra-mid-node) at (ra-mid) {};
			\node[draw, rectangle] (rb) [right= 3mm of r-mid] {\phantom{a}}
      coordinate[below=1cm of rb] (rb-mid);
      \node[draw, rounded corners] (ra1) [left= 1mm of ra-mid-node] {$1$};
			\node[draw, rectangle] (ra2) [right= 1mm of ra-mid-node] {$2$};
			\node[draw, rounded corners] (rb-mid-node) at (rb-mid) {{\color{orange}$5$}};
      \node[draw, rounded corners, fill=gray!30] (r-trap) [below= 1.2cm of rb-mid-node] {$v_1, 1, \forall$};
      \node (label-r-trap) [left of=r-trap, xshift=-1mm] {$\mathbf{1}$};

      \node[draw, rectangle, fill=gray!30] (t) [left= 1cm of r-trap] {$v_3, 1, \forall$};
      \node (label-t) [left of=t, xshift=-1mm] {$\mathbf{1}$};
      \node[draw, rounded corners] (tb) [below= 1cm of t] {\phantom{a}}
      coordinate[below = 1cm of tb] (tb-mid);
      \node[draw, rounded corners] (ta) [left= 1cm of tb] {\phantom{a}}
      coordinate[below = 1cm of ta] (ta-mid);
      \node[draw, rounded corners] (tc) [right= 1cm of tb] {\phantom{a}}
      coordinate[below = 1cm of tc] (tc-mid);
      \node[draw, rounded corners] (ta-mid-node) at (ta-mid) {$1$};
      \node (tb-mid-node) at (tb-mid) {};
      \node[draw, rectangle] (tb1) [left= 1mm of tb-mid-node] {$2$};
      \node[draw, rounded corners] (tb2) [right= 1mm of tb-mid-node] {$3$};
      \node[draw, rectangle] (tc-mid-node) at (tc-mid) {$4$};

			\draw[->] (q) -- (qa);
      \draw[->] (q) -- (qb);
      \draw[->, blue, very thick] (qa) -- (qa1);
      \draw[->] (qa) -- (qa2);
      \draw[->, blue, very thick] (qb) -- (qb-mid-node);

      \draw[->] (qa1) -- (q-trap);
      \draw[->, blue, very thick] (qa2) -- (q-trap);
      \draw[->] (qb-mid-node) -- (q-trap);
      \draw[->] (qa1) |- (p);
      \draw[->] (qb-mid-node) |- (p);

			\draw[->] (p) -- (pa);
      \draw[->] (p) -- (pb);
      \draw[->, blue, very thick] (p) -- (pc);
      \draw[->] (pa) -- (pa-mid-node);
      \draw[->] (pb) -- (pb1);      
      \draw[->] (pb) -- (pb2);
      \draw[->] (pc) -- (pc-mid-node);

      \draw[->] (pa-mid-node) -- (p-trap1);
      \draw[->] (pb1) -- (p-trap1);
      \draw[->] (pb2) -- (p-trap1);
      \draw[->] (pc-mid-node) -- (p-trap2);
      \draw[->, blue, very thick] (pc-mid-node) -- (r);
      \coordinate[below= 2 cm of q-trap, xshift=-1.3cm] (mid1);
      \draw[->, blue, very thick] (pb1) |- (mid1) |- (q); 

			\draw[->] (r) -- (ra);
      \draw[->] (r) -- (rb);
      \draw[->] (ra) -- (ra1);
      \draw[->, blue, very thick] (ra) -- (ra2);
      \draw[->, blue, very thick] (rb) -- (rb-mid-node);

      \draw[->] (ra1) -- (r-trap);
      \draw[->] (ra1) -- (t);
      \draw[->, blue, very thick] (ra2) -- (r-trap);
      \coordinate[right= 3.35cm of p] (mid6);
      \coordinate[above= 0.6cm of mid6] (mid7);
      \coordinate[right= 1.5cm of q-trap] (mid8);
      \draw[->] (rb-mid-node) -| (mid6) -- (p);
      \draw[->] (rb-mid-node) -| (mid7) -| (mid8) -- (q-trap);

			\draw[->] (t) -- (ta);
      \draw[->] (t) -- (tb);
      \draw[->, blue, very thick] (t) -- (tc);
      \draw[->] (ta) -- (ta-mid-node);
      \draw[->] (tb) -- (tb1);      
      \draw[->] (tb) -- (tb2);
      \draw[->] (tc) -- (tc-mid-node);

      \coordinate[below= 7.4cm of p-trap1] (mid2);
      \coordinate[below= 2cm of p-trap1] (mid3);
      \draw[->] (ta-mid-node) |- (mid2) -- (mid3) -| (p-trap2);
      \draw[->] ([xshift=-2pt]tb1.south) |- (mid2) -- (mid3) -| (p-trap2);
      \draw[->] (tb2) |- (mid2) -- (mid3) -| (p-trap2);
      \draw[->] (tc-mid-node) |- (mid2) -- (mid3) -| (p-trap2);
      \coordinate[below = 5mm of tc-mid-node, xshift=2.2cm] (mid4);
      \coordinate[right=2.2cm of tc-mid-node] (mid5);
      \draw[->, blue, very thick] ([xshift=+2pt]tb1.south) |- (mid4) |- (r);
      \draw[blue, very thick] (tc-mid-node) -- ([xshift=-6pt] mid5);

      \path[->] (p-trap2) edge [loop right, looseness=4] node {} (p-trap2);
      \draw[<-] (q-trap) to[out=245, in=290, looseness=7] (q-trap);
      \draw[<-] (p-trap1) to[out=245, in=290, looseness=7] (p-trap1);
      \draw[<-] (r-trap) to[out=245, in=290, looseness=7] (r-trap);

		\end{tikzpicture} 
		\caption{\small 
		Reduced version of the parity/parity game from~\cref{fig:parityparityexample}.
}
		\label{fig:parityparityexample2}
	\end{figure} 


A closer inspection of the reduced game provides the following intuition.
Player~$\forall$ wins because every play starting from his winning region either reaches one of the nodes~$(v_1, 1, \exists)$, $(v_1, 1, \forall)$\textemdash which form $\forall$-winning traps signaling priority~$1$\textemdash or visits priority~$5$ infinitely often, which is the maximal priority in the game.
The reason why~$\forall$-player has a winning strategy while~$\exists$-player does not\textemdash despite the symmetry of the parity/parity game graph\textemdash is that the maximal priority~$5$ ($2k + 2 + p = 4 + 1$) reflects the maximal $\beta$-priority~$1$ of the mutually unfair play~$(v_2 v_3)^\omega$. 
If the maximal $\beta$-priority of this play were even, the priority signaled in the rightmost branch of~$(v_2, 1, \forall)$ would instead be an even number greater than all other priorities in the game.
In that case, using a strategy analogous to the blue $\forall$-winning strategy, $\exists$-player could ensure that every play starting from any node other than~$(v_1, 1, \exists), $ $(v_1, 1, \forall)$ either reaches an $\exists$-winning trap or yields a play the maximal priority of which has the even value~$2k + 2 + p$. However, the reduced game would be (linearly) larger in this case, as discussed above.
\end{example}

  Next, we formally state and prove the reduction. 
  \begin{restatable}{theorem}{restatablefairParityPlusParityToParity}\label{thm:fairParityPlusParityToParity} Let
$G=(A, \text{Parity}(\lambda),\text{Parity}(\Gamma))$ where $A = (V_\exists,V_\forall,E,E_f)$ is a fair game arena, $V = V_\exists \dotcup V_\forall$ and $\lambda:V\to[2k]$ and $\Gamma: V\to [d]$ are priority functions.
Then there exists a parity game $G'$ with $6 n d(k+2)$ nodes and
$2k+2+d$ priorities with set $V\times[d]\times[2]$ of base nodes
such that for all $v\in V$, $\exists$-player wins $v$ in $G$ if
and only if $\exists$-player wins $(v,1,\exists)$ in $G'$.
\end{restatable}

We construct the parity game $G'$ following the above description, using
the gadgets from~\cref{fig:existentialgadgetparityparity}
and~\ref{fig:universalgadgetparityparity} to treat fair nodes. 

\subsubsection{An overview of the proof of~\cref{thm:fairParityPlusParityToParity}}
 We show that for all $v\in V$, $v \in \Win_\exists(G)$ if and only if $(v, 1, \exists) \in \Win_\exists(G')$. We do so by assuming one player wins the node $(v, 1, \exists)$ in $G'$ and show that the same player wins $v$ in $G$, similar to the proof of~\cref{thm:fairParityToParity}. 
Throughout the proof we will assume $\exists$-player to be the winner of $(v, 1, \exists)$ in $G'$ and show that she wins $v$ in $G$. The proof of the dual case, where $\forall$-player is the winner of $(v, 1, \exists)$, will require identical arguments, and therefore we will leave it to the reader to confirm the second side. 

Note that even though the proof of~\cref{thm:fairParityPlusParityToParity} is substantially harder than that of~\cref{thm:fairParityToParity}, the fact that we have symmetric winning conditions for the players (parity/parity in contrast to parity/$\bot$) implies that the players have dual winning conditions and allows us to ommit one side of the proof, a luxury we could not afford in the case of~\cref{thm:fairParityToParity}.

As in the proof of~\cref{thm:fairParityToParity}, we will proceed by extracting an $\exists$-strategy $s$ in $G$, based on a positional winning $\exists$-strategy $s'$ in $G'$. Then, we show that for every $\forall$-strategy $t$ in $G$ there exists a counter $\exists$-strategy to $t$ s.t. every play compliant with both strategies and start at $v$ is won by the existential player. This will imply $v \in \Win_\exists(G)$ by determinacy, and conclude the proof. In some cases, $s$ will suffice as the counter strategy to $t$, and in others we will need to iterate $s$ to get a counter strategy $s^\circ$.



We will give the proof of~\cref{thm:fairParityPlusParityToParity} in 4 subsections. 
In~\cref{subsec:construction-of-gadget-game} we will formally construct the parity game $G'$ (a.k.a., the gadget game), and prove some of its properties. In~\cref{subsec:restriction-of-gadget-game} we will restrict the choices of $\exists$-player in $G'$ to a positional winning $\exists$-strategy $s'$, obtaining the $s'$-subgame $G'_{s'}$ and prove some properties of $G'_{s'}$. In~\cref{subsec:construction-of-s} we will construct the $\exists$-strategy $s$. Finally in~\cref{subsec:the-parity-plus-parity-to-parity-proof} we will bring all the arguments together and give the main proof.

In the last subsection~\cref{subsec:parity-parity-strategy-size} we will show that fair parity/parity games require exponential strategies. 
\subsubsection{Construction of $G'$ and its properties}\label{subsec:construction-of-gadget-game}
\input{construction-of-gadget-game.tex}

\subsubsection{The $s'$-subgame $G'_{s'}$ and its properties}\label{subsec:restriction-of-gadget-game}
\input{restriction-of-gadget-game.tex}
\subsubsection{Construction of the $\exists$-strategy $s$ in $G$}\label{subsec:construction-of-s}

\input{construction-of-s.tex}

\subsubsection{The main proof of~\cref{thm:fairParityPlusParityToParity}}\label{subsec:the-parity-plus-parity-to-parity-proof}

\input{parity-parity-to-parity-main-proof}
\subsubsection{Strategy sizes of fair parity/parity games}\label{subsec:parity-parity-strategy-size}

We obtain the following bound on strategy sizes for fair parity/parity games.

\begin{lemma}\label{lem:strategysizes}
Let $G$ be a fair parity/parity game on $n$ nodes.
Then for both players the memory requirement of
winning strategies in $G$ is at most $n^2\cdot n^n$.
Furtermore, for each player a family of fair parity/$\bot$
games $(G_n)_{n\in\mathbb{N}}$ exists such that for all $n$, every winning
strategy for the respective player requires memory at least $2^n$.
\end{lemma}
\begin{proof}[Proof of~\cref{lem:strategysizes}]
For the upper bound, we note that in a winning $i$-strategy for a fair parity/parity game, as constructed in the proof
of~\cref{thm:fairParityPlusParityToParity}, the nodes in $V_i \setminus \Vfair$ have strategies with
quadratic memory, but the nodes in $\Vfair_i$ may have to traverse all their fair successors, and possibly one more successor.
In the worst case, this requires an additional local memory of $|E_f(v)| + 1\leq n$ for each $v \in \Vfair_i$, and causes an exponential blowup in the overall memory required.

\begin{linenomath}
For the lower bound, we consider the case for $\exists$-player; the result for $\forall$-player is obtained by switching the player's roles. Define the family $(G_n)_{n\in\mathbb{N}}$ of games by letting $G_n$ (for $n\in\mathbb{N}$) have exactly $n+1$ nodes, one node $x$ owned by $\forall$-player and $n$ nodes $y_i$ owned by $\exists$-player;
let there be an edge from $x$ to every node $y_i$ and two \emph{fair} edges from every node $y_i$ back to $x$. Let all nodes have priority $0$. Then every winning $\exists$-strategy in $G_n$ necessarily is $\exists$-fair. There is a fair $\exists$-strategy $s$ that
uses one bit as local memory for each node $y_i\in \Vfair_\exists$, and therefore uses memory of overall size $2^n$.
The claim follows since there is no $\exists$-fair strategy that
uses less memory than $s$, which is shown by induction on $n$.
\end{linenomath}
\end{proof}

%% file: construction-of-gadget-game.tex

We define the (fairness-free) parity game $G'$, where the gadgets for fair parity/parity games replace fair nodes as follows.
Define \begin{align*}
G'=(V'_\exists,V'_\forall,E',\Omega:V'\to[2k+2+d])
\end{align*}
to be the parity game (with $V'=V'_\exists\dotcup V'_\forall$)
defined by
\begin{align*}
V'_\exists&=V_\exists\times[d]\times[2]\,\cup\\
&\quad\,\{\,(v,p,b)_i\,\mid \,v\in V_\forall,\,1\leq i\leq k+2,\, 1\leq p\leq d,\, b\in[2]\,\}\,\,\cup\\
&\quad\,\{\,(v,p,b)_i^\exists\,\mid\, v\in V_\exists \cup V_\forall ,\, 1\leq i\leq k+1,\, 1\leq p\leq d,\, b\in[2]\,\}\\
V'_\forall&=V_\forall\times[d]\times[2]\,\,\cup\\
&\quad\,\{\,(v,p,b)_i\,\mid\, v\in V_\exists,\, 1\leq i\leq k+1,\, 1\leq p\leq d,\, b\in[2]\,\}\,\,\cup\\
&\quad\,\{\,(v,p,b)_i^\forall\,\mid\, v\in V_\exists,\, 1\leq i\leq k,\, 1\leq p\leq d,\, b\in[2]\,\}\,\,\cup\\
&\quad\,\{\,(v,p,b)_i^\forall\,\mid\, v\in V_\forall,\, 2\leq i\leq k+2, \, 1\leq p\leq d, \,b\in[2]\,\}\end{align*}
and
\begin{align*}
E'((v,p,b))=&\begin{cases}
E((v,p,b)) & v\in \Vn \\
\{\,(v_i,p,b)\,\mid\, 1\leq i\leq k+1\,\}& v\in \Vfair_\exists\\
\{\,(v_i,p,b)\,\mid\, 1\leq i\leq k+2\,\}& v\in \Vfair_\forall
\end{cases} \\
E'((v,p,b)_i)=&
\begin{cases}
 \{\,(v,p,b)_i^\exists,\,(v,p,b)_i^\forall\,\} & i\leq k, \,v\in V_\exists\\
 \{\,(v,p,b)_{k+1}^\exists\,\} & i=k+1, \,v\in V_\exists\\
 \{\,(v,p,b)_i^\forall,\,(v,p,b)_i^\exists\,\} & 1<i<k+2,\, v\in V_\forall\\
 \{\,(v,p,b)_1^\exists\,\} & i=1, \,v\in V_\forall\\
 \{\,(v,p,b)_{k+2}^\forall\,\} & i=k+2,\, v\in V_\forall
\end{cases} \\
E'((v,p,b)_i^\exists)=& \begin{cases}
E((v,p,b))   & v\in V_\exists,\, i\leq k\\
E((v,p,\exists))   & v\in V_\exists, \,i=k+1,\, b=\exists\\
E((v,1,\exists))   & v\in V_\exists,\, i=k+1,\, b=\forall\\
E_f((v,p,b))   & v\in V_\forall\\
\end{cases}\\
E'((v,p,b)_i^\forall)=& \begin{cases}
E((v,p,b))         & v\in V_\forall,\, i< k+2\\
E((v,p,\forall))   & v\in V_\forall,\, i=k+2\\
E_f((v,p,b))       & v\in V_\exists\\
\end{cases}
\end{align*}
for $v\in V$, $p\in[d]$, $b\in[2]$. 
Finally, we put
\begin{align*}
\Omega((v,p,b))&=\lambda(v) & 
\Omega((v,p,b)_i^\exists)&=\begin{cases}
2i-1 & v \in V_\forall \text{ or } i\neq k+1 \text{ or }  b=\exists\\
2k+2+p & v \in V_\exists \text{ and } i=k+1 \text{ and } b=\forall
\end{cases}\\
\Omega((v,p,b)_i)&=1 &
\Omega((v,p,b)_i^\forall)&=\begin{cases}
2i &  v\in V_\exists\\
2i-2 & v\in V_\forall
\end{cases}
\end{align*}
for $v\in V$, $p\in[d]$ and $b\in[2]$.
The claimed bound on the number of priorities
in $G'$ is immediate from the definition of $\Omega$.
To derive the bound on the number of nodes, we note
that $|V\times[d]\times[2]|\leq 2nd$
and that a gadget includes at most $3k+5$ many nodes (for the universal gadget);
hence $|V'|\leq 2nd(3k+5) < 6nd(k+2)$
as claimed.


We set now some notation and prove some properties of the gadget game $G'$. 

\begin{notation}[Base node]\label{not:base-node-base-sequence} We refer to the first component of a node $u = (v, p, b)$ as the \emph{base node} of $u$ and we denote it with $u{\mid}_V =v$. We extend this notation to plays in $G'$ while ommitting repetitions in gadget nodes, and define the \emph{base sequence} of a play $\bm{\rho}$ as $w^1|_V \,w^2|_V\, \ldots$ where $w^1 w^2 \ldots$ is the subsequence of $\bm{\rho}$ that skips over the gadget nodes.  We denote the base sequence of $\bm{\rho}$ with $\bm{\rho}|_V$. Further, we call the set $\{\, v \,|\, v \in \bm{\rho}|_V \, \}$ the \emph{base set} of $\bm{\rho}$ and denote it with $\bm{\rho}^{\{\}}|_V$.
\end{notation}

\begin{notation}[Eventually stabilizing and alternating plays] Let \\ $\bm{\rho} = (v_0, p_0, b_0) (v_1, p_1, b_1)\ldots$ be a play in $G'$. We call $\bm{\rho}$ \emph{eventually stabilizing} if the memory value $b$ stabilizes in $\bm{\rho}$, i.e. there exists an $i$ such that for all $j > i$, $b_i = b_j$. Otherwise, we call it \emph{alternating}. \end{notation}

The following two lemmata state that in alternating plays in $G'$, the winner is determined by the $\beta$ priorities of the base nodes (\cref{lemma:alternating-play-priority}), and in 
eventually stabilising plays, the winner is determined by the $\alpha$ priorities of the base nodes and the priorities of gadget nodes below $2k+2$ (\cref{lemma:stabilizing-p-b}).

\begin{lemma} \label{lemma:alternating-play-priority} 
    If a play $\bm{\rho}$ is alternating, then the maximum priority visited infinitely often in $\bm{\rho}$ is $2+k+\bar{p}$ where $\bar{p}$ is the maximum $\beta$-priority of the nodes in $\Inf(\bm{\rho}|_V)$.
\end{lemma}    

\cref{lemma:alternating-play-priority} follows from the way the $p$ values are updated at the edges of $G'$.

\enskip 

The following lemma states that in eventually-stabilizing plays, the winner is determined by the $\alpha$ priorities of the base nodes and gadget nodes. 
This is because $\beta$ priorities are signalled only by the rightmost edges of existential gadgets with $b = \forall$ (in the form of $2k+2+p$ where $p$ is the $\beta$ priority of the base node) and these edges are not visited infinitely often in eventually stabilizing plays. 

\begin{lemma}\label{lemma:stabilizing-p-b}
    If a play $\bm{\rho}$ in $G'$ is eventually stabilizing, then both the $p$ and $b$ values in $\bm{\rho}$ stabilize. Furthermore, priorities over $2k+2$ are not signalled infinitely often in $\bm{\rho}$, allowing the winner to be determined by the $\alpha$ priorities of base nodes, and the priorities of gadget nodes below $2k+2$. 
\end{lemma}

\begin{proof}[Proof of~\cref{lemma:stabilizing-p-b}]
Let $\bm{\rho}$ be an eventually stabilizing play. Then there exists some $i$ s.t. for all $j > i$, $p_i = p_j$ and $b_i = b_j$. This is because the rightmost edge of any existential gadget (see~\cref{fig:existentialgadgetparityparity}) with $b= \forall$ will eventually not be taken, and therefore the $p$ value will eventually not be reset. As the $p$ values on a play are non-decreasing except for the resets, the $p$ value will eventually stabilize.

Furthermore, as the righmost edge of any existential gadget with $b = \forall$ is eventually not taken, the priorities $2k+2+p$ are signaled finitely often in the play. That is, the $\beta$ values of the base nodes eventually do not contribute to determining the winner of the play. Therefore, the winner is determined by the $\alpha$ priorities of the base nodes and the priorities of gadget nodes below $2k+2$. \end{proof}

The following lemma states that the winner of a node $(w, p, b)$ in the gadget game is decided by its base node $w$. 

\begin{lemma}\label{lemma:winningunformity} Let $w \in V$. Then the game nodes $(w, p, b)$ for all $p$ and $b$ are won by the same player in $G'$.
\end{lemma}
\begin{proof}[Proof of~\cref{lemma:winningunformity}]
	We consider the case that existential player wins $(w, p, b)$; the case for the other player is dual.
	Let $s^\circ$ be a positional winning $\exists$-strategy in $G'$ that wins $(w, p, b)$.
	We construct a strategy $s'$ that uses
	memory values from $[d]$ and $[2]$ to mimic the behavior of $s^\circ$ from any starting point. Given a game node $\bm{u}\in V'_\exists$
	and memory values $p',b'$, we define $s'$ depending
	on the shape of $\bm{u}$.
	\begin{itemize}
		\item For game nodes $\bm{u}$ that are of the shape $(z,p'',b'')$ such that $z\in V_\exists\setminus \Vfair$ (that is, $(z,p'',b'')$ is not a gadget node), define $s'$ to simply take the move that $s^\circ$ prescribes for the game node $(z,p',b')$. Formally, we put $s'(z,p'',b'')=(z',p''',b'')$, where $s^\circ(z,p',b')=(z',q,c)$ and $p'''$ is the maximum of $p''$ and the $\beta$-priority of $z$. 

		\item For game nodes $\bm{u}$ that are an exit node of a gadget
		(that is, $\bm{u}=u^\exists_i$ for some $u=(z,p'',b'')$),
		define $s'$ to take the move that $s^\circ$ prescribes for the game node $(u')^\exists_i$ where $u'=(z,p',b')$. That is, $s'$ copies the exit behavior of $s^\circ$ on the gadget for $(z,p',b')$.
		Formally, put $s'(u^\exists_i)=(z',p''',b'')$, where $s^\circ((u')^\exists_i)=(z',q,c)$ and $p'''$ is the maximum of $p''$ and the $\beta$-priority of $z$. 
		
		\item For game nodes $\bm{u}$ that are internal nodes of a universal gadget
		(that is, $\bm{u}=u_i$ for some $u=(z,p'',b'')$ such that $z\in \Vfair_\forall$),
		define $s'$ to take the choice that $s^\circ$ prescribes for the game node $u'_i$ where $u'=(z,p',b')$. That is, $s'$ copies the internal behavior of $s^\circ$ on the gadget for $(z,p',b')$.
		Formally, put $s'(u_i)=u_i^\exists$
		if $s^\circ(u'_i)={u'}_i^\exists$, and $s'(u_i)=u_i^\forall$ otherwise.
		Do not update the memory values.

		\item For existential gadget nodes $\bm{u}$ that are the entry node of a gadget for
		an existentially fair base node (that is, $\bm{u}$ is of the shape $\bm{u}=(z,p'',b'')$ such that $z\in \Vfair_\exists$), define $s'$ to take the branch that $s^\circ$ prescribes for the game node $u'=(z,p',b')$. That is, $s'$ copies the entry behavior of $s^\circ$ on the gadget for $u'$. Formally, put $s'(u)=u_i$ where $i$ is such that $s^\circ(u')=u'_i$. Do not update the memory values.

	\end{itemize}
	The memory update in $s'$ on existential moves has been described above. On universal moves, the strategy $s'$ simply updates the memory value according to the respective move according to $s^\circ$; e.g. for a universal move from
	game node $(z,p'',b'')$ with memory content $p'$ and $b'$ to game node $(z',p''',b''')$, there is a corresponding move from $(z,p',b')$ to
	$(z',q,c)$ for some $q$ and $c$. Use $q$ and $c$ as new memory values.
	
	Then for every play $\pi$ starting at a node $(w,p',b')$ that is compatible with strategy $s'$ using starting memory values $p$ and $b$, there is 
	an induced play $\pi_s$ starting at $(w,p,b)$ such that the base sequences of $\pi$ and $\pi_s$ are exactly the same, and both
	$\pi$ and $\pi_s$ take exactly the same branches through gadgets.
	In particular, $\pi$ is eventually stabilizing if and only if
	$\pi_s$ is eventually stabilizing.
	The sequence of the memory values that $s'$ uses to construct $\pi$ is just sequence of the memory components in the nodes of $\pi_s$.
	
	It follows that $\pi$ and $\pi_s$ are won by the same player: If
	both $\pi$ and $\pi_s$ are eventually stabilizing, then the winner of
	both plays is determined by their sequences of $\alpha$-priorities and the priorities of at most $2k+2$ coming from the gadget nodes (\cref{lemma:stabilizing-p-b}), which
	are identical as both plays agree on base nodes and gadget branches. If both
	$\pi$ and $\pi_s$ are alternating, then the winner of
	both plays is determined by their sequences of $\beta$-priorities (\cref{lemma:alternating-play-priority}), which
	again are identical as both plays agree on base nodes.
\end{proof}


%% file: restriction-of-gadget-game.tex
\smallskip
Let $s'$ be a positional $\exists$-strategy in $G'$.
    We will denote the subgame of $G'$ where $\exists$ nodes have only the outgoing edges $u \to s'(u)$ by $G'_{s'}$, and call it \emph{the $s'$-subgame}.
    Recall that all plays that start from $v$ in $G'_{s'}$ are $\exists$-winning.
  
 Next, we will set up some notation and extract observations about $G'_{s'}$.

\enskip 
\begin{notation}[$n_u$]\label{not:n} For an existential gadget node $u = (v, p, b)$ where $v \in \Vfair_\exists$ we will denote the index $i$ of its unique successor $u_i$ in $G'_{s'}$, by $n_u$. For a universal gadget node $u = (v, p, b)$ where $v \in \Vfair_\forall$ we will denote the largest index $i$ s.t. $s'(u_i) = u^\exists_i$ by $n_u$.\footnote{i.e., on every branch to the right of $u \to u_{n_u} \to u^\exists_{n_u}$ of the universal gadget, $s'$ chooses the universal successor, that is, for every $i > n_u$, $s'(u_i) = u^\forall_i$.} \end{notation}

\enskip

Next, in~\cref{lemma:nu-regularity-for-exists-player},\cref{cor:lemma-regularity} and~\cref{lemma:nu-regularity-for-forall-player}, we show that we can restrict attention to winning strategies $s'$ that have some restrictions on their structure in term of choices on the branches of universal gadgets. 


\begin{lemma}\label{lemma:nu-regularity-for-exists-player} There exists an existential winning strategy $s'$ in $G'$ such that for every universal gadget node $u = (v, p, b)$ won by $\exists$-player, $s'(u_i) = {u_i^\exists}$ for all $i < n_u$. 
\end{lemma}

\begin{proof}[Proof of~\cref{lemma:nu-regularity-for-exists-player}]
Let $u = (v, p, b)$ be as given. We claim that all $u^\exists_i$ with $i < n_u$ are also winning for her. The existence of the desired $s'$ follows from the claim. 

Since $u$ is winning for $\exists$, we know all $u_i$ are as well. To get the winning strategy from $u^\exists_i$ for some $i < n_u$, it is enough to set the strategy to take the edge $u^\exists_i \to w$ where $w = s'(u^\exists_{n_u})$. 
The equality follows since $u^\exists_{i}$ and $u^\exists_{n_u}$ have the same set of successors and the former has a smaller odd priority w.r.t. the later. 
\end{proof}

\enskip

\begin{corollary}[\cref{lemma:nu-regularity-for-exists-player}]\label{cor:lemma-regularity} It follows from \cref{lemma:nu-regularity-for-exists-player} that for every universal gadget node $u = (v, p, b)$ we can WLOG assume that in $G'_{s'}$, 
all $u_i$ with $i \leq n_u$ have only their right (existential) successors and all $u_i$ with $i > n_u$ have only left (universal) successors. It follows from the proof of the lemma that we can further assume that for all $i, j \leq n_u$, $s'(u^\exists_i) = s'(u^\exists_j)$. 
\end{corollary}

\enskip

The next lemma, \cref{lemma:nu-regularity-for-forall-player}, states that there is a strategy $s'$ with restricted choices in universal gadgets with memory value $b = \forall$. It states that a winning $\exists$-strategy can be made to \textit{not} take the existential branch on the ${k+1}^\text{th}$ branch of such nodes, resulting in $n_u \neq k+1$ (i.e. $n_u \leq k$) for every universal gadget node $u = (v, p, \forall)$ in $\Win_\exists(G')$. This guarantees the existence of an existential branch for each such $u$ in $G'_{s'}$, a fact we will exploit in the proof.

\begin{lemma}\label{lemma:nu-regularity-for-forall-player} There exists a winning $\exists$-strategy in $G'$ which satisfies $n_u \leq k$ for every universal gadget node $u = (v, p, \forall)$ in $\Win_\exists(G')$.
\end{lemma}

\begin{proof}[Proof of~\cref{lemma:nu-regularity-for-forall-player}] Let $s'$ be a positional winning $\exists$-strategy and $u = (v, p, \forall) \in \Win_\exists(G')$ be some universal gadget node in $G'$. 
Assume that $n_u=k+1$, that is, $s'(u_{k+1})=u^\exists_{k+1}$.
We will show that in this case, the existential player wins from $u^\forall_{k+1}$ as well. This implies that $s'$ can be made to take the edge to $u^\forall_{k+1}$ instead of $u^\exists_{k+1}$ from $u_{k+1}$, giving an equivalent winning strategy with $n_{u} \leq k$. We will use proof by contradiction.

Observe that the existential player wins the game from $u_{k+2}$, and therefore from any successor of $u^\forall_{k+2}$ upon seeing the even priority $2k+2$, i.e. the priority of $u^\forall_{k+2}$. 
Furthermore, the successors of $u^\forall_{k+1}$ and $u^\forall_{k+2}$ are identical, due to $u$ having $b= \forall$. The only difference between these two branches are the 
even priorities, $2k$ and $2k+2$, seen upon visiting the gadget nodes $u^\forall_{k+1}$ and $u^\forall_{k+2}$, respectively.
Assume a play $\bm{\rho}$ in $G'$ that visits $u^\forall_{k+2}$ is $\exists$-winning while an identical play $\bm{\pi}$ that replaces all occurrences of $u^\forall_{k+2}$ in $\bm{\rho}$ with $u^\forall_{k+1}$, is $\forall$-winning.\footnote{Note here that $\bm{\pi}$ is indeed a play in $G'$, since in any play, the branch $u \to u_{k+2} \to u^\forall_{k+2} \to w$ can be replaced with $u \to u_{k+1} \to u^\forall_{k+1} \to w$ for any $w \in E(u)$.} We will show that this is not possible, which will imply the existence of an equivalent strategy to $s'$ with $n_u \leq k$, as required. 

The fact that plays $\bm{\rho}$ and $\bm{\pi}$ have different winners implies that they cycle over $u$, since they are otherwise identical. This in turn implies that $\bm{\rho}$ cycles over $u^\forall_{k+2} $ and  $\bm{\pi}$ cycles over  $u^\forall_{k+1}$. Then, $\max\{\Omega(w) \mid w \in \Inf(\bm{\rho})\} \geq 2k+2$ and $\max\{\Omega(w) \mid w \in \Inf(\bm{\pi})\} \geq 2k$, due to the priorities seen on nodes $u^\forall_{k+2} $ and $u^\forall_{k+1}$ being $2k+2$ and $2k$, respectively.  

Furthermore, since these two plays are identical except for the gadget nodes $u^\forall_{k+1}$ and $u^\forall_{k+2}$, either both of them are alternating, or both of them are eventually stabilizing. If both of them were alternating,  due to~\cref{lemma:alternating-play-priority}, the maximum priority visited infinitely often in both plays would be the same value, larger than $2k+2$, corresponding to the maximum $\beta$ value seen infinitely often in the base sequence $\bm{\rho}|_V = \bm{\pi}|_V$. 

So, we know that both plays are eventually stabilizing. In this case, by~\cref{lemma:stabilizing-p-b} we get that that maximum priority seen infinitely often in both plays is at most $2k+2$. This entails
 $\max\{\Omega(w) \mid \Inf(\bm{\rho})\} = 2k+2$ and $\max\{\Omega(w) \mid \Inf(\bm{\pi})\} = 2k+1$. The later is due to $\max\{\Omega(w) \mid w \in \Inf(\bm{\pi})\} \geq 2k$ and $\bm{\pi}$ being $\forall$-winning.

As $\max\{\Omega(w) \mid \Inf(\bm{\pi})\} = 2k+1$, the play $\bm{\pi}$ cycle over a node with priority $2k+1$. As they agree on all nodes except $u^\forall_{k+2} $ and $u^\forall_{k+1}$, $\bm{\rho}$ and $\bm{\pi}$ both cycle over this node. The priority $2k+1$ is signalled by two types of nodes in $G'$: (i) the rightmost branch of the existential gadget of a gadget node $w$ with $b:= \exists$, and (ii) on the successor $w^\exists_{k+1}$ of a universal gadget node $w$. 
Note that all plays that pass through $u$ and a node of type (i) infinitely often are alternating, since in $u$, $b := \forall$, and passing a node of type (i) sets $b$ to $\exists$. 
Since we know $\bm{\rho}$ and $\bm{\pi}$ are eventually stabilizing, they must visit a node of type (ii) infinitely often. Note that a node $w^\exists_{k+1}$ of type (ii) occurs in $\bm{\rho}$ and $\bm{\pi}$ iff $w$ is a universal gadget node
with $n_w = k+1$, with $b = \forall$ (since the plays are eventually stabilizing). Next we finalize the proof by showing the impossibility of this scenario. 

We show that $G'_{s'}$ does not contain any cycles $\mathcal{C}$, such that $\mathcal{C}$ contains
a node of type (ii) with $b = \forall$, and no nodes of type (i)\textemdash so, the $b$ value is $\forall$ for all nodes in $\mathcal{C}$. The absence of such cyles in $G'_{s'}$ implies that any play that passes through both $u$ (which has $b = \forall$) and a node of type (ii) infinitely often, must also pass through a node of type (i) infinitely often. In the previous paragraph, we showed that $\bm{\rho}$ and $\bm{\pi}$ do not pass through a node of type (i) infinitely often. Therefore we get that $\bm{\rho}$ and $\bm{\pi}$ do not pass through a node of type (ii) infinitely often either, obtaining a contradiction and finalizing the proof of~\cref{lemma:nu-regularity-for-forall-player}.

Let us show the absence of such cycles: Assume to the contrary that such a cycle exists in $G'_{s'}$.
Let $u_1 = (v_1, p_1, \forall), $ $,\ldots, u_c=(v_c, p_c, \forall)$ be the universal gadget nodes in $G'_{s'}$ where $n_{u_i} = k+1$ 
for each $i \in \{1, \ldots, c\}$, 
and each $(u_i)^\exists_{k+1}$ lies on a cycle $\mathcal{C}_i$ in $G'_{s'}$ that 
does not contain a node of type (i)\textemdash so, the $b$ value is $\forall$ for all nodes in $\mathcal{C}_i$. 

As $\mathcal{C}_i$ is $\exists$-winning (since it is a cycle in $G'_{s'}$) and $(u_i)^\exists_{k+1}$ has priority $2k+1$, we know that $\mathcal{C}_i$ contains a node with priority $2k+2$, which comes from the rightmost edge of a universal gadget node $w = (w', p', \forall)$. Furthermore, $n_w = k+1$, because if it was not, then $G'_{s'}$ would contain a cycle $\mathcal{C'}_i$ that is identical to $\mathcal{C}_i$ except that it passes through $w^\forall_{k+1}$ instead of $w^\forall_{k+2}$, and thus has the maximum priority $2k+1$. 
Then, $(u_i)^\exists_{k+1}$ in fact lies on a cycle with the rightmost edge (i.e. $(u_j)^\forall_{k+2}$) of some $u_j$ for $j \in \{1, \ldots, c\}$. 

Imagine a directed graph on $c$ nodes, labeled from $1$ to $c$, which contains an edge from node $i$ to node $j$ iff there exists a cycle in $G'_{s'}$ that does not contain any nodes of type (i), and contains $(u_i)^\exists_{k+1}$ and $(u_j)^\forall_{k+2}$. Since each $(u_i)^\exists_{k+1}$ lies on a cycle $\mathcal{C}_i$ (that is devoid of type (i) nodes) with the rightmost branch $(u_j)^\exists_{k+2}$ of some node $u_j$, each node in this graph has an outgoing edge. Then, there is a cycle in the graph, say $i^1 \to i^2 \to \cdots \to i^l \to i^1$. This implies that there is a (non-alternating) cycle in $G'_{s'}$ that passes through $(u_{i^m})^\exists_{k+1}$ for all $m \in \{1, \ldots, l\}$, formed by the universal player strategy that takes the edge $u_{i^m} \to (u_{i^m})_{k+1}$ for all $u_{i^m}$ with $m \in \{1, \ldots, l\}$. This forms an odd cycle in $G'_{s'}$, yielding a contradiction. 

\end{proof}

\noindent \textbf{Structure of $s'$:} Without loss of generality, we can restrict attention to a positional winning $\exists$-strategy $s'$ in $G'$ that obeys the restrictions stated in~\cref{cor:lemma-regularity} and~\cref{lemma:nu-regularity-for-forall-player}. 

Next, we construct an $\exists$-strategy $s$ in $G$, based on $s'$.

%% file: construction-of-s.tex
In this section, we construct an $\exists$-strategy $s$ in $G$.
 In the rest of the proof we will fix an arbitrary $\forall$-strategy $t$ and use $s$ to show the existence of a counter $\exists$-strategy, s.t. the play starting at $v$ and compliant with both strategies is $\exists$-winning. Showing this will imply $v \in \Win_\exists(G)$ by the determinacy of fair parity/parity games. To show the existence of a counter $\exists$-strategy, we will go to a case distinction on the play $\rho = \play_v(s, t)$: We will show that either $\rho$ is $\exists$-winning\textemdash so $s$ is a counter strategy to $t$, or we can iterate $s$ to obtain a counter strategy to $t$.


 \smallskip
  \textbf{Idea behind the construction of $s$}. Intuitively, we want to construct $s$ so that on every play, an $s$-compliant play in $G$ acts like the projection of an $s'$-compliant play in $G'$, in the sense that the base sequence of the $s'$-compliant play is exactly the $s$-compliant play. For this, we construct $s$ based on the choices $s'$ makes, similar to the proof of~\cref{thm:fairParityToParity}. However, this time the nodes in $G'$ carry additional memory values $p$ and $b$, whereas the nodes in $G$ do not have such values. Our strategy $s$ should also take into account these memory values while deciding the successor of an $\exists$-node in $G$, since $s'$ can prescribe different successors to $(u, p, b)$ and $(u, p', b')$. For this, we will equip $s$ with a memory that mimics the $p$ and $b$ changes in $G'_{s'}$. However, there is one catch here: There can be different $s'$-plays in $G'$, say $\bm{\rho}^1$ and $\bm{\rho}^2$, that have an identical projection to $G$, i.e. $\bm{\rho}^1|_{V} = \bm{\rho}^2|_V$. That is, $\bm{\rho}^1$ and $\bm{\rho}^2$ pass through the same base nodes, but with different $p$ and $b$ values. This can happen because even though $G'_{s'}$ fixes the $\exists$-player choices, $\forall$-player can take different gadget branches to reach the same base nodes, with different $p$ and $b$ values. To overcome this issue, as in the proof of~\cref{thm:fairParityToParity}, we will restrict our attention to a subgame of $G'_{s'}$ with limited choices for the $\forall$-player, and obtain $LG'_{s'}$.
  
  After obtaining $LG'_{s'}$, we will define a natural (unique) \enquote{expansion} of an $s$-compliant play $\rho$ to $LG'_{s'}$. This expansion will of course be an $s'$-compliant play (thus, $\exists$-winning), and will provide us with a canonical way to update the $p$ and $b$ values in $G$. 

  
  {\bf Constraining $G'_{s'}$ to $LG'_{s'}$}. We will constrain $G'_{s'}$ to a subgame by limiting the choices of the $\forall$-player on universal gadgets. For every universal gadget encountered in $G'_{s'}$, we limit the choices of a universal gadget node $u = (w, p, b) \in V'_{\forall}$ (where $w \in \Vfair_{\forall}$) only to the branches $u \to u_{n_{u}}$ and $u \to u_{n_u+1}$\textemdash both of which always exist. So, we remove all the other branches of $u$ out of $G'_{s'}$. We call the remaining subgame $LG'_{s'}$, standing for \emph{limited $G'_{s'}$}. Note that as $LG'_{s'}$ is a subgame of $G'_{s'}$, it is still $\exists$-winning.

  \textbf{About the expansion of $\rho$}. Next, we want to use $LG'_{s'}$ to define a unique \enquote{expansion} of a play $\rho$ in $G$ to $G'_{s'}$. 
  Note that for a universal gadget node $u = (w,p, b)$, the only successor of $u_{n_u}$ in $G'_{s'}$ is $u^\exists_{n_u}$, and the only successor of $u_{n_u + 1}$ in $G'_{s'}$ is $u^\forall_{n_u+1}$. Therefore, in $LG'_{s'}$ , the only branches rooted at $u$ are
  \begin{align}
   & u \to u_{n_u} \to u^\exists_{n_u} \to s'(u^\exists_{n_u}) \tag{branch 1 - $\forall$}
   \label{branch1-pp-uni}\\
   & u \to u_{n_u+1} \to u^\forall_{n_u+1} \to E(u) \tag{branch 2 - $\forall$}\label{branch2-pp-uni}
  \end{align}
 As $E(u)|_V$ contains all successors of $w = u|_V$, in every play that contains $w \in \Vfair_{\forall}$, and some  successor $w'$ of $w$, we can take~(\ref{branch2-pp-uni}) to $w'$, i.e. the path $u \to u_{n_u+1} \to u^\forall_{n_u+1} \to w'$ in the expansion. Moreover, if $w' =s'(u^\exists_{n_u}) $, we can take~(\ref{branch1-pp-uni}) to $w'$ as well. In the expansion of $\rho$, we want to take~(\ref{branch1-pp-uni}) if the successor $w'$ of a $w \in \rho$ is $ s'(u^\exists_{n_u})$, and take~(\ref{branch2-pp-uni}) otherwise. An edge $w \to w'$ in $\rho$ where $w \in \Vfair_\forall$ will appear in the expansion of $\rho$ as~(\ref{branch1-pp-uni}) or~(\ref{branch2-pp-uni}), according to the above-defined rule. 

   What about an edge $w \to w'$ in $\rho$ where $w \in \Vfair_\exists$? Recall that for $u = (w, p, b) \in V'_\exists$ (where $w \in \Vfair_\exists$), the only successor of $u$ in $G'_{s'}$ (and therefore in $LG'_{s'}$) is $u_{n_u}$. Thus, if $n_u \leq k$, there are two paths rooted at $u$ in $LG'_{s'}$: 
   \begin{align}
   & u \to u_{n_u} \to u^\exists_{n_u} \to s'( u^\exists_{n_u})\tag{branch 1 - $\exists$}\label{branch1-pp-exi}\\
   & u \to u_{n_u} \to u^\forall_{n_u} \to E_f(u)\tag{branch 2 - $\exists$}\label{branch2-pp-exi}
   \end{align}
    If $n_u = k+1$, then only~(\ref{branch2-pp-exi}) exists in $LG'_{s'}$. By the construction of $s$, we will make sure that $w'$ is always in $E_f(w) \cup s'( u^\exists_{n_u})$, and in case $n_u = k+1$ that $w' = s'( u^\exists_{n_u})$. 
    Then, as before, if $w' =  s'( u^\exists_{n_u})$, the expansion includes~(\ref{branch1-pp-exi}); otherwise, it will take~(\ref{branch2-pp-exi}). 

  For an edge $w \to w'$ in $\rho$ where $w \not \in \Vfair$, since $u = (w, p, b)$ is not a gadget node, the expansion will take the only successor of $u$ in $LG'_{s'}$ with base node $w'$.
  
  This way, for every $s$-compliant play $\rho$ in $G$, we get a canonical \enquote{expansion} to $LG'_{s'}$. Below, we will define the expansion formally, after constructing $s$. For the sake of introducing $s$, we would like the reader to observe that the expansion given above, as it is unique for every $s$-compliant play, provides us with a consistent way of updating the memory values $p$ and $b$ in $G$. 

 \textbf{Construction of $s$}.
Let $w \in \Vfair_\exists$ and $u = (w, p, b)$ be the existential gadget node in $G'$ with memory values $p$ and $b$. If 
 $s'$ prescribes $u$ its rightmost successor, i.e. $s'(u) = u_{k+1}$, then $s$ prescribes $w$ its successor $s'(u^\exists_{k+1})|_V$ whenever $w$ is reached in $G$ with memory values $p$ and $b$.\footnote{The memory updates are detailed below.}
  We will call such node-memory instances, where $s'$ prescribes the rightmost successor, \enquote{unfair $\exists$-instances}.

Otherwise, i.e. if $s'$ prescribes the existential gadget node $u$ any other successor, then whenever a play reaches $w$ with the memory values $p$ and $b$, $s$ prescribes $w$ all its successors in $\{s'(u^\exists_i)|_V \}\cup E_f(w)$ in a cyclic fashion,\footnote{The order in which the successors are prescribed by $s'$ does not matter. It only matters that \textit{whenever $w$ is seen infinitely often with memory values $p$ and $b$ in a play $\rho$ in $G$, all its successors in $\{s'(u^\exists_i)|_V \}\cup E_f(w)$ are also seen infinitely often in $\rho$}.} using some local memory.

Now let $w \in V_\exists \setminus \Vfair_\exists$ and $(w,p,b)$ be the node in $G'$ with memory values $p$ and $b$. Then $(w,p,b)$ is not a gadget node. In this case, $s$ simply prescribes $s(w) = s'(w, p, b)|_V$ whenever  $w$ is reached in $G$ with memory values $p$ and $b$. 

\bigskip

Now we explain how the memory values $p$ and $b$ are updated during an $s$-compliant play $\rho$ in $G$. 

\smallskip
 \textbf{The memory updates in $G$}. In the beginning of every play $\rho$ compliant with $s$ in $G$, the memory values $p$ and $b$ are initialized with $1$ and $\exists$, respectively. Throughout the play, $p$ is updated in each step to be the maximum of the $p$ value of the previous node and the $\beta$ priority of the current node, except when the current instance is an unfair $\exists$-instance and $b = \forall$. In this case, the $p$ value  
is reset to $1$. The memory value $b$ is set to $\exists$ during unfair $\exists$-instances and 
is set to $\forall$ whenever the previous instance is a fair universal node $w$ with memory values are $p, b$ such that $n_u = k+1$ for $u = (w, p, b)$, and the current base node is different from $s'(u^\exists_{n_u})|_V$.
Otherwise, $b$ is left unchanged. It can be observed that the memory updates in an $s$-compliant play $\rho$ that start from $w$ mimic the change of $p$ and $b$ values in the expansion of $\rho$ that start from $(w, 1, \exists)$.


\enskip 

\textbf{The expansion of an $s$-compliant play in $G$ to $LG'_{s'}$}. A play $\bm{\rho}$ in $LG'_{s'}$ is said to be the \emph{expansion} of a play $\rho=v_1 v_2 \ldots $ compliant with $s$ in $G$ to $LG'_{s'}$ if, 
$\bm{\rho}$ starts from $(v_1, 1, \exists)$ and for each edge $(v_i, v_{i+1})$ in $\rho$, 
\begin{itemize} 
    \item If $v_i \not \in \Vfair$ with memory values $p, b$, we extend $\bm{\rho}$ with $(v_i, p, b) \to (v_{i+1}, p', b)$ appropriately,
    \item If $v_i \in \Vfair_\exists$ with memory values $p, b$, for $u = (v_i, p, b)$, in case $v_{i+1} = s'(u_{n_u}^\exists)|_V$, we extend $\bm{\rho}$ with $u \to u_{n_u} \to u^\exists_{n_u} \to (v_{i+1}, p', b')$ appropriately. 
    Otherwise, i.e. in case $v_{i+1} \neq s'(u_{n_u}^\exists)|_V$, by construction of $s$, $v_{i+1} \in E_f(v_i)$, In this case, we extend $\bm{\rho}$ with $u \to u_{n_u} \to u^\forall_{n_u} \to (v_{i+1}, p', b')$ appropriately. 
    \item If $v_i \in \Vfair_\forall$ with memory values $p, b$, for $u = (v_i, p, b)$, in case $v_{i+1} = s'(u_{n_u}^\exists)|_V$, we again extend $\bm{\rho}$ with $u \to u_{n_u} \to u^\exists_{n_u} \to (v_{i+1}, p', b')$ appropriately. 
    Otherwise, i.e. in case $v_{i+1} \neq s'(u_{n_u}^\exists)|_V$, then we extend $\bm{\rho}$ with $u \to u_{n_u+1} \to u^\forall_{n_u+1} \to (v_{i+1}, p', b')$ appropriately. Note that $u^\forall_{n_u+1}$ always exists since for universal gadgets $n_u$ is at most $k+1$.
\end{itemize}
Note that $\bm{\rho}$ restricted to nodes in $V$, is exactly $\rho$, i.e.  $\bm{\rho}|_V = \rho$. Furthermore, $\bm{\rho}$ and $\rho$ have consistent memory values, i.e. if in $\rho$ a node $w$ is seen with memory values $p$ and $b$, then 
the corresponding node of $w$ in $\bm{\rho}|_V$ is $(w, p, b)$. We will call $\rho$ the \emph{projection} of $\bm{\rho}$ to $G$. 

We would like to note that the choice of memory values $p :=1, b:\exists$ of the initial node of $\bm{\rho}$ is arbitrary, and can be chosen as any other value pair as needed.

Finally, observe that $\bm{\rho}$ is contained in $LG'_{s'}$, regardless of the $p$ and $b$ values of the initial node of the expansion.

\begin{notation}\label{not:bold-version} For an $s$-compliant play in $G$ shown by a latin letter, we will denote its expansion to $G'$ via the bold version of the same letter, such as $\rho$ and $\bm{\rho}$, or $\iota$ and $\bm{\iota}$. 
\end{notation}

%% file: parity-parity-to-parity-main-proof.tex
In this section we will show that for every $\forall$-strategy $t$ in $G$, there exists a counter $\exists$-strategy such that every play that starts at $v$ and is compliant with both strategies is $\exists$-winning. Since fair parity/parity games are determined, this implies that $v$ is $\exists$-winning in $G$, i.e. $v \in \Win_\exists(G)$, concluding the proof of~\cref{thm:fairParityPlusParityToParity}

For this, we fix an arbitrary $\forall$-strategy $t$ in $G$. Then we take a play $\rho$ that starts at $v$, and is compliant with $s$ and $t$, and go to a case distinction on the behaviour of $\rho / \bm{\rho}$. 

In the first two cases,~\ref{item:case1} and \ref{item:case2}, we will show that $\rho$ is $\exists$-winning\textemdash so, $s$ is the counter $\exists$-strategy. In~\ref{item:case3}, we will construct an $\exists$-winning play $\nu$ that starts from $v$ in $G$ and is compliant with $t$. We will not explicity construct the $\exists$-strategy $s^\circ$ used to construct $\nu$, but the proof will use the construction of $s$ iteratively to obtain $s^\circ$, by manipulating the memory values of $s$ to jump between different nodes in $LG'_{s'}$. To construct $\nu$ in~\ref{item:case3}, we will use~\ref{item:case1} and \ref{item:case2} as subroutines. 

\enskip
\begin{enumerate}[label=\textbf{Case }\arabic*., leftmargin=*, itemsep=0.5em]
    \item $\rho$ is $\forall$-unfair. \label{item:case1}
    \item $\rho$ is $\forall$-fair and $\bm{\rho}$ is eventually stabilizing.\label{item:case2}
    \item $\rho$ is $\forall$-fair and $\bm{\rho}$ is alternating. \label{item:case3}
\end{enumerate}

\enskip

\noindent \ref{item:case1} ($\rho$ is $\forall$-unfair). If $\rho$ is $\exists$-fair, then it is automatically $\exists$-winning. Assume $\rho$ is $\exists$-unfair. 
By construction of $s$, if $\rho$ is $\exists$-unfair then there exists an existential gadget node $u = (w, p, b)$ with $w \in \Vfair_\exists$ such that the rightmost branch of $u$ is taken infinitely often in $\bm{\rho}$ (since otherwise $s$ cycles through $E_f(w)$). 
First assume that in $\bm{\rho}$, the rightmost branch of a universal gadget is visited infinitely often as well. Then, $\bm{\rho}$ is alternating. By~\cref{lemma:alternating-play-priority}, the winner of $\bm{\rho}$ is determined by the maximum $\beta$ priority in $\Inf(\bm{\rho}|_V) = \Inf(\rho)$. As $\rho$ is mutually unfair, the winner of $\rho$ is determined by the maximum $\beta$ priority in $\Inf(\rho)$ as well. Since $\bm{\rho}$ is $\exists$-winning, $\rho$ is also $\exists$-winning.

Now assume in $\bm{\rho}$ the rightmost branch of no universal gadget is visited infinitely often. Then, $\bm{\rho}$ is eventually stabilizing. By~\cref{lemma:stabilizing-p-b}, the value of $b$ stabilizes in $\bm{\rho}$, and the maximum priority seen infinitely often is at most $2k+2$.
As the rightmost branches of existential nodes are visited infinitely often, $b$ stabilizes at $\exists$ and the priority $2k+1$ is signalled infinitely often. On the other hand, as the righmost branches of universal gadgets are not visited infinitely often, the priority $2k+2$ is not signalled infinitely often. This implies the maximum priority seen infinitely often in $\bm{\rho}$ is $2k+1$, yielding a contradiction as $\bm{\rho}$ is $\exists$-winning.  

\enskip

\noindent  \ref{item:case2} ($\rho$ is $\forall$-fair and $\bm{\rho}$ is eventually stabilizing). 
First assume\textemdash towards a contradiction\textemdash that
there exists an existential gadget node $u$ whose rightmost branch is taken infinitely often in $\bm{\rho}$. 
As $\bm{\rho}$ is eventually stabilizing and rightmost branches of existential gadgets are visited infinitely often, the value of $b$ stabilizes at $\exists$. From here on, identical arguments to the second paragraph of~\cref{item:case1} follows, yielding a contradiction.


Therefore, the righmost edges of existential gadgets are not visited infinitely often in $\bm{\rho}$. By the construction of $s$, this implies that $\rho$ is $\exists$-fair. 


We focus our attention on $\bm{\tilde{\rho}}$, i.e. the infinitely repeating tail of $\bm{\rho}$. Let $G'_{\bm{\tilde{\rho}}}$ be the subgame of $LG'_{s'}$ restricted to the nodes and edges taken in $\bm{\tilde{\rho}}$. Let $\bar{p}$ and $\bar{b}$ be the unique $p$ and $b$ values 
in $\bm{\tilde{\rho}}$, whose existence is guaranteed by~\cref{lemma:stabilizing-p-b}, as $p$ and $b$ eventually become constants in $\bm{\rho}$. Then for each node $w \in \bm{\tilde{\rho}}|_V$ a unique node $(w, \bar{p}, \bar{b})$ exists in $G'_{\bm{\tilde{\rho}}}$. 
The arguments we use from here until the end of the proof~\ref{item:case2} closely follow the arguments from the proof of~\cref{thm:fairParityToParity}, part (ii) of side $\Rightarrow$.

Recall that $\bm{\tilde{\rho}}$ is $\exists$-winning. Let $m$ be the maximum (even) priority seen in $\bm{\tilde{\rho}}$\textemdash and therefore in $G'_{\bm{\tilde{\rho}}}$. We will show the existence of a base node in $\bm{\tilde{\rho}}$ with the $\alpha$ priority $m$. This will in turn imply that the maximum $\alpha$ priority in $\tilde{\rho}$ is also $m$. Recall that as $\rho$ is mutually fair, the $\alpha$ priorities determine the winner. Therefore, this will show that $\rho$ is $\exists$-winning, concluding the proof.


First observe that every subgraph $S$ of $G'_{\bm{\tilde{\rho}}}$ that is an end-component (i.e., when $\bm{\tilde{\rho}}$ reaches a node in $S$, it stays in $S$) is equal to $G'_{\bm{\tilde{\rho}}}$, since $G'_{\bm{\tilde{\rho}}}$ is already the infinite tail of a play.
To prove the existence of a base node with priority $m$ in $\bm{\tilde{\rho}}$, we will assume that no such base node exists in $G'_{\bm{\tilde{\rho}}}$, and show the existence of an end component $S$ of 
 $G'_{\bm{\tilde{\rho}}}$ that is devoid of gadget nodes with priority $m$. This will show that $G'_{\bm{\tilde{\rho}}}$ has a base node with priority $m$, concluding the proof.

Whenever a node $w \in \Vfair_\exists$ is visited infinitely often in $\rho$, the successor $u_{n_u}$ of the existential gadget node $u = (w, \bar{p}, \bar{b})$ is visited infinitely often in $\bm{\tilde{\rho}}$, as $u_{n_u}$ is the only successor of $u$ in $LG'_{s'}$.
 Furthermore, due to $\rho$ being $\exists$-fair, both $s'(u^\exists_{n_u})$ and all the successors in $E_f(u)$ are visited infinitely often in $\bm{\rho}$. Therefore, both outgoing branches~(\ref{branch1-pp-exi}) and~(\ref{branch2-pp-exi}) of $u$ are in $G'_{\bm{\tilde{\rho}}}$.%
 \footnote{Unless, $E_f(w) = \{s'(u^{\exists}_{n_u})|_V\}$, in which case only~(\ref{branch1-pp-exi}) is taken.} 
 Also whenever the node $u^\forall_{n_u}$ is in $G'_{\bm{\tilde{\rho}}}$, the node $u^\exists_{n_u}$ is also in $G'_{\bm{\tilde{\rho}}}$.
Similarly, due to $\rho$ being $\forall$-fair, whenever a node $w \in \Vfair_\forall \cap \Inf(\rho)$, 
both the branches~(\ref{branch1-pp-uni}) and~(\ref{branch2-pp-uni}) are in $G'_{\bm{\tilde{\rho}}}$.
\footnote{Unless, $E_f(w) = \{s'(u^{\exists}_{n_u})|_V\}$ and $t(w)$ is eventually always the unique fair successor of $w$, in which case only~(\ref{branch1-pp-uni}) is taken.}
So, whenever the node $u^\forall_{n_u+1}$ is in $G'_{\bm{\tilde{\rho}}}$, the node $u^\exists_{n_u}$ is also in $G'_{\bm{\tilde{\rho}}}$. 

Assume there are no base nodes that carry priority $m$ in $G'_{\bm{\tilde{\rho}}}$. Then, there is a gadget node in $G'_{\bm{\tilde{\rho}}}$ that carries this priority. Recall that for all $u$ in $G'$ and all $i$, $u^\exists_i$ has an odd priority and $u^\forall_i$ has an even priority. 
Given the above characteristics of $G'_{\bm{\tilde{\rho}}}$, 
we can guarantee that whenever there is a gadget node carrying the maximum even priority $m$, it is either $u^\forall_{n_u}$ in an existential gadget or $u^\forall_{n_u+1}$ in a universal gadget. In both cases, there exists an existential 
sibling $u^\exists_{n_u}$ of this gadget node in  $G'_{\bm{\tilde{\rho}}}$, carrying the odd priority $m-1$. 
So, we can remove the branches of the (existential or universal) gadget that carry priority $m$ in $G'_{\bm{\tilde{\rho}}}$, and get a subgame $H$ of $G'_{\bm{\tilde{\rho}}}$ without dead-ends. 
Since $H$ is a subgame of $LG'_{s'}$, all infinite plays in $H$ are $\exists$-winning. Furthermore, due to the maximum priority in $H$ being $m-1$, we know that none of the gadget nodes $u$ which have successors with priorities $m$ and $m-1$ in $G'_{\bm{\tilde{\rho}}}$ can lie on a cycle in $H$.
We focus on the cyclic part $\Cyc(H)$ of $H$. Since none of the nodes in $\Cyc(H)$ carry priorities $m$ or $m-1$, all nodes in $\Cyc(H)$ have the same outgoing edges in $G'_{\bm{\tilde{\rho}}}$ and in $\Cyc(H)$. 
Then, once $\bm{\tilde{\rho}}$ reaches a node in $\Cyc(H)$, it stays in $\Cyc(H)$, i.e. $\Cyc(H)$ is an end-component of $G'_{\bm{\tilde{\rho}}}$ without priority $m$ gadget nodes. As previously discussed, this concludes the proof. 

\enskip

\noindent  \ref{item:case3} ($\rho$ is $\forall$-fair and $\bm{\rho}$ is alternating). The rest of this section will be devoted to the proof of this case. We will inductively construct a $t$-compliant play $\nu$ that starts at $v$ in $G$ and its (roughly speaking) expansion $\bm{\nu}$ to $LG'_{s'}$. We will show that either during the construction of $\nu$ and $\bm{\nu}$, we fall into the subcases~\ref{item:case1} and~\ref{item:case2}, which directly ends the proof showing that $t$ is not $\forall$-winning in $G$; or the inductively constructed play $\nu$ is not $\forall$-winning in $G$.

Before we start the proof of \ref{item:case3}, we introduce some notation and make observations that will be used in the proof.

\enskip 

\begin{notation}
	We denote the subgame of $LG'_{s'}$ where the rightmost edges of all universal gadgets are removed, by $LG^{-}_{s'}$.	
\end{notation}

That is, in $LG'_{s'}$ the only successor of $u$ is $u_{n_u}$ (whose only successor is $u^\exists_{n_u}$).

\begin{lemma}\label{lemma:Gminus-cycles}
	No cycle of $LG^{-}_{s'}$ contains a universal gadget node $u$ with $n_u = k+1$.
\end{lemma}

\begin{proof}[Proof of~\cref{lemma:Gminus-cycles}]
Observe that all infinite plays in $LG^{-}_{s'}$ are eventually stabilizing. This implies that the maximum priority on cycles in $LG^{-}_{s'}$ is at most $2k+1$ by~\cref{lemma:stabilizing-p-b}, and the absence of priority $2k+2$ gadget nodes, due the absence of righmost edges of universal gadgets. 
The priority $2k+1$ is visited on a cycle in case a universal gadget branch~(\ref{branch1-pp-uni}) exists in $LG^{-}_{s'}$, caused by a universal node $u$ with $n_u = k+1$. Recall that such a $u$ has two branches~(\ref{branch1-pp-uni}) and~(\ref{branch2-pp-uni}) in $LG'_{s'}$, but only~(\ref{branch1-pp-uni})remains in $LG^{-}_{s'}$, since the rightmost branch is removed. 
So, if a node $u$ with $n_u = k+1$ exists on a cycle in $LG^{-}_{s'}$, then the maximum priority $2k+1$ is seen on this cycle. Such a cycle cannot exist since all plays in $LG^{-}_{s'}$ are $\exists$-winning. This disallows the existence of universal nodes $u$ with $n_u = k+1$ on cycles of $LG^{-}_{s'}$.
\end{proof}

\begin{observation} If $t$ wins a node $w$ in $G$, we can WLOG assume it does so with every finite history $h$ that ends at $w$. 
\end{observation}

\begin{observation}\label{obs:unique-elongation-for-t-s} For every given $\exists$-strategy $s^i$, and for every given $t$-compliant history $h$ in $G$, there exists a
	unique extension of $h$ that's compliant with both $t$ and $s^i$. That is, there is a 
	unique play $h \cdot \sigma$
	in $G$, such that $\sigma$ is compliant with both $s^i$ and $t$, and $h \cdot \sigma$ is compliant with $t$. 
\end{observation}


Now we start the proof of~\ref{item:case3} 

\enskip

\noindent \textbf{The proof of~\ref{item:case3}} We construct a $t$-compliant play $\nu$ in $G$ and a play $\bm{\nu}$ in $LG'_{s'}$. By slight abuse of Notation~\ref{not:bold-version}, we denote these plays by $\nu$ and $\bm{\nu}$, even though $\bm{\nu}$ is not exactly the expansion of $\nu$. However, to justify our abuse of notation, these two plays will exhibit very similar behavior to a play and its expansion. This behaviour, which is an invariant at each iteration of the construction algorithm for $\nu$ and $\bm{\nu}$ is presented via Inv.~\ref{inv:nu-bold-nu} and~\cref{lemma:nu-bold-nu}. We present these below, but postpone their proof. 
Note that $\bm{\nu}$ is in $LG'_{s'}$ and not necessarily in $LG_{s'}^-$. 

The plays $\nu$ and $\bm{\nu}$ will be constructed in iterations. In each iteration $i$ there are prefixes $\nu^i$ and $\bm{\nu}^i$ and infinite plays $\nu^i\cdot \rho^i$ and $\bm{\nu}^i\cdot \bm{\rho^i}$ and we identify a prefix $\nu^{i+1}$ of $\nu^i\cdot \rho^i$ that extends $\nu^i$ and a prefix $\bm{\nu}^{i+1}$ of $\bm{\nu}^i\cdot \bm{\rho^i}$ that extends $\bm{\nu}^i$ such that $\nu= \lim_{i \to \infty} \nu^i$ and $\bm{\nu} = \lim_{i \to \infty} \bm{\nu}^i$. The following invariant will hold for all $i \in \mathbb{N}$. 

\begin{invariant}\label{inv:nu-bold-nu} The set of nodes visited in $\nu^i$ is exactly the set of base nodes visited in $\bm{\nu}^i$ (i.e. even though $\nu^i = \bm{\nu}^i|_V$ does not necessarily hold, $\nu^i = {\bm{\nu}^i}^{\{\}}|_V$ holds.\footnote{See Not.~\ref{not:base-node-base-sequence}}). 
\end{invariant}

Furthermore, the infinitely often visited base nodes of $\nu$ and $\bm{\nu}$ are identical, as stated in~\cref{lemma:nu-bold-nu}. 

\begin{lemma}\label{lemma:nu-bold-nu}
	The set of infinitely often visited nodes in $\nu$ are exactly the infinitely often visited base nodes in $\bm{\nu}$, i.e. $\Inf(\nu) = \Inf(\bm{\nu}|_V)$.
\end{lemma}

In the proof we will exploit the similarity of $\nu$ and $\bm{\nu}$ stated via Inv.~\ref{inv:nu-bold-nu} and~\cref{lemma:nu-bold-nu}. The similarity, in addition to $\bm{\nu}$ being a play in $LG'_{s'}$ starting at $(v, 1, \exists)$, therefore $\exists$-winning; while $\nu$ being a $t$-compliant play in $G$ starting at $v$, will allow us to prove that $\nu$ is $\exists$-winning.




\enskip 

 Next, we formally construct $\nu$ and $\bm{\nu}$. 
Afterwards, we prove Inv.~\ref{inv:nu-bold-nu} and~\cref{lemma:nu-bold-nu}, and go to a case distinction on the structure of $\nu$ and $\bm{\nu}$, the last part of the proof. 

\enskip

\noindent \textbf{The construction of $\nu$ and $\bm{\nu}$} 

\smallskip
The construction will follow the below structure:
\begin{itemize}
    \item \textbf{Invariants: } We will introduce the invariants about $\rho^{i},\bm{\rho}^{i}, \nu^i $ and $\bm{\nu}^i$ that will be preserved at each iteration $i$.
    \item \textbf{Base case (Initialization): } We will initialize the construction with values for $\rho^0, \bm{\rho}^0, \nu^0$ and $\bm{\nu}^0$. We will show that the invariants hold in the base case. 
    \item \textbf{Inductive Step (Construction): } Assuming the invariants hold at iteration $i-1$ (i.e. they hold for $\rho^{i-1}, \bm{\rho}^{i-1}, \nu^{i-1}$ and $\bm{\nu}^{i-1}$), we will show how to construct the values at iteration $i$ (i.e. $\rho^{i}, \bm{\rho}^{i}, \nu^{i}$ and $\bm{\nu}^{i}$). We will show that the plays constructed at iteration $i$ preserve the invariants. 
    \begin{itemize}
        \item \textbf{Selection of $\bm{u}^i$: } We will identify a node $\bm{u}^i$ in $\bm{\rho}^{i-1}$, through a list of properties $\bm{u}^i$ must satisfy.
        \item  \textbf{Construction of the plays at iteration $i$ : } Assuming we were able to select a $\bm{u}^i$ satisfying the listed properties, we will construct the new values of $\rho^{i}, \bm{\rho}^{i}, \nu^{i}, \bm{\nu}^{i}$. While constructing these plays, we will simultaneously show that they preserve the invariants.
        \item \textbf{Existence of $\bm{u}^i$: } Lastly, we will show that a $\bm{u}^i$ satisfying the listed properties always exists.
    \end{itemize}
\end{itemize}

\enskip

\noindent \textbf{-- Invariants: } Inv. $\ref{inv:set}$ lists the invariants of the construction. 

\begin{invariant}\label{inv:set} At each iteration $i \in \mathbb{N}^+$,
	\begin{enumerate}[label=(\roman*)]
		\item $\nu^i$ is a finite $t$-compliant play in $G$.\label{inv:set-one}
		\item $\bm{\nu}^i$ is a finite play in $LG'_{s'}$.\label{inv:set-two}
		\item $\nu^i\cdot \rho^i$ is a(n infinite) $t$-compliant play in $G$,
		$\bm{\nu}^i\cdot \bm{\rho}^i$ is a play in $LG'_{s'}$, and $\bm{\rho}^i$ is the expansion of $\rho^i$ to $LG'_{s'}$.\label{inv:set-three}
		\item Either $\nu^i\cdot \rho^i$ is $\exists$-winning (concluding the construction and finalizing the proof, as it is a $t$-compliant play in $G$; note that winning is prefix independent), or $\rho^i$ and $\bm{\rho}^i$ fall into~\ref{item:case3}, i.e. $\rho^i$ is $\forall$-fair and $\bm{\rho}^i$ is alternating.\label{inv:set-four}
		\item $\nu^{i-1}$ is a prefix of $\nu^{i}$ and $\bm{\nu}^{i-1}$ is a prefix of $\bm{\nu}^{i}$.\label{inv:set-five}
	\end{enumerate}
\end{invariant}

\noindent \textbf{-- Base case (Initialization): } Initially, we set $\rho^0 := \rho$, $\nu^0 := \bm{\nu}^0 := \varepsilon$, and  $\bm{\rho}^0$ is the expansion of $\rho^0$ to $LG'_{s'}$ (starting from the initial vertex $(v, 1, \exists)$). 
The invariants clearly hold for the initial values of $\rho^0, \bm{\rho}^0, \nu^0$ and $\bm{\nu^0}$. 


\smallskip
\noindent \textbf{-- Inductive Step (Construction): } Assuming the invariant holds for $\rho^{i-1}, \bm{\rho}^{i-1},$ $ \nu^{i-1}$ and $\bm{\nu}^{i-1}$, we will construct the plays $\rho^{i}, \bm{\rho}^{i}, \nu^{i}$ and $\bm{\nu}^{i}$ at iteration $i \geq 1$.
For this, we need to first identify the node $\bm{u}^i$.

\smallskip
\noindent \textbf{\textbullet \,\, Identifying $\bm{u}^i$: } Let $\bm{\rho}^{i-1} = u_1 u_2 \ldots$ and let $n > 1$ be the minimal index such that $u_n = (w^i, p^i, b^i)$ satisfies the following properties. We will denote $u_n$ by $\bm{u}^i$ to express that it is \textit{the selected node at iteration $i$}.
\begin{enumerate}[label=$\bm{u}^i$-Prop.~\arabic*, leftmargin=*, align=left]
    \item $w^i \in \Vfair_\forall$ (i.e. $\mathbf{u}^i$ is the root node of a universal gadget in $G'$).\label{item:u-1}
    \item $u_{n} \to u_{n+1}$ is the rightmost edge of $\mathbf{u}^i$, i.e. $\mathbf{u}^i \to \mathbf{u}_{k+2}^i$ (the existence of the rightmost edge of the universal gadget in $\bm{\rho}^{i-1}$ implies that $n_{\bm{u}^i} = k+1$, due to the definition of expansion).\label{item:u-2}
    \newline Define $\bar{w}^i := s'((\mathbf{u}^i_{k+1})^\exists)|_V$. That is, $\bar{w}^i$ is the base node of the unique successor of $(\mathbf{u}^i_{k+1})^\exists$ in $LG'_{s'}$.\footnote{We will try to remove the rightmost edges from $\bm{\rho}^i$ as much as possible by taking the edge $\mathbf{u}^i \to \mathbf{u}^i_{k+1}$ instead of the rightmost edge $\mathbf{u}^i \to \mathbf{u}^i_{k+2}$ in universal gadgets whenever possible.} Further, define $\bm{\varsigma}^i := u_1 \ldots u_{n-1}$ to be the prefix of $\bm{\rho}^i$ 
    until $u_n$ and excluding $u_n$. 
     As usual, define $\varsigma^i$ to be the projection of $\bm{\varsigma}^i$ to $G$, i.e. $\varsigma^i = \bm{\varsigma^i}|_V$ is a play in $G$.
    \item There is a $t$-compliant extension $\nu^{i-1} \cdot \varsigma^i \cdot \kappa$ of $\nu^{i-1}$ in $G$ such that $\kappa$
    has a prefix $\mathcal{C}^{w^i}$ that is a cycle over $w^i$, and $\kappa$ continues after $\mathcal{C}^{w^i}$ with $\bar{w}^i$. That is, $\mathcal{C}^{w^i} = w^i \ldots w^i$ is a cycle over $w^i$ and $\kappa = \mathcal{C}^{w^i} \cdot \bar{w}^i \cdots$. Moreover, there exists a cycle $\bm{ \mathcal{C}^{\mathbf{u}^i}}$ over $\mathbf{u}^i$ in $LG'_{s'}$, and its set of base nodes is the same as the set of nodes of $\mathcal{C}^{w^i}$, i.e. $\{ \, v \mid v \in \bm{ \mathcal{C}^{\mathbf{u}^i}}|_V\, \} = \{ \, v \mid v \in \mathcal{C}^{w^i} \, \}$. \label{item:u-3}
\end{enumerate}

\smallskip
\noindent \textbf{\textbullet \,\,Construction of the plays at iteration $i$ : } 
Assuming we obtained $\bm{u}^i$, whose existence will be proven later, we construct $\bm{\nu}^i$ as follows. We set
\[\bm{\nu}^i := \bm{\nu}^{i-1} \cdot \bm{\varsigma}^i \cdot \bm{ \mathcal{C}^{\mathbf{u}^i}} \to \mathbf{u}_{k+1}^i \to (\mathbf{u}_{k+1}^i)^{\exists} \to \bar{\mathbf{u}}^i  \text{ where } \bar{\mathbf{u}}^i = (\bar{w}^i, \bar{p}^i, b^i)\]  with the appropriate $\bar{p}^i$.

Note that when we write the concatenation of two finite plays, such as $\bm{\nu}^{i-1}$ and $\bm{\iota}^i$ where the last vertex of the former matches the first vertex of the latter, as $\bm{\nu}^{i-1} \cdot \bm{\iota}^i$, we assume that this shared vertex is not duplicated in the resulting sequence.

Next, we show that $\bm{\nu^i}$ is indeed a finite play in $LG'_{s'}$ [Showing Inv.~\ref{inv:set}-\ref{inv:set-two} is preserved]. 
\enskip

\begin{lemma}\label{lemma:nu-i-in-gadget-game}
	The play $\bm{\nu}^i$ is a finite play in $LG'_{s'}$. 
\end{lemma}

\begin{proof}[Proof of~\cref{lemma:nu-i-in-gadget-game}]It is clear that $\bm{\nu}^i$ is finite. We will show $\bm{\nu}^i$ is a play in $LG'_{s'}$ by showing that every edge of $\bm{\nu}^i$ exists in $LG'_{s'}$. (i) By construction, $\bm{\nu}^{i-1} \cdot \bm{\varsigma}^i \cdot \bm{ \mathcal{C}^{\mathbf{u}^i}}$ is a play in $LG'_{s'}$. (ii) As $\mathbf{u}^i$ is a universal node in $G'$, all its outgoing edges are in $LG'_{s'}$. So, in particular the edge $\bm{ \mathcal{C}^{\mathbf{u}^i}} \to \bm{u}^i_{k+1}$ is in $LG'_{s'}$. (iii) According to~\ref{item:u-2}, $\bm{u}^i \to \bm{u}^i_{k+2}$ is in $G'_{\bm{\rho}^i}$, where $\bm{\rho}^i$ is the expansion of a $t$-compliant play $\rho^i$ in $G$. By construction of expansion, $\bm{u}^i \to \bm{u}^i_{k+2}$ being in $G'_{\bm{\rho}^i}$ implies $n_{\bm{u}^i} = k+1$. According to Not.~\ref{not:n}, this in turn implies that $s'$ prescribes to $\bm{u}^i$ its existential successor, i.e. the edge $\mathbf{u}_{k+1}^i \to (\mathbf{u}_{k+1}^i)^{\exists}$ is in $LG'_{s'}$. (iv) Finally, the edge $(\mathbf{u}_{k+1}^i)^{\exists} \to \bar{\mathbf{u}}^i$ is in $LG'_{s'}$ as according to~\ref{item:u-3}, $\bar{w}^i = s'((\mathbf{u}_{k+1}^i)^{\exists})|_V$; $b^i$ is left unchanged as the transition is not through a rightmost gadget branch, and $\bar{p}^i$ is set accordingly. \end{proof}

\enskip

To obtain $\nu^i$, we simply extend $\nu^{i-1}$ with $\varsigma^i \cdot \mathcal{C}^{w^i} \cdot \bar{w}^i$.

\begin{lemma}\label{lemma:nu-i-in-original-game}
The play $\nu^i$ is a finite $t$-compliant play in $G$ whose nodes are the base nodes of $\bm{\nu}^i$.
\end{lemma}

\begin{proof}[Proof of~\cref{lemma:nu-i-in-original-game}]
Due to~\ref{item:u-3}, $\nu^i$ is clearly a finite $t$-compliant play in $G$  [Showing Inv.~\ref{inv:set}-\ref{inv:set-one}]. Moreover, Inv.~\ref{inv:nu-bold-nu} is preserved, i.e. the set of nodes of $\nu^i$ is the same as the set of base nodes of $\bm{\nu}^i$ (due to how the cycle $\mathcal{C}^{w^i} $ is defined in~\ref{item:u-3}).
\end{proof}

\enskip

Lastly, we define $\rho^{i}$ and $\bm{\rho}^{i}$. We will define $\rho^{i}$ as an extension of $\nu^i$, so that $\nu^i \cdot \rho^{i}$ is a $t$-compliant play in $G$ and the expansion $\bm{\rho}^{i}$ of $\rho^{i}$, and $\bm{\nu}^i \cdot \bm{\rho}^{i}$ are plays in $LG'_{s'}$. For this, we extract an $\exists$-strategy $s^i$ in $G$, by pulling back ${s'}$, just as we did in the construction of $s$. However, while pulling $s'$ back we need to be careful about the memory update function, so that we can start $\rho^{i}$ from the last node of $\nu^i$, that is $\bar{w}$ (with memory values $p^i, b^i$ determined by the history $\nu^i$); and we can start its expansion $\bm{\rho}^{i}$ from the last node of $\bm{\nu}^i$, that is $\bar{\bm{u}}^i = (\bar{w}^i, \bar{p}^i, \bar{b}^i)$.
Note that even though $\nu^{i}$ and $\bm{\nu}^i$ have the same set of base nodes, the ending memory values $p^i, b^i$ and $\bar{p}^i, \bar{b}^i$ are not necessarily equivalent. This is because the cycle $\bm{\mathcal{C}^{\bm{u}^i}}$ can take righmost gadget edges and reset $b$, or $\mathcal{C}^{w^i}$ can take transitions that reset $b$. Even though these cycles pass through the same set of base nodes, one is not the expansion of the other, so there is no guarantee that they will output consistent memory values. 

We therefore need to slightly adjust the memory update function in $G$, and identify the values $p^i$ and $b^i$ with the values $\bar{p}^i$ and $\bar{b}^i$ at the end of $\nu^i$. That is, we interrupt the memory update function in $G$ at the end of $\nu^i$, and ask it to treat the current memory values at $\bar{w}$ as $\bar{p}^i$ and $\bar{b}^i$. Note that this does not affect anything in $G$, since we used the memory values $p$ and $b$ only as a way to ensure consistency between a play in $G$ and its expansion to $LG'_{s'}$. In particular, this adjustment in the memory values will be invisible to $t$, since $t$ only looks at the history of a play to determine the next move.  

As $\nu^i$ is a play compliant with $t$, by Obs.~\ref{obs:unique-elongation-for-t-s}, there exists a unique $t$-compliant extension $\nu^i \cdot \sigma$ of $\nu^i$ such that $\sigma$ is compliant with both $s^i$ and $t$. We denote $\sigma$ by $\rho^{i}$. 
As explained above, we adjust the memory values at the end of $\nu^{i}$ to $\bar{p}, \bar{b}$ to match the values of the last node of $\bm{\nu}^{i}$, and take the expansion $\bm{\rho}^{i}$ of $\rho^{i}$ to $LG'_{s'}$. Note that the starting node of $\bm{\rho}^{i}$ is $\bar{\bm{u}}^i = (\bar{w}, \bar{p}, \bar{b})$. Therefore both $\bm{\rho}^{i}$ and $\bm{\nu}^{i} \cdot \bm{\rho}^{i}$ are plays in $LG'_{s'}$ (as $\bm{\nu}^i$ is a play in $LG'_{s'}$ due to~\cref{lemma:nu-i-in-gadget-game}).
Clearly, $\rho^{i}$ is a $t$-compliant play in $G$ [Showing Inv.~\ref{inv:set}-\ref{inv:set-three} is preserved].

If $\rho^{i}$ and $\bm{\rho}^{i}$ fall into~\ref{item:case1} or~\ref{item:case2} (i.e. if $\rho^{i}$ is $\forall$-unfair or $\bm{\rho}^{i}$ is eventually stabilizing), by the respective proofs we get that $\rho^{i}$ is $\exists$-winning. This implies that $\nu^i \cdot \rho^i$ is also $\exists$-winning. However, as shown in the previous paragraph, $\nu^i \cdot \rho^i$ is a $t$-compliant play in $G$, and its initial node is $v$. This shows that $t$ does not win $v$, concluding the proof of this case. Otherwise, $\rho^{i}$ and $\bm{\rho}^{i}$ fall into~\ref{item:case3} [Showing Inv.~\ref{inv:set}-\ref{inv:set-four} is preserved].


\enskip 

We showed all invariants in Inv.~\ref{inv:set} are preserved, except for~\ref{inv:set-five}, which is easily seen by how $\nu^i$ and $\bm{\nu}^i$ are constructed. 


\enskip

\noindent \textbf{\textbullet \,\, Existence of $u^i$: } The last step on confirming the well definedness of the construction is proving that the existence of a $\bm{u}^i$ is guaranteed in each iteration $i >0$. 

\begin{lemma}\label{lemma:existence-of-u-guaranteed}
  Given that the invariants given in Inv.~\ref{inv:set} holds for $\nu^{i-1}, \bm{\nu}^{i-1}, \rho^{i-1}$ and $\bm{\rho}^{i-1}$, and that $\rho^{i-1}$ and $\bm{\rho}^{i-1}$ fall into~\ref{item:case3}, there always exists a $\bm{u}^i$ in $\bm{\rho}^{i-1}$ that satisfies the properties given via~\ref{item:u-1}-\ref{item:u-3}.
\end{lemma}

\begin{linenomath}
\begin{proof}[Proof of~\cref{lemma:existence-of-u-guaranteed}] By assumption $\rho^{i-1}$ and $\bm{\rho}^{i-1}$ fall into~\cref{item:case3}, so $\bm{\rho}^{i-1}$ is alternating, and therefore, it visits the rightmost edge of a node $\bm{u} = (w, p, b)$ infinitely often. Also $\rho^{i-1}$ is $\forall$-fair, so it visits $\bar{w} = s'(\mathbf{u}^\exists_{k+1})|_V$ (a fair successor of $w$) infinitely often. Consequently, $\bm{\rho}^{i-1}$ visits an edge $(w, p', b') \to (\bar{w}, p'', b'')$, as well as $\mathbf{u}$,  infinitely often. Let $\tilde{\rho}^{i-1}$ denote the suffix of $\rho^i$ consisting of only nodes that appear infinitely often in $\rho^i$. By considering a long enough infix of $\tilde{\rho}^{i-1}$, we guarantee finding a cycle $\mathcal{C}^w$ on $w$ that is continued by $\bar{w}$ in $\rho^{i-1}$, and that visits all the nodes in $\tilde{\rho}^{i-1}$. 
We can shift $\mathcal{C}^w$ in $\rho^{i-1}$ to guarantee that its expansion $\bm{\mathcal{C}}^w$ to $LG'_{s'}$ starts with $\mathbf{u}$. As $\bm{u}$ is visited infinitely often in $\bm{\tilde{\rho}^{i-1}}$, we can further extend the expansion $\bm{\mathcal{C}}^w$ enough to guarantee it also ends with $\mathbf{u}$. We call this new cycle ${\mathcal{C}}^{\mathbf{u}}$. 
As $\bm{\tilde{\rho}}^{i-1}|_V = \tilde{\rho}^{i-1}$, the set of the base nodes of $\bm{\mathcal{C}}^{\mathbf{u}}$ will be all the nodes in $\tilde{\rho}^{i-1}$, same with the nodes of $\mathcal{C}^w$. Therefore, $\mathbf{u}$ satisfies the properties~\ref{item:u-1}-\ref{item:u-3}. 
\end{proof}
\end{linenomath}

Having proven the invariants and~\cref{lemma:existence-of-u-guaranteed}, we have shown that our construction is well-defined. Now let us prove Inv.~\ref{inv:nu-bold-nu} and~\cref{lemma:nu-bold-nu}, and proceed to the final case distinction over $\nu$ and $\bm{\nu}$ that will finalise the proof.

\enskip

\noindent \textbf{Proofs of Inv.~\ref{inv:nu-bold-nu} and~\cref{lemma:nu-bold-nu}}. 

\begin{proof}[Proof of Inv. \ref{inv:nu-bold-nu}]
Inv.~\ref{inv:nu-bold-nu} trivially holds for the initial assignments $\nu^0 = \bm{\nu}^0 = \varepsilon$ and during the construction, we showed that it is preserved.
\end{proof} 


\begin{proof}[Proof of~\cref{lemma:nu-bold-nu}]
Observe that $\bm{\nu}$ has the shape 

\begin{equation}\label{eq:bold-nu}
	\bm{\nu} = \bm{\varsigma}^1 \cdot \bm{\mathcal{C}^{\mathbf{u}^1}} \cdot \mathbf{u}^1_{k+1} \cdot (\mathbf{u}^1_{k+1})^\exists \cdot \bm{\varsigma}^2\cdot \bm{\mathcal{C}^{\mathbf{u}^2}} \cdot \mathbf{u}^2_{k+1} \cdot (\mathbf{u}^2_{k+1})^\exists \cdot \bm{\varsigma}^3 \cdot  \bm{\mathcal{C}^{\mathbf{u}^3}} \cdots
\end{equation}

and $\nu$ has the shape

\begin{equation}\label{eq:nu}
	\nu = \varsigma^1 \cdot \mathcal{C}^{w^1} \cdot \varsigma^{2} \cdot \mathcal{C}^{2} \cdot \varsigma^{3} \cdots
\end{equation}

By construction, we have that for each $i$, $\bm{\varsigma}^i|_V = \varsigma^i$, and the base nodes of $\bm{\mathcal{C}^{\mathbf{u}^i}}$ and $\mathcal{C}^{w^i}$ are the same. This directly implies~\cref{lemma:nu-bold-nu}.
\end{proof}
\enskip

\noindent \textbf{The Intuition and the Final Case Distinction}

\smallskip
Before the last part of the proof, which will go to a case distinction of $\nu$ and $\bm{\nu}$ and show that $v$ is $\exists$-winning, we would like to summarise the idea behind the construction given so far, and how it connects to the upcoming proof. 


\enskip

\noindent \textit{Idea Behind the Construction and the Proof.} Intuitively, at every iteration $i$, we construct a finite play $\bm{\nu}^i$ that starts at $\bm{\nu}^{i-1}$, continues along $\bm{\rho}^{i-1}$ until it encounters the first node ($\bm{u}^i$) whose rightmost edge is taken and satisfies the desired properties, then takes a cycle over $\bm{u}^i$ (which possibly visits the rightmost edge of $\bm{u}^i$) whose base nodes will match a cycle over $w^i$ in a $t$-compliant extension of $\nu^{i-1}$; then we continue the play from $\bm{u}^i$ with its $k+1^{\text{th}}$ successor (\ref{branch1-pp-uni}). By this last move, we refuse to take the rightmost branch of $\bm{u}^i$ (\ref{branch2-pp-uni}), except for on cycles $\bm{\mathcal{C}^{\bm{u}^i}}$. This $k+1^{\text{th}}$ successor of $\bm{u}^i$ is indeed in $LG'_{s'}$ (as explained in~\ref{item:u-2}, since $n_{\bm{u}^i} = k+1$), but not necessarily in $G'_{\bm{\rho}^i}$. This redirection allows us to accumulate the rightmost branches of universal gadgets of $\bm{u}^i$ in the cyclic parts of $\bm{\nu}$, and making sure that once the cyclic parts are removed from $\bm{\nu}$, in the remaining part of $\bm{\nu}$\textemdash still a play in $LG'_{s'}$\textemdash, the rightmost branches of a universal gadget node appears only if the node does not satisfy the properties~\ref{item:u-1}-\ref{item:u-3} as $\bm{u}^i$ is chosen as minimal with these properties. 
This will allow us to arrive at a contradiction and finish the proof.

\enskip

Before the final case distinction, we will prove one last easy lemma,~\cref{lemma:nus}, that summarises existing information, and justifies our proof strategy.

\begin{lemma}\label{lemma:nus}
 $\nu$ is a $t$-compliant play in $G$ with the initial node $v$ and $\bm{\nu}$ is a play in $LG'_{s'}$. Therefore, $\bm{\nu}$ is $\exists$-winning. 
\end{lemma}

\begin{linenomath}
\begin{proof}[Proof of~\cref{lemma:nus}]
As $\rho^0 = \rho$, which is a play in $G$ that starts from $v$, the initial node of $\nu$ is also $v$. Furthermore, as for each $i$, $\nu^i$
 is $t$-compliant (Inv.~\ref{inv:set}-\ref{inv:set-one}), and $\nu^i$ is a prefix of $\nu^{i+1}$ (Inv.~\ref{inv:set}-\ref{inv:set-five}), $\nu = \lim_{i \to \infty} \nu^i$ is also a $t$-compliant play in $G$. Similarly, as for all $i$, $\bm{\nu}^i$ is a play in $LG'_{s'}$ (Inv.~\ref{inv:set}-\ref{inv:set-two}) and $\bm{\nu}^i$ is a prefix of $\bm{\nu}^{i+1}$ (Inv.~\ref{inv:set}-\ref{inv:set-five}),  $\bm{\nu} = \lim_{i \to \infty} \bm{\nu}^i$ is also a play in $LG'_{s'}$. 
\end{proof}
\end{linenomath}

\smallskip

\cref{lemma:nus} reveals that showing $\nu$ is not $\forall$-winning is sufficient to finish the proof. This is exactly what we will do next. 

\enskip

\begin{linenomath}
Let $\widetilde{\bm{\nu}}$ be a suffix of infinitely recurring nodes in $\bm{\nu}$ (\ref{eq:bold-nu}) that starts with some $\bm{\varsigma}^i$,\footnote{Here we slightly abuse the notation. By definition, $\tilde{\bm{\nu}}$ represents \textit{the largest} tail of $\bm{\nu}$ that consists of nodes in $\Inf(\bm{\nu})$, but here we chose one that is convenient for us.} i.e., 
$$ \bm{\widetilde{\nu}} = \bm{\varsigma}^j \cdot \bm{\mathcal{C}^{\mathbf{u}^j}} \cdot \mathbf{u}^j_{k+1} \cdot (\mathbf{u}^j_{k+1})^\exists \cdot \bm{\varsigma}^{j+1}\cdot \bm{\mathcal{C}^{\mathbf{u}^{j+1}}} \cdot \mathbf{u}^{j+1}_{k+1} \cdot (\mathbf{u}^{j+1}_{k+1})^\exists \cdot \bm{\varsigma}^{j+2} \cdots$$

We remove the cycles from $\bm{\widetilde{\nu}}$ to obtain another play in $LG'_{s'}$:
$$\bm{\widetilde{\nu}}^{-} :=\bm{\varsigma}^j \cdot \mathbf{u}^j_{k+1} \cdot (\mathbf{u}^j_{k+1})^\exists \cdot \bm{\varsigma}^{j+1}\cdot \mathbf{u}^{j+1}_{k+1} \cdot (\mathbf{u}^{j+1}_{k+1})^\exists \cdot \bm{\varsigma}^{j+2} \cdots$$
\end{linenomath}

Note that $\bm{\widetilde{\nu}}^{-}$ is a play in $LG'_{s'}$, since removing cycles from a play on a graph, gives another play on the same graph.

\enskip

\begin{linenomath}
Similarly, let $\widetilde{\nu}$ be the tail of $\nu$(\ref{eq:nu}) starting from $\varsigma^j$, that is,
$$\widetilde{\nu} = \varsigma^j \cdot \mathcal{C}^{w^j} \cdot \varsigma^{j+1} \cdot \mathcal{C}^{w^j} \cdot \varsigma^{j+2} \cdots$$
\end{linenomath}

\enskip

From Inv.~\ref{inv:nu-bold-nu} and~\cref{lemma:nu-bold-nu}, we know that $\widetilde{\nu}$ also contains only infinitely recurring nodes of $\nu$.

\enskip

We will finish the proof by a case distinction on the behaviour of $\widetilde{\nu}$, $\widetilde{\bm{\nu}}$ and $\widetilde{\bm{\nu}}^{-}$.
We consider whether the rightmost edges of universal gadgets (i.e., those of priority $2k+2$) appear in $\widetilde{\bm{\nu}}^{-}$ and the $\forall$-fairness of $\widetilde{\nu}$.
\begin{enumerate}[label=\textbf{D}\arabic*., leftmargin=*, itemsep=0.5em]
    \item The rightmost edges of universal gadgets do not exist in $\widetilde{\bm{\nu}}^-$, \label{item:D1}
    \item The rightmost edge of a universal gadget exists in $\widetilde{\bm{\nu}}^-$ and  $\widetilde{\nu}$ is $\forall$-unfair (i.e. $\nu$ is $\forall$-unfair)  \label{item:D2}
    \item The rightmost edge of a universal gadget exists in $\widetilde{\bm{\nu}}^-$ and $\widetilde{\nu}$ is $\forall$-fair (i.e. $\nu$ is $\forall$-fair) \label{item:D3}
\end{enumerate}

We show that all 3 cases yield contradictions.

\enskip

\noindent \ref{item:D1} Recall that $\widetilde{\bm{\nu}}^-$ is a play in $LG^{-}_{s'}$. By Item~\ref{item:u-2}, we have that each $\mathbf{u}^i$ in $\widetilde{\bm{\nu}}^-$ has $n_{\mathbf{u}^i} = k+1$. This introduces a contradiction through~\cref{lemma:Gminus-cycles}.

\enskip

\noindent \ref{item:D2} Recall that as a result of~\cref{lemma:nu-regularity-for-forall-player}, we picked $s'$ so that $n_{\bm{u}}$ can be equal to $k+1$ for a universal node $\bm{u} = (w, p, b) \in \Win_\exists(G')$ only if $b = \exists$. Let $\sigma$ be an $s$ play in $G$, then by the definition of expansion, if $\bm{\sigma}$ passes through the rightmost edge of a universal gadget rooted at $\bm{u} = (w, p, b)$, then $b = \exists$. Since the rightmost edges of universal nodes set the $b$ value to $\forall$, the existence of such edges in $\widetilde{\bm{\nu}}^{-}$ implies that $\widetilde{\bm{\nu}}^{-}$, and therefore $\widetilde{\bm{\nu}}$, is eventually alternating. Therefore the maximum $\beta$ priority of the base nodes of $\widetilde{\bm{\nu}}$ is even (by~\cref{lemma:alternating-play-priority}). However, by~\cref{lemma:nu-bold-nu} we have that the base nodes of $\widetilde{\bm{\nu}}$ are exactly the nodes of $\widetilde{\nu}$. Since $\nu$ is $\forall$-unfair and $\forall$-winning, the maximum $\beta$ priority seen in $\widetilde{\nu}$ is odd. This yields a contradiction. 

\enskip

\noindent \ref{item:D3} As $\widetilde{\nu}$ is $\forall$-fair, for every universal gadget node $\mathbf{u} = (w, p, b)$ whose rightmost branch is taken in $\widetilde{\bm{\nu}}^{-}$ (as before, we know $n_{\bm{u}} = k+1$ and $b = \exists$)  
 the edge $w \to \bar{w}$ exists in $\widetilde{\nu}$
 where $\bar{w} = s'(\mathbf{u}^\exists_{k+1})|_V$\textemdash a fair successor of $w$.

On the other hand, as in every iteration in the construction of $\bm{\nu}$ we chose $\bm{u}^i$ to be the minimal, the rightmost edge of a node $\bm{u}$ is taken in $\bm{\nu}^-$ only if the properties in Items~\ref{item:u-1}-\ref{item:u-3} could not be satisfied for $\bm{u}$. That is,
there exists no $t$-compliant extension $\kappa$ of $\nu^{i-1} \cdot \varsigma^i$ that has a prefix $\mathcal{C}^{w^i}$ that is a cycle over $w^i$ and is continued by $\overline{w}^i$; and a cycle $\bm{\mathcal{C}}^{\mathbf{u}}$ in $LG'_{s'}$ over $\mathbf{u}$ that visits the same base nodes as $\mathcal{C}^{w^i}$. 

To yield a contradiction, we will construct such $\kappa$, $\mathcal{C}^{w}$ and $\bm{\mathcal{C}}^{\mathbf{u}}$ for $\mathbf{u}$. 

\enskip

As $\mathbf{u} = (w,p,b)$ is in $\widetilde{\bm{\nu}}^{-}$, it is in $\bm{\varsigma}^{j'}$ for some $j' > j$.

\smallskip

Then, $\bm{\nu}^{j'-1} = \bm{\varsigma}^1 \cdot \bm{\mathcal{C}}^{\mathbf{u}^1} \cdots \bm{\varsigma}^{j'-1} \cdot \bm{\mathcal{C}}^{\mathbf{u}^{j'-1}} \cdot \mathbf{u}^{j'-1}_{k+1} \cdot (\mathbf{u}^{j'-1}_{k+1})^{\exists} \cdot \bar{\mathbf{u}}^{j'-1}$. 

\enskip 

Similarly, $\nu^{j'-1} = \varsigma^1 \cdot \mathcal{C}^{w^1} \cdots \varsigma^{j'-1} \cdot \mathcal{C}^{w^{j'-1}} \cdot \bar{w}^{j'-1}$.

\enskip 

Let $\bm{\varsigma}^{j;}_{\leftarrow\mathbf{u}}$ be the prefix of $\bm{\varsigma}^{j'}$ until $\mathbf{u}$, and let $\bm{\varsigma}^{j'}_{\rightarrow\mathbf{u}}$ be its suffix starting with $\mathbf{u}$. That is, $\bm{\varsigma}^{j'} = \bm{\varsigma}^{j'}_{\leftarrow\mathbf{u}} \cdot \bm{\varsigma}^{j'}_{\rightarrow\mathbf{u}}$.

\smallskip

As for all $i$, $\bm{\varsigma}^i$ is an expansion of $\varsigma^i$; 
$\nu^{j'-1} \cdot \varsigma^{j'}_{\leftarrow\mathbf{u}}$ is a $t$-compliant finite play in $G$, where $\varsigma^{j'}_{\leftarrow \mathbf{u}}$ is the projection of $ \bm{\varsigma}^{j'}_{\leftarrow \mathbf{u}}$ to $G$. 

\smallskip

Then, the tail $\kappa = \varsigma^{j'}_{\rightarrow\mathbf{u}} \cdot \mathcal{C}^{w^{j'}} \cdot \varsigma^{j'+1} \cdot  \mathcal{C}^{w^{j'+1}}  \cdots $ of $\nu$ is an extension of $\nu^{j'-1} \cdot \varsigma^{j'}_{\leftarrow\mathbf{u}}$, as required.\footnote{Here, $\varsigma^{j'}_{\mathbf{u}}$ serves as $\varsigma^i$ where $\mathbf{u} = \mathbf{u}^i$. } $\kappa$, being a tail of $\widetilde{\nu}$, takes the edge $w \to \bar{w}$ infinitely often. Consequently, we can find a prefix $\mathcal{C}^{w}$ of $\kappa$ that is a cycle over $w$, and $\kappa = \mathcal{C}^{w} \cdot \bar{w} \cdots$. By allowing the cycle to be longer, we can ensure that $\mathcal{C}^{w}$ visits all nodes of $\widetilde{\nu}$. 
Further, as $\widetilde{\bm{\nu}}$ visits $\mathbf{u}$ infinitely often, we can pick a long enough cycle $\bm{\mathcal{C}}^{\mathbf{u}}$ that contains all nodes in $\widetilde{\bm{\nu}}$, and therefore all the base nodes in $\widetilde{\nu}$. This makes sure that base nodes of $\bm{\mathcal{C}}^{\mathbf{u}}$ and the nodes of $\mathcal{C}^{w}$ agree.

With this, the rightmost edge of $\mathbf{u}$ cannot be taken in $\widetilde{\bm{\nu}}^{-}$ in the first place. This yields the desired final contradiction, and ends the proof. 
\qed

%% file: fixpointcharacterization.tex
\section{Fixpoint Characterization of Winning Regions}~\label{section:fixpoints}


In this section, we characterize the winning regions in 
fair games with parity conditions by means of fixpoint expressions.
Thereby we provide an alternative, symbolic route to solve such games, rather than by reducing
to parity games. We start by briefly recalling
details on Boolean fixpoint expressions.

\myparagraph{Fixpoint expressions and fixpoint games.}\label{definition:fpgames}
Let $U$ be a finite set, let $o$ be a natural number and let
$f:\mathcal{P}(U)^o\to \mathcal{P}(U)$ be
a monotone function, that is, assume that whenever we have
sets $X_j,Y_j\subseteq U$ such that
$X_j\subseteq Y_j$ for all $1\leq j\leq o$, then $f(X_1,\ldots,X_o)\subseteq
f(Y_1,\ldots,Y_o)$.
Then the data $f$ and $o$ induce the \emph{fixpoint expression}
\begin{equation}
e=\eta_o X_o.\,\eta_{o-1} X_{o-1}.\,\ldots.\nu X_2.\,\mu X_1.\,f(X_1,\ldots,X_o)
\end{equation}
where $\eta_i=\nu$ if $i$ is even and $\eta_i=\mu$ if $i$ is odd.
We define the semantics of fixpoint expressions using standard parity games.
Given a fixpoint expression $e$, the associated \emph{fixpoint game}
$G_{e}=(W_\exists,W_\forall,E, \text{Parity}(\kappa))$ for the priority function $\kappa:W_\exists\cup W_\forall\to\mathbb[o]$ is the following
parity game.
We put $W_\exists=U\times\{1,\ldots ,o\}$, $W_\forall=\mathcal{P}(U)^o$.
Moves and priorities are defined by
\begin{align*}
E(v,i)&=\{\overline{Z}\in W_\forall\mid v\in f(\overline{Z})\} & \kappa(v,i)&=i\\
E(\overline{Z})&=\{(v,i)\mid v\in Z_i\} &
\kappa(\overline{Z})&=0
\end{align*}
for $(v,i)\in W_\exists$ and $\overline{Z}=(Z_1,\ldots, Z_o)\in W_\forall$.
Then we say that $v\in U$ is \emph{contained} in $e$ (denoted $v\in e$) if and only
if $\exists$-player wins the node $(v,1)$ in $G_{e}$.

\begin{remark}
    The above game semantics for fixpoint expressions has been shown 
    to be equivalent to the more traditional Knaster-Tarski semantics~\cite{BaldanEA19};
    the cited work takes place in a more general setting and therefore uses slightly more verbose 
     parity games.
\end{remark}

\subsection{Winning Regions in Standard Parity Games}\label{definition:parityfixpoint}

\leavevmode\par\medskip

We start by recalling the well-known fixpoint characterization of 
winning regions in parity games. So let $G=(A,\text{Parity}(\lambda))$ be a parity game where $A = (V_\exists,V_\forall,E)$ is a game arena, $V = V_\exists \dotcup V_\forall$ and  $\lambda: V \to[2k]$ a priority function. 
To be able to write fixpoint expressions over such games
we define monotone operators on subsets of $V$ by putting
\begin{align*}
	\Diamond X &= \{v\in V\mid E(v)\cap X\neq\emptyset\} &
	\Box X &= \{v\in V\mid E(v)\subseteq X\} \\
\end{align*}
for $X\subseteq V$ and also put 
\begin{align*}
	\Cpre(X)=(V_\exists\cap\Diamond X)\cup(V_\forall\cap\Box X).
\end{align*}

Thus $\Diamond X$ and $\Box X$
are the sets of nodes that have an outgoing edge leading to $X$,
and for which all outgoing edges lead to $X$, respectively. 
and $\Cpre(X)$ is the set of nodes
from which $\exists$-player can force the game to reach a node
from $X$ in one step.
Also, we define $C_i=\{v\in V\mid \lambda(v)=i\}$ for $1\leq i\leq 2k$,
that is, $C_i$ is the set of game nodes with priority $i$.

	Using this notation, we define a function $\mathsf{parity}:\mathcal{P}(V)^{2k}\to\mathcal{P}(V)$ by putting
\begin{align*}
	\mathsf{parity}(X_1,\ldots, X_{2k}):= (C_1 \cap \Cpre(X_1)) \cup \ldots \cup  (C_{2k} \cap \Cpre(X_{2k}))
\end{align*}
for $(X_1,\ldots ,X_{2k})\subseteq \mathcal{P}(V)^{2k}$.
This function is monotone and it is well-known (see e.g~\cite{Walukiewicz02}) 
that the fixpoint expression 
\begin{align*}
\nu_{2k} X_{2k}.\,\mu_{2k-1} X_{2k-1}.\,\ldots.\nu X_2.\,\mu X_1.\,\mathsf{parity}(X_1,\ldots, X_{2k})
\end{align*}
induced by the data $\mathsf{parity}$ and $2k$ characterizes the winning region
in parity games with priorities $1$ through $2k$.
Here, the function $\mathsf{parity}$ encodes the evaluation of one step in the parity game and the
alternating arrangement of least and greatest fixpoints encodes the objective that the maximal priority that occurs infinitely often is even.

\subsection{Winning Regions in Fair Parity/$\bot$ Games}\label{definition:paritybotfixpoint}

\leavevmode\par\medskip

Next we present a fixpoint characterization of the winning
regions in fair games of the form $G=(A,\text{Parity}(\lambda),\bot)$ where $A = (V_\exists,V_\forall,E,E_f)$ is a fair game arena, $V = V_\exists \dotcup V_\forall$ and  $\lambda: V \to[2k]$ a priority function. 

In the upcoming fixpoint characterization,
the function $\mathsf{parity}$ will still apply to `normal' nodes $\Vn$,
but we also introduce additional functions
\begin{align*}
		\Diamond_f X &= \{v\in V\mid E_f(v)\cap X\neq\emptyset\} &
\Box_f X &= \{v\in V\mid E_f(v)\subseteq X\}
\end{align*}
that specialize $\Diamond$ and $\Box$ to fair edges.

To obtain a fixpoint characterization, we directly follow the minimal
gadget constructions for fair parity/$\bot$ games as shown in~\cref{fig:optimalexistentialparitybotgadgets}.
To this end, we define additional functions that encode
the inner gadget nodes.
For $1\leq i < k$, put
\begin{align*}
\Apre_{\exists}(X_i,X_{i+1})&=\Diamond X_i\cap \Box_f X_{i+1} &
\Apre_{\forall}(X_i,X_{i+1})&=\Diamond_f X_i\cap \Box X_{i+1},
\end{align*}
encoding nodes $(v^\forall,2i)$ for
$v\in \Vfair_\exists$ and $v\in \Vfair_\forall$, respectively
(here, $\mathsf{Apre}$ stands for \emph{alternative} predecessor function, as it encodes the additional $\forall$-choice of whether a fair edge is to be taken).
Then, given an even number $p$, we let $I_p=\{ i\mid i \text{ odd}, p \leq i< 2k\}$
denote the set of odd priorities that lie strictly between $p$ and $2k$,
and put 
\begin{align*}
\varphi^\fair_{\exists, p}&=
 \begin{cases} \bigcup_{i\in I_p}\Apre_{\exists}(X_{i},X_{i+1})\,\cup\, \Diamond\, X_{2k+1} & \text{$p$ is odd} \\
\bigcup_{i\in I_p}\Apre_{\exists}(X_{i},X_{i+1})\, \cup \,\Diamond\, X_{2k+1} \cup\Box_f \,X_p & \text{$p$ is even},
\end{cases}\\
\varphi^\fair_{\forall, p}&=
 \begin{cases} \bigcup_{i\in I_p}\Apre_{\forall}(X_{i},X_{i+1}) & \text{$p$ is odd} \\
\bigcup_{i\in I_p}\Apre_{\forall}(X_{i},X_{i+1})\, \cup \Box \,X_p & \text{$p$ is even}
\end{cases}
\end{align*}

 Using this notation, the winning region for the existential player in fair parity/$\bot$ games with priorities $1$ through $2k$ can be characterized by the fixpoint expression induced by $2k+1$ and the function
$\chi$ that is defined to
 map $(X_1,\ldots, X_{2k+1})\in\mathcal{P}(V)^{2k+1}\to\mathcal{P}(V)$
 to the set
 \begin{align*}
 \chi(X_1,\ldots, X_{2k+1})= &(\Vn \,\cap\, \mathsf{parity})\,\cup\\
& (\Vfair_\exists\, \cap \, \bigcup_{i \in [2k]}\, C_i \,\cap \,\varphi^\fair_{\exists, i})\, \cup \\ &(\Vfair_\forall\, \cap \, \bigcup_{i \in [2k]}\, C_i \,\cap \,\varphi^\fair_{\forall, i})
 \end{align*}
The function $\chi$ therefore treats normal nodes from
$V^\mathsf{n}$ in the same way as nodes in standard parity games
are treated, but for fair nodes with priority $i$, the 
functions $\varphi^\mathsf{fair}_{\exists,i}$ and
$\varphi^\mathsf{fair}_{\forall,i}$ are used to encode
the respective gadget construction.
The full fixpoint expression then is
\begin{align}\label{eq:mutuallyfairparityfp}
 e=\mu X_{2k+1}.\, \nu X_{2k}.\, \mu X_{2k-1}\,\ldots \,\nu X_2.\,\mu X_1. \,
 \chi(X_1,\ldots, X_{2k+1})
\end{align}

Next we show that this fixpoint expression
characterizes the winning region of $\exists$-player in fair parity/$\bot$ games.

\begin{restatable}{theorem}{restatablefairparityfpcorrectness}\label{theorem:parityfpcorrectness} Let $G=(A,\text{Parity}(\lambda),\bot)$ where $A = (V_\exists,V_\forall,E,E_f)$ is a fair game arena, $V = V_\exists \dotcup V_\forall$ and  $\lambda: V \to [2k]$ is a priority function. 
	Then the fixpoint expression given in~\eqref{eq:mutuallyfairparityfp} characterizes $\Win_\exists(G)$.
\end{restatable}

\begin{proof}[Proof of~\cref{theorem:parityfpcorrectness}]
	The proof is by mutual transformation of winning strategies in the game 
	$G'$ from the proof~\cref{thm:fairParityToParity} on one hand, and in the semantic game $G_e$ for~\eqref{eq:mutuallyfairparityfp} on the other hand.
	
	Let $G' = (V'_\exists, V'_\forall, E', \Omega:V' \to [2k+1])$ be the parity game equivalent to $G$, given in the proof of~\cref{thm:fairParityToParity} that is depicted by the existential gadgets in~\cref{fig:existentialgadgetsparitybot}.
	Recall that the nodes in $V'$ are in the shape of $v, v_i, v_i^\forall$ or $v_i^\exists$ for $v \in V$; and that nodes of shape $v \in V$ are won by a player in $G'$ if and only if they are won by the same player in the fair parity/$\bot$ game $G$.  Recall that we denote nodes $v_i$
	also by $(v,i)$.
	By the semantics of fixpoint expressions given in~\cref{definition:fpgames},
	it suffices to show that
	$\exists$-player wins $v \in V$ in $G'$ if and only if $\exists$-player wins $v$ in $G_e$, i.e. $v\in e$.

	In $G_e$, $\exists$-player can move from
	nodes $v\in V$ to every tuple $(Z_1, \ldots, Z_{2k+1})$ 
	such that $v\in\chi(Z_1,\ldots, Z_{2k+1})$, that is,
	to every tuple $(Z_1, \ldots, Z_{2k+1})$ that satisfies the following conditions.
	\begin{align}
		\intertext{If $v \in \Vn$ with priority $i \in [2k]$ then}
		& v \in \Cpre(Z_i) \text{ from }\mathsf{parity}\label{line:Vn}
		\intertext{
			If $v \in \Vfair_\exists$ with \emph{odd} priority $i \in [2k]$, then we have } 
		& v \in \varphi^{\fair}_{\exists, i} \,=\, \Apre_\exists(Z_i, Z_{i+1}) \,\cup \,\ldots \,\cup\, \Apre_\exists(Z_{2k-1}, Z_{2k}) \,\cup\, \Diamond\, Z_{2k+1} \label{line:oddphiexists}
		\intertext{If $v \in \Vfair_\exists$ with \emph{even} priority $i \in [2k]$, then we have}  
		& v \in \varphi^{\fair}_{\exists, i} \,= \,\Box_f \,Z_i\, \cup \,\Apre_\exists(Z_{i+1}, Z_{i+2})\, \cup \,\ldots\, \cup \,\Apre_\exists(Z_{2k-1}, Z_{2k}) \,\cup\, \Diamond\, Z_{2k+1} \label{line:evenphiexists}
		\intertext{If $v \in \Vfair_\forall$ with \emph{odd} priority $i \in [2k]$, then we have} 
		& v \in \phi^{\fair}_{\forall, i}\, =\, \Apre_\forall(Z_i, Z_{i+1}) \,\cup\, \ldots \,\cup \,\Apre_\forall(Z_{2k-1}, Z_{2k}) \label{line:oddphiforall}
		\intertext{If $v \in \Vfair_\forall$ with \emph{even} priority $i \in [2k]$, then we have}  
		& v \in \phi^{\fair}_{\forall, i} \,= \,\Box \,Z_i \,\cup \,\Apre_\forall(Z_{i+1}, Z_{i+2})\, \cup\, \ldots\, \cup\, \Apre_\forall(Z_{2k-1}, Z_{2k}) \label{line:evenphiforall}
	\end{align}
	Note that for a $v$ with priority $i$, it is sufficient to pick $Z_j=\emptyset$
	for all for $j < i$, and in particular $Z_j = \emptyset$ for $j \neq i$ when $v \in \Vn$. 
	Since existential player owns nodes $(v,i)$ and moves to a universal player node $\bar{Z} = (Z_1, \ldots, Z_{2k+1})$ where $ v \in \chi(\bar{Z})$,
	without loss of generality, we can assume $\exists$-player always makes these more economical moves, restricting the options of the universal player.  
	As a reaction to $\exists$-player moving from
	$v$ to $\bar{Z}$, $\forall$-player can challenge every node
	$w$ that is contained in some set $Z_j$;
	since $\exists$-player moves economical, the move from
	$\bar{Z}$ to $w$ signals a priority $j \geq i$.
	
	So let $s'$ be a positional $\exists$-strategy in $G'$ that wins the game from every node in the winning region of $\exists$-player in $G'$. 
	We construct a positional $\exists$-strategy $s$ in $G_{e}$ as follows. We put $s(v)=\bar{Z}$,
	where the tuple $\bar{Z}= (Z_1, \ldots, Z_{2k+1})$ is constructed 
	by case distinction.
	
	\begin{itemize}
		\item For $v \in \Vn$, put $Z_{\lambda(v)} = \begin{cases} E(v) \quad  &\text{if} \quad v \in V_\forall \\ \{s'(v)\} \quad &\text{if} \quad v\in V_\exists, \end{cases}$
		and $Z_j = \emptyset$ for all other $j$.
		\item For $v \in \Vfair_\exists$, we have $s'(v) = v_i$ for some $i \in [\mathbf{c}, k+1]$ where $\mathbf{c} = \lceil \lambda(v) \backslash 2 \rceil$.
		If $i \in [\mathbf{c}+1, k]$ or $\lambda(v)$ is odd, then $v_i$ has two successors $v_i^\exists \in V'_\exists$ and $v_i^\forall \in V'_\forall$. If $i=k+1$ then it has one successor $v_{k+1}^\exists \in V'_\exists$; and
		if $i = \mathbf{c}$ and $\lambda(v)$ is even then $v_i$ has 
		one successor $v_\mathbf{c}^\forall \in V'_\forall$. 
		Note that $\exists$-player wins all the successors of $v_i$ using $s'$.
		We put 
		\begin{align*}
			&Z_{2i-1} = \{s'(v_i^\exists)\} &&\text{ unless } i = \mathbf{c} \text{ and }\lambda(v) \text{ is even, in which case }Z_{2i-1} = \emptyset  \\
			&Z_{2i} = E_f(v)  &&\text{ unless }i=k+1, \text{ in which case }Z_{2i} \text{ does not exist}
		\end{align*}
		and $Z_j = \emptyset$ for all other $j$.
		\item For $v \in \Vfair_\forall$, we have $s'(v) = v_i$ for $i \in [\mathbf{c}, k]$ for which the same constaints apply. We put
		\begin{align*}
			&Z_{2i-1} = \{s'(v_i^\exists)\} &&\text{ unless } i = \mathbf{c} \text{ and } \lambda(v) \text{ is even, in which case } Z_{2i-1}=\emptyset  \\
			&Z_{2i} = E(v)  &&
		\end{align*}
		and $Z_j = \emptyset$ for all other $j$.
	\end{itemize}
	We note that $s$ is indeed a valid strategy for $G_e$ 
	since we have $v\in\chi(s(v))$ in all cases.
	
	It remains to show that $s$ wins every node $v$ in $G_e$
	that is won by $s'$ in $G'$. First we make a small observation about the plays compliant with $s$ in $G_e$.
	Whenever $\exists$-player moves to a node $\bar{Z}$, there is either one non-empty set $Z_j$ for $j \in [2k+1]$; or there are two non-empty sets $Z_j$ and $Z_{j+1}$ where in the latter case $j$ is odd.
	When $\forall$-player picks the next node from such a $\bar{Z}$, in the first case it moves to $(w', j)$ for some $w' \in Z_j$ signalling priority $j$; whereas in the second case it can signal either $j$ or $j+1$ depending on the successor $(w', \cdot)$ and which set the $w'$ is inside of. 
	We assume without loss of generality that, whenever $\forall$-player moves to a node $(w', \cdot)$ with $w' \in Z_{j} \cap Z_{j+1}$, it moves to $(w', j)$, since the universal player prefers seeing odd priorities.
	
	Let $\pi$ be a play in $G_{e}$ that adheres to $s$ and starts at a node $v$ such that $(v, 1, \exists)$ is in the winning region of $\exists$ in $G'$.
	
	Then $\pi$ induces a play $\rho$ of $G'$ that adheres to $s'$ as follows:
	whenever $\exists$-player moves from $w$ to $s(w)=\bar{Z}$, 
	$\forall$-player picks a $j$ such that $Z_j \neq \emptyset$ and some $w'\in Z_j$ signalling
	priority $j$. 

	We construct $\rho$ as follows:
	\begin{itemize}
		\item if $w\in \Vn$, then $\rho$ proceeds from $w$ to $w'$;
		\item if $w\in \Vfair$, then $\rho$ proceeds from
		$w$ to $s'(w)=w_i$. If $w'=s'(w_i^\exists)$ and if $w_i$ has a successor $w_i^\exists$ (i.e. $i\neq \mathbf{c}$ with $\lambda(v)$ even),
		then $\rho$ proceeds from $w_i$ to $w_i^\exists$
		and from there to $w'$.
		If $w'\neq s'(w_i^\exists)$ or $w_i^\exists$ does not exist,
		then $\rho$ proceeds from $w_i$ to $w_i^\forall$
		and from there to $w'$. In the first case, $\rho$ visits priority $2i-1$, whereas in the second case it visits $2i$.
	\end{itemize}
	It is not hard to see that $\rho$ visits a node with priority $j$ whenever $\pi$ visits a node $(w', j)$, signalling $j$.
	For $v \in \Vn$, this is trivial since only $Z_{\lambda(v)}$ is non-empty. For $v \in \Vfair$, this is due to the assumption we have on $s'$ and on the construction of $s$:
	Both $s'$ and $s$ prefer moves that signal the odd priority $j$ whenever $w' = succ(w)$. Note that the priority $\lambda(v)$ of a node $v \in \Vfair$ in $\pi$ does not 
	affect the winning condition since in the minimal gadgets that we are mimicing, each visit to a fair node $v$ is followed by a visit to a gadget node with priority $\geq \lambda(v)$. 
	As $\rho$ adheres to $s'$ and $s'$ is a winning strategy,  we conclude that $\exists$-player wins $\pi$.
	
	\bigskip
	For the converse direction, let $s: W_\exists \to W_\forall$ be a positional winning $\exists$-strategy in $G_{e}$. We define a positional
	strategy $s'$ for $\exists$-player in $G'$ as follows.
	For each $v$ in $G_e$ such that $\exists$-player wins from $v$
	with strategy $s$,
	we have $s(v) = \bar{Z}$ for some $\bar{Z} = (Z_1, \ldots, Z_{2k+1})$ such that
	\begin{itemize}
		\item If $v \in \Vn\cap V_\exists$, then we have $v\in \Cpre(Z_{\lambda(v)})$ by~\cref{line:Vn}. In this case pick some
		$w \in Z_{\lambda(v)}$ and put $s'(v)=w$.\label{item:Vn}
		
		\item If $v \in \Vfair_\exists$, then we have by~\cref{line:oddphiexists} and~\cref{line:evenphiexists} that $v \in \Apre_\exists(Z_{2i-1}, Z_{2i})$ for some $2i-1 \in [\lambda(v), 2k]$, or $\lambda(v)$ is even and $v \in \Box_f \,Z_{\lambda(v)}$, or $v \in \Diamond \, Z_{2k+1}$.
		\\ In this case we put
		\begin{equation*}
			s'(v) = \begin{cases} v_i &\text{ if } v \in \Apre_\exists(Z_{2i-1}, Z_{2i})\\
				v_{\mathbf{c}} &\text{ if } \lambda(v) \text{ is even and } v \in \Box_f\, Z_{\lambda(v)} \\
				v_{k+1} &\text{ if } v \in \Diamond\, Z_{2k+1} \end{cases}
		\end{equation*}
		In the first case, $v_i$ has a successor $v_i^\exists \in V_\exists$. We put $s(v_i^\exists)=w$ for some $w\in Z_{2-1}$. Similarly in the last case, we put $s(v_{k+1})=w$ for some $w \in Z_{2k+1}$.
		
		\item If $v \in \Vfair_\forall$, then we have by~\cref{line:oddphiforall} and~\cref{line:evenphiforall} that $v \in \Apre_\forall(Z_{2i-1}, Z_{2i})$ for some odd $2i-1 \in [\lambda(v), 2k]$, or $\lambda(v)$ is even and $v \in \Box \,Z_{\lambda(v)} $.
		\\ In this case we put
		\begin{equation*}
			s(v) = \begin{cases} v_i & \text{ if } v \in \Apre_\forall(Z_{2i-1}, Z_{2i})\\
				v_\mathbf{c} & \text{ if } \lambda(v) \text{ is even and } v \in \Box\, Z_{\lambda(v)} \end{cases}
		\end{equation*}
		In the first case $v_i$ has a successor $v_i^\exists \in V_\exists$ and we put $s(v_i^\exists)=w$ for some $w\in Z_{2i-1}$.
	\end{itemize}
	\begin{linenomath}
		To see that $s'$ is a winning strategy, let $\rho$ be a play that adheres to $s'$ and starts at some node $v$ such that $\exists$ wins $v$ in $G_{e}$ by following strategy $s$. Then by construction, $\rho$ induces a play $\pi$ in $G_{e}$ that adheres to $s$. As $s$ is a winning strategy, the maximum priority $m$ that is signalled infinitely often in $\pi$ is even. By definition of $G_{e}$,
		$\pi$ infinitely often takes a transition of the shape
		$\bar{Z} \to (w, m)$ where by our assumption, $w \in Z_{m} \setminus Z_{m-1}$. Furthermore, transitions that signal a priority greater than
		$m$ occur only finitely often in $\pi$.
		Observe that whenever $\pi$ takes such a transition, it translates to $\rho$ either as a node $v\in \Vn$ with $\lambda(v) = m$ being visited, or as taking a path $v \to v_{m\backslash 2} \to v^\forall_{m\backslash 2} \to w \in Z_m \setminus Z_{m-1}$ from some node $v$ with $\lambda(v)\leq m$. This path sees $m$ as the maximum priority due to $v^\forall_{m \backslash 2}$. Similarly, a larger priority $m'>m$ is not seen in $\rho$ infinitely often either, since this would reflect to $\pi$ as some node $(w', m')$ being visited infinitely often by $\pi$. 
		Thus, the maximum priority seen infinitely often in $\rho$ is also $m$, making it a winning play for $\exists$-player.
	\end{linenomath}
\end{proof}

We point out that for $\forall$-fair parity games (with $\Vfair_\exists=\emptyset$),~\cref{eq:mutuallyfairparityfp} instantiates
to the fixpoint characterization given in~\cite{banerjee2022fast}; it follows that
the parity game reductions from~\cref{section:reductionparity} apply to the one-sided fair parity games considered in~\cite{banerjee2022fast} as well.

\subsection{Winning Regions in Fair Parity/Parity Games}\label{definition:parityparityfixpoint}

\leavevmode\par\medskip

For fair parity/parity games, we obtain a similar fixpoint characterization,
encoding the reduction to parity games presented in
\cref{subsection:reductionparityparity} along the lines of Figures~\ref{fig:existentialgadgetparityparity} and~\ref{fig:universalgadgetparityparity}.
So let $V$, $E$ and $E_f$ be the set of nodes, edges and fair edges, respectively, in a fair parity/parity game $G$ with priority
functions $\lambda:V\to[2k]$ and $\Gamma:V\to[d]$.
Recall that
nodes in the game $G'$ from~\cref{thm:fairParityPlusParityToParity}
are annotated with memory values $p\in[d]$ and $b\in[2]$
in which the game $G'$
records the maximal $\Gamma$-priority that has recently been seen and the
player that has last insisted on keeping control in one of its fair nodes.
Hence all involved functions in the fixpoint expression work over (subsets of) the set $V\times [d]\times [2]$ of base nodes,
consisting of game nodes that are annotated with memory values.
If no confusion arises, we denote the elements of $[2]$ by $\exists$ and $\forall$.

In more detail, we overload the operators $\Diamond$, $\Box$, $\Diamond_f$ and
$\Box_f$ to work over $V\times [d]\times[2]$ and update the memory values accordingly,
using the adapted successor functions $E':\mathcal{P}(V\times [d]\times[2])\to \mathcal{P}(V\times [d]\times[2])$ and $E'_f:\mathcal{P}(V\times [d]\times[2])\to \mathcal{P}(V\times [d]\times[2])$
from~\cref{subsection:reductionparityparity} that update memory values. 
We also adapt the sets $C_i$ for $1\leq i\leq 2k$,
$\Vn$, $\Vfair_\exists$ and $\Vfair_\forall$,
and also the functions $\Cpre$ and $\mathsf{parity}(X_1,\ldots, X_{2k})$
accordingly.
Then the adapted version of $\Cpre$ computes the sets $\Cpre(X)$ of nodes
$(v,p,b)$ from which $\exists$-player can force the game to reach a tuple $(w,p',b)$ from $X$ in one step, where $p'=\max(p,\Gamma(v))$ and
again, the fixpoint induced by $2k$ and $\mathsf{parity}$ characterizes the winning region in parity games
with priorities $1$ through $2k$, in which nodes are annotated with memory
values $p$ and $b$, but without ever actually using the memory values.
This formula will still apply to `normal' nodes $\Vn$ in the fixpoint characterization of fair parity/parity games.

Also, we put
\begin{equation}
	\odot (X_{2k+1},\ldots X_{2k+2+d}) = \{(v,p,b)\in V\times [d]\times[2]\mid (E(v)\times\{p'\}\times\{\exists\}) \cap X_{i}\neq\emptyset\}
\end{equation}
for $X\subseteq V\times[d]\times[2]$, where $p'=p$ and $i=2k+1$ if $b=\exists$ 
and $p'=1$ and $i=2k+2+p$ if $b=\forall$. 
Thus
$\odot (X_{2k+1},\ldots X_{2k+2+d})$ consists of tuples $(v,p,b)$ for which
there is a move $(v,w)$ in $G$ such that $(w,p,\exists)\in X_{2k+1}$
if $b=\exists$, and $(w,1,\exists)\in X_{2k+2+p}$
if $b=\forall$.
Intuitively, this encodes the  branch in gadgets for
fair $\exists$ nodes; the constraints correspond to signalling priority 
$2k+2+p$ resetting the memory component from $p$ to $1$ if the universal player last insisted on keeping control
at one of its fair nodes, and to signalling $2k+1$ and leaving
the memory component $p$ unchanged if the universal player last insisted on keeping control at one of its fair nodes.

Then we redefine
\begin{align*}
	\varphi^\fair_{\exists, j}&=
	\bigcup_{i\in [k]}(\Diamond X_{2i-1}\cap \Box_f X_{2i})\,\cup\, \odot(X_{2k+1},\ldots, X_{2k+2+d})\\ 
	\varphi^\fair_{\forall, j}&=
	\bigcap_{i\in [k]} (\Box X_{2i}\cup\Diamond_f X_{2i+1})\cap  \Diamond_f{X_1}\cap\Box' X_{2k+2} 
\end{align*}
where $\Box'X_{2k+2}=\{(v,p,b)\mid \forall w\in E(v).\,(w,\max(p,\Gamma(v)),\forall)\in X_{2k+2}\}$, that is, $\Box'X_{2k+2}$ contains those tuples $(v,p,b)$
such that for all moves from $v$ to $w$ in $G$, the resulting tuple
with updated memory $(w,\max(p,\Gamma(v)),\forall)$ is contained in
$X_{2k+2}$.
Using this notation, the winning region for the existential player in fair parity/parity games with priority functions $\lambda:V\to[2k]$ and
$\Gamma:V\to[d]$ can be characterized by the fixpoint expression induced by the function
$\chi$ that is defined to
map argument sets $X_1,\ldots, X_{2k+2+d}\subseteq V\times [d]\times[2]$
to the set
\begin{align*}
	\chi(X_1,\ldots, X_{2k+2+d})= &(\Vn \,\cap\, \mathsf{parity})\,\cup\\
	& (\Vfair_\exists\, \cap \, \bigcup_{i \in [ 2k]}\, C_i \,\cap \,\varphi^\fair_{\exists, i})\, \cup \\ &(\Vfair_\exists\, \cap \, \bigcup_{i \in [2k]}\, C_i \,\cap \,\varphi^\fair_{\forall, i})
\end{align*}
Then in full form, we get the following alternating nested fixpoint expression
\begin{align}\label{eq:mutuallyfairparityparityfp}
	e'=\eta X_{2k+2+d}.\, \bar{\eta} X_{2k+1+d}.\,\ldots \,\nu X_2.\,\mu X_1. \,
	\chi(X_1,\ldots, X_{2k+2+d})
\end{align}
where $\eta=\mu$ and $\bar{\eta}=\nu$ if $d$ is odd and vice versa if $d$ is even.

\begin{restatable}{theorem}{restatableparityparityfpcorrectness}\label{theorem:parityparityfpcorrectness}
	Let $G=(A,\text{Parity}(\lambda),\text{Parity}(\Gamma))$ where $A = (V_\exists,V_\forall,E,E_f)$ is a fair game arena, $V = V_\exists \dotcup V_\forall$ and  $\lambda: V \to [2k]$, $\Gamma: V \to [d]$ are priority functions. 
	Then there is a fixpoint expression over $V\times[d]\times[2]$
	with alternation depth $2k+2+d$ that characterizes $\Win_\exists(G)$.
\end{restatable}

\begin{linenomath}
	
	\begin{proof}[Proof of~\cref{theorem:parityparityfpcorrectness}]
		First we note that~\cref{eq:mutuallyfairparityparityfp} has the claimed parameters. 
		To see that~\cref{eq:mutuallyfairparityparityfp} characterizes $\mathsf{Win}_\exists(G)$, it suffices by~\cref{thm:fairParityPlusParityToParity} 
		to show that for all $v\in V$,
		$(v,1,\exists)\in e'$ if and only if $\exists$-player wins the node $(v,1,\exists)$
		in the game $G'$ from the proof of~\cref{thm:fairParityPlusParityToParity}.
		Again the proof is by mutual transformation of winning strategies in the
		game $G'$ from the proof of~\cref{thm:fairParityPlusParityToParity} on one hand, and in the semantic game $G_{e'}$ for the fixpoint expression
		$e'$ on the other hand. 
		The transformations are very similar to the ones in the proof
		of~\cref{theorem:parityfpcorrectness}, with main differences being in
		the updating of memory-values by means of the adapted functions $E$ and $E_f$,
		and the treatment of  branches in gadgets for nodes $(v,p,b)$ such that $v\in\Vfair_\exists$
		that now incorporate a case-distinction on the content of memory component $b$, as encoded by the 
		clause $\odot (X_{2k+1},\ldots X_{2k+d})$,
		and the treatment of  branches in gadgets for nodes $(v,p,b)$ such that $v\in\Vfair_\forall$ in which the memory component is flipped to $\forall$. as encoded by the clause $\Box' X_{2k+2}$.
	\end{proof}
\end{linenomath}

%% file: conclusion.tex
\section{Conclusion}

We introduce two-player games with local
transition-fairness constraints for both players, allowing two
objectives $\alpha$ and $\beta$ to decide the winner
of plays in which both players play fair and both players play unfair,
respectively. We show the determinacy of this class of games in the case
that $\alpha$ and $\beta$ are $\omega$-regular objectives.
In the special case that both $\alpha$ and $\beta$ are parity conditions,
there is a reduction to standard parity games with blow-up quadratic in the
number of priorities used by $\alpha$ and $\beta$; if $\beta=\top$
or $\beta=\bot$,
the reduction becomes even linear. We present
both enumerative and symbolic methods to realize this reduction; in the
process, we also obtain a tight exponential bound on the memory
required by winning strategies in fair parity/parity games.
We expect that the central idea behind the reduction generalizes from parity
objectives to more general settings such as fair games in which $\alpha$
and $\beta$ are Rabin, Streett, or
even Emerson-Lei conditions, but leave this issue for future work.

\section*{Acknowledgments}
We are grateful to an anonymous referee for contributing the proof of
the lower bound in Lemma~\ref{lem:strategysizes}.